\documentclass[pdflatex,sn-mathphys-num]{sn-jnl}
\usepackage{graphicx}
\usepackage{multirow}
\usepackage{amsmath,amssymb,amsfonts}
\usepackage{amsthm}
\usepackage{mathrsfs}
\usepackage[title]{appendix}
\usepackage{xcolor}
\usepackage{textcomp}
\usepackage{manyfoot}
\usepackage{booktabs}
\usepackage{algorithm}
\usepackage{algorithmicx}
\usepackage{algpseudocode}
\usepackage{listings}
\usepackage{subcaption}
\usepackage{array}
\usepackage{tabularx}
\theoremstyle{thmstyleone}
\newtheorem{theorem}{Theorem}
\newtheorem{proposition}[theorem]{Proposition}
\newtheorem{lemma}[theorem]{Lemma}
\newtheorem{corollary}[theorem]{Corollary}
\theoremstyle{thmstyletwo}

\newtheorem{remark}{Remark}
\theoremstyle{thmstylethree}
\newtheorem{definition}{Definition}
\begin{document}
	
	\title[Multiple Spatial Roots and Darboux Limits]{Real-Spectrum Darboux Limits at Multiple Spatial Roots of the Coupled Fokas--Lenells System}
	
	\author[1]{\fnm{Muchen} \sur{Dong}}
	
	\author*[1,]{\fnm{Lei} \sur{Wang}}\email{50901924@ncepu.edu.cn}
	
	\affil[1]{\orgdiv{School of Mathematics and Physics}, \orgname{North China Electric Power University}, \orgaddress{\city{Beijing}, \postcode{102206}, \state{Beijing}, \country{People’s Republic of China}}}
	
	\abstract{\unboldmath
		We develop a Darboux-based approach to multiple spatial branch points of the coupled Fokas--Lenells system on the plane-wave background. Unlike conventional generalized Darboux transformations based on the coalescence of distinct Darboux spectral poles, the degeneration is defined by the multiplicity of two or three roots of the single fixed-endpoint equation \(F(\chi;\mu_0)=0\) in the \(\chi\)-plane. For double roots, the branch-point condition reduces to a quartic polynomial, which is shown to possess at least two real roots, counted with multiplicity, when the squared background amplitudes are positive. This leads to a natural distinction between regular nonreal-spectrum branch points and singular real-spectrum limits. In the latter case, the same spectral endpoint is subsequently displaced as \(\mu=\mu_0+\eta\) solely to regularize the Darboux kernel; no collision of several distinct Darboux poles is involved. At a real double branch point, we characterize the complete projective \(J\)-null cone in the mixed double--simple eigenspace and derive the associated directional Darboux limit, including the first-order variation of the remaining simple branch. At a triple branch point, which is necessarily real in the parameter regime considered here, we construct the general full-chain limit and recover the lower derivative sectors as special cases. In both real-spectrum settings, the \(J\)-null condition removes the apparent singularity of the Darboux kernel and yields limiting denominators with strictly negative real parts. We also distinguish a persistent zero branch of the full polynomial spectral curve from an effective ramification point of the reduced rational spectral map. Finally, we identify two forms of nonlocalized far-field behavior. The complete nonzero pure double-root null-circle family approaches universal nonbackground plateaus along unbounded kernel level-set curves, whereas triple-root solutions with a nonzero second-derivative contribution possess an explicit nonbackground plateau along the characteristic line \(x+t/(3k^2)=0\). These results reveal how the spatial branch-point structure controls real-spectrum Darboux limits and their asymptotic behavior in the coupled Fokas--Lenells systems.}
	
	
	\maketitle
	
	\section{Introduction}
	\label{sec:introduction}
	
	Multicomponent integrable wave equations provide a natural framework for nonlinear interactions among polarization, spin, or carrier modes. Their canonical representative is the Manakov system \cite{Manakov1974},
	\begin{equation*}
		\mathrm{i}{\bf q}_t+{\bf q}_{xx}+2({\bf q}^{\dagger}{\bf q}){\bf q}=0,
		\qquad
		{\bf q}=(q_1,q_2)^{\mathsf T},
	\end{equation*}
	which is the integrable two-component extension of the focusing nonlinear Schr\"odinger equation. Although its nonlinear coupling is isotropic, the additional polarization degree of freedom produces spectral branches and coherent structures with no scalar counterpart. On nonzero backgrounds, the associated inverse-scattering problem already exhibits a genuinely multicomponent branch-point geometry \cite{KrausBiondiniKovacic2015}, while Darboux--dressing methods give a unified construction of bright, dark, and mixed vector waves \cite{DegasperisLombardo2007}. The Manakov system therefore serves both as a fundamental physical model and as a benchmark for understanding how vector degrees of freedom alter the spectral organization of nonlinear waves.
	
	Breathers and rogue waves provide particularly clear manifestations of this vector structure. Semirational Manakov solutions include vector Peregrine waves, bright--dark rogue waves, and mixed rational--soliton states \cite{BaronioEtAl2012}. Akhmediev breathers, Kuznetsov--Ma solitons, general breathers, and their rogue-wave limits have been obtained explicitly in the two-component setting \cite{PriyaSenthilvelanLakshmanan2013}. Higher-order vector rogue waves display decomposition patterns richer than those of the scalar nonlinear Schr\"odinger equation \cite{LingGuoZhao2014}. More recently, distinct nondegenerate Kuznetsov--Ma branches and their physical spectra have been identified \cite{CheEtAl2022}, and the vector dressing method has been used to classify fundamental Manakov breathers and their interactions, including resonant fusion and decay processes \cite{GelashRaskovalov2023}. These developments show that, even for the most classical vector integrable model, the multiplicity and geometry of spatial spectral branches are central to the classification of coherent structures.
	
	The Fokas--Lenells equation belongs to a different but closely related part of the integrable hierarchy. It was introduced from the viewpoints of bi-Hamiltonian integrability and nonlinear optical pulse propagation and may be regarded as an integrable nonlinear Schr\"odinger-type model incorporating higher-order dispersive and nonlinear corrections \cite{Fokas1995,Lenells2009,LenellsFokas2009}. For two interacting modes, a coupled form is
	\begin{equation*}
		\begin{aligned}
			{\rm i}D_\xi q_{1,\tau}
			-\frac{\eta}{2}q_{1,\xi\xi}
			+\left(2|q_1|^2+\sigma|q_2|^2\right)D_\xi q_1
			+\sigma q_1q_2^*D_\xi q_2
			&=0,\\
			{\rm i}D_\xi q_{2,\tau}
			-\frac{\eta}{2}q_{2,\xi\xi}
			+\left(2\sigma|q_2|^2+|q_1|^2\right)D_\xi q_2
			+q_2q_1^*D_\xi q_1
			&=0,
		\end{aligned}
	\end{equation*}
	where
	\begin{equation*}
		D_\xi=1+{\rm i}\nu\partial_\xi,
		\qquad
		\eta=\pm1,
		\qquad
		\sigma=\pm1.
	\end{equation*}
	In the symmetric case \(\sigma=1\), the limit in which the higher-order correction is removed reduces, after a simple scaling, to Manakov-type dynamics. For nonzero \(\nu\), the operator \(D_\xi\) introduces space--time coupling and self-steepening effects beyond the standard slowly varying envelope approximation. The model is consequently relevant to ultrashort-pulse propagation in birefringent fibers and to related multicomponent dispersive media \cite{ChenEtAl2018,YeEtAl2019,LingSu2024}.
	
	Under the coordinate and gauge transformation
	\begin{equation*}
		\xi=\nu\eta(x-t),
		\qquad
		\tau=-2\nu^2t,
		\qquad
		q_j(\xi,\tau)=\frac{{\rm i}}{2\nu}
		\exp\!\left[{\rm i}\eta(x+t)\right]u_j(x,t),
		\qquad j=1,2,
	\end{equation*}
	the coupled model takes the reduced form studied below. We restrict attention to the symmetric reduction \(\sigma=1\), corresponding to the coupled Fokas--Lenells model employed by Chen \emph{et al.}\ in their study of anomalous Peregrine solitons beyond the conventional threefold-amplification limit \cite{ChenEtAl2018}. Physically, this reduction retains the nonlinear interaction between two optical components together with the mixed space--time coupling relevant to ultrashort-pulse dynamics.
	
	The coupled Fokas--Lenells system has supported a broad range of exact-solution studies. Its connection with the vector Kaup--Newell hierarchy yields multi-Hamiltonian structures and infinitely many conservation laws, while generalized Darboux transformations generate bright, dark, anti-dark, breather-like, and rational solutions on zero and nonzero backgrounds \cite{LingFengZhu2018}. Solitons, breathers, and rogue waves have also been constructed through iterated Darboux transformations \cite{ZhangEtAl2017,YangZhang2018}. Anomalous Peregrine structures whose peaks exceed the conventional threefold background amplitude, together with their interactions and numerical excitation, were reported in \cite{ChenEtAl2018}. A nonrecursive Darboux representation subsequently produced general rogue-wave hierarchies involving both double and triple roots of the spatial characteristic equation \cite{YeEtAl2019}; the large-parameter decomposition of higher-order triple-root rogue waves has more recently been related to Okamoto polynomial hierarchies \cite{LingSu2024}. Thus multiple spatial roots are already known to be decisive in the rational dynamics of the coupled system. What has remained less understood is the local geometry of these roots at a fixed Darboux spectral value, especially when that value lies on the real spectrum.
	
	Darboux transformations constitute a standard algebraic method for generating exact solutions of integrable nonlinear equations from eigenfunctions of their associated linear spectral problems. Comprehensive accounts of the classical theory, its relation to soliton solutions, and its matrix formulation for integrable systems can be found in the monographs of Matveev and Salle and of Gu, Hu, and Zhou \cite{MatveevSalle1991,GuHuZhou2005}. For multicomponent integrable wave equations, Degasperis and Lombardo developed a general Darboux--dressing formalism and applied it to the construction of soliton solutions on vanishing and nonvanishing backgrounds \cite{DegasperisLombardo2007,DegasperisLombardo2009}. In a conventional generalized Darboux construction for higher-order solutions, several distinct Darboux spectral parameters coalesce at a common spectral value, and Taylor expansion with respect to the spectral parameter provides the derivative data required for the resulting confluent transformation. This procedure underlies, for example, the generalized Darboux formula and higher-order rogue waves of the nonlinear Schr"odinger equation \cite{GuoLingLiu2012}, as well as the nonrecursive rogue-wave construction for the coupled Fokas--Lenells system \cite{YeEtAl2019}. A complementary algebraic formulation based on Darboux matrices possessing a single higher-order pole has recently been developed in \cite{LiWang2025}.
	
	The degeneration investigated here is fundamentally different. A double or triple spatial branch point is defined at a single fixed squared spectral endpoint \(\mu_0\) by the multiplicity of the roots of \(F(\chi;\mu_0)=0\) in the \(\chi\)-plane. Thus the confluent chain is generated by the local structure of the spatial spectral branches at \(\mu_0\), rather than by a collision of distinct Darboux poles in the \(\lambda\)-plane. If \(\mu_0\notin\mathbb R\), the resulting branch coalescence can be inserted directly into the ordinary one-fold Darboux formula. If \(\mu_0>0\), however, the Darboux kernel is singular at \(\lambda_0=\sqrt{\mu_0}\in\mathbb R\). We then introduce the auxiliary displacement \(\mu=\mu_0+\eta\) only after the multiple root has been fixed, in order to regularize the real-spectrum kernel and take a directional limit. No family of distinct Darboux poles is made to coalesce. On the real spectrum, the kernel contains a quotient proportional to
	\begin{equation*}
		\frac{\langle {\bf y}|J|{\bf y}\rangle}{\lambda^*-\lambda},
	\end{equation*}
	and direct substitution of a nonzero real \(\lambda\) is singular. A finite nontrivial limit then requires a \(J\)-null leading confluent vector together with the first nonvanishing variation of both the eigenfunction and the spectral denominator. The resulting construction is directional: in general, it retains information about how \(\mu\) approaches the real axis.
	
	The main contributions of this work are as follows. First, we distinguish the complete polynomial spatial characteristic equation from its reduced rational parametrization. This reveals a persistent zero branch at the critical condition \(2-k(A-B)=0\), which may disappear under rational cancellation but remains part of the spectral problem. Second, we derive and classify the double- and triple-root conditions. In particular, the double-root equation is quartic and always possesses at least two real roots in the physical region; effective double roots are real precisely when their squared spectral values are real, while the triple-root configuration is proved, without assuming real spectral coefficients, to lie on the positive real spectrum. Third, we formulate a unified directional real-spectrum limit theorem and apply it to the projective \(J\)-null sets generated by the double- and triple-root splitting branches. The limiting denominators have strictly negative real parts, which yields global regularity for all admissible coefficient sectors treated below. Finally, we show that the resulting real-spectrum rational structures need not be uniformly localized on the plane-wave background: the two pure double-root plateau factors satisfy an exact reciprocal law along unbounded cubic kernel-level graphs, whereas the two triple-root plateau factors obey an exact weighted component constraint along a distinguished characteristic line.
	
	The paper is organized as follows. Section~\ref{sec:preliminaries} introduces the reduced model, its normalized Lax pair, the plane-wave spatial spectral curve, the elementary modes, the persistent zero branch, the one-fold Darboux transformation, and the unified directional real-spectrum limit theorem. Section~\ref{sec:spectral-geometry} classifies double and triple spatial roots and separates the regular nonreal and singular real spectral sectors. Section~\ref{sec:double-darboux} develops the Darboux constructions associated with double roots, including the general mixed double--simple real-spectrum limit and the explicit pure double-root family. Section~\ref{sec:triple-darboux} treats the full triple-chain limit and its lower coefficient sectors. Section~\ref{sec:nonlocal-asymptotics} analyzes the selected nonlocalized far-field structures, and Section~\ref{sec:examples} illustrates the spectral transitions and representative solutions. Section~\ref{sec:conclusion} concludes the paper.
	
	\section{Spectral and Darboux Preliminaries}
	\label{sec:preliminaries}
	
	\subsection{The plane-wave background and normalized Lax pair}
	\label{subsec:model-background}
	\label{subsec:normalized-lax}
	
	We consider the reduced two-component Fokas--Lenells system
	\begin{equation}
		\label{eq:CFL}
		\begin{aligned}
			u_{1,xt}+u_1+{\rm i}\left(|u_1|^2+\frac{1}{2}|u_2|^2\right)u_{1,x}+\frac{{\rm i}}{2}u_1u_2^*u_{2,x}&=0,\\
			u_{2,xt}+u_2+{\rm i}\left(|u_2|^2+\frac{1}{2}|u_1|^2\right)u_{2,x}+\frac{{\rm i}}{2}u_2u_1^*u_{1,x}&=0,
		\end{aligned}
	\end{equation}
	on the plane-wave background
	\begin{equation}
		\label{eq:background}
		u_{1,0}=a_1\exp\left[{\rm i}(k_1x+\omega_1t)\right],\qquad
		u_{2,0}=a_2\exp\left[{\rm i}(k_2x+\omega_2t)\right].
	\end{equation}
	We are interested in backgrounds with unequal carrier wave numbers and, for simplicity, set \(k_1=-k_2=k\). Substitution of \eqref{eq:background} into \eqref{eq:CFL} gives \(\omega_1=1/k-A\) and \(\omega_2=-1/k-B\), where \(A=a_1^2\) and \(B=a_2^2\). Throughout the paper, we assume
	\begin{equation}
		\label{eq:physical-region}
		A>0,\qquad B>0,\qquad k\neq0.
	\end{equation}
	For later use, we set \(S=A+B\) and \(R=A-B\).
	
	Let
	\[
	{\bf u}=
	\begin{pmatrix}
		u_1\\
		u_2
	\end{pmatrix},
	\qquad
	P=\operatorname{diag}(1,0,0),
	\qquad
	\Phi=
	\begin{pmatrix}
		\psi\\
		\phi\\
		\varphi
	\end{pmatrix}.
	\]
	A normalized Lax representation of \eqref{eq:CFL} is
	\begin{equation}
		\label{eq:laxpair}
		\Phi_x=U(\lambda)\Phi,\qquad \Phi_t=V(\lambda)\Phi,
	\end{equation}
	where
	\begin{equation}
		\label{eq:U}
		U(\lambda)=\frac{2{\rm i}}{\lambda^2}P+\frac{1}{\lambda}
		\begin{pmatrix}
			0 & {\bf u}_x^\dagger\\
			{\bf u}_x & 0_{2\times2}
		\end{pmatrix},
	\end{equation}
	and
	\begin{equation}
		\label{eq:V}
		V(\lambda)=\frac{{\rm i}\lambda^2}{2}P+\frac{{\rm i}\lambda}{2}
		\begin{pmatrix}
			0 & -{\bf u}^\dagger\\
			{\bf u} & 0_{2\times2}
		\end{pmatrix}
		+\frac{{\rm i}}{2}
		\begin{pmatrix}
			{\bf u}^\dagger{\bf u} & 0\\
			0 & -{\bf u}{\bf u}^\dagger
		\end{pmatrix}.
	\end{equation}
	Here, \(\dagger\) denotes Hermitian conjugation. The compatibility condition \(U_t-V_x+[U,V]=0\) is equivalent to system~\eqref{eq:CFL}.
	
	\begin{remark}
		\label{rem:scalar-gauge}
		Let
		\[
		\mathcal J=\operatorname{diag}(1,-1,-1),\qquad
		Q=
		\begin{pmatrix}
			0 & u_1^* & u_2^*\\
			u_1 & 0 & 0\\
			u_2 & 0 & 0
		\end{pmatrix}.
		\]
		A standard Lax representation of \eqref{eq:CFL} is \(\Psi_x=\widehat U(\lambda)\Psi\) and \(\Psi_t=\widehat V(\lambda)\Psi\), where
		\[
		\widehat U(\lambda)=\frac{{\rm i}}{\lambda^2}\mathcal J+\frac{1}{\lambda}Q_x,\qquad
		\widehat V(\lambda)=\frac{{\rm i}\lambda^2}{4}\mathcal J+\frac{{\rm i}\lambda}{2}Q\mathcal J+\frac{{\rm i}}{2}\mathcal JQ^2.
		\]
		Because \(P=(I+\mathcal J)/2\), the matrices in \eqref{eq:U}--\eqref{eq:V} satisfy
		\[
		U(\lambda)=\widehat U(\lambda)+\frac{{\rm i}}{\lambda^2}I,\qquad
		V(\lambda)=\widehat V(\lambda)+\frac{{\rm i}\lambda^2}{4}I.
		\]
		The associated eigenfunctions are therefore related by
		\[
		\Phi(x,t;\lambda)=\exp\left(\frac{{\rm i}x}{\lambda^2}+\frac{{\rm i}\lambda^2t}{4}\right)\Psi(x,t;\lambda).
		\]
		Thus the two representations differ only by a scalar gauge factor and produce the same zero-curvature equation. We use the normalized form \eqref{eq:laxpair} throughout.
	\end{remark}
	
	\subsection{The spatial spectral curve and elementary modes}
	\label{subsec:background-spectral-curve}
	
	Set \(\mu=\lambda^2\). Substitution of the plane-wave background \eqref{eq:background} into the spatial Lax equation yields the meromorphic spectral map
	\begin{equation}
		\label{eq:Lambda}
		\mu=\Lambda(\chi)=\frac{2-k^2\left(\dfrac{A}{\chi+k}+\dfrac{B}{\chi-k}\right)}{\chi},\qquad
		\chi\notin\{0,\pm k\}.
	\end{equation}
	Here, \(\mu\) is the squared spectral parameter, whereas \(\Lambda\) denotes the map from the spatial spectral variable \(\chi\) to \(\mu\). Writing
	\begin{equation}
		\label{eq:Lambda-ND}
		\Lambda(\chi)=\frac{N(\chi)}{D(\chi)},
	\end{equation}
	we have \(N(\chi)=2\chi^2-k^2S\chi-2k^2+k^3R\) and \(D(\chi)=\chi(\chi^2-k^2)\).
	
	For a fixed value of \(\mu\), the spatial eigenvalues are determined by the polynomial equation
	\begin{equation}
		\label{eq:F}
		F(\chi;\mu)=D(\chi)\bigl[\mu-\Lambda(\chi)\bigr]=\mu\chi^3-2\chi^2+k^2(S-\mu)\chi+k^2(2-kR)=0.
	\end{equation}
	The two forms play complementary roles: derivatives of \(\Lambda\) describe the local ramification geometry, while \(F\) retains the complete fixed-\(\mu\) spatial characteristic equation. This distinction becomes essential when \(N\) and \(D\) have a common factor, since cancellation in the rational representation may remove an entire spectral branch.
	
	The possible singular points \(\chi=\pm k\) are excluded by the following elementary observation.
	
	\begin{proposition}
		\label{prop:no-pm-k-root}
		In the physical regime \eqref{eq:physical-region}, neither \(\chi=k\) nor \(\chi=-k\) is a root of the fixed-\(\mu\) spatial equation. Consequently, no double or triple branch point can occur at \(\chi=\pm k\).
	\end{proposition}
	
	\begin{proof}
		Direct substitution into \eqref{eq:F} gives \(F(k;\mu)=k^3(S-R)=2Bk^3\) and \(F(-k;\mu)=-k^3(S+R)=-2Ak^3\). Both quantities are nonzero for every \(\mu\) because \(A>0\), \(B>0\), and \(k\neq0\). Hence \(\chi=\pm k\) cannot be spatial spectral roots or multiple roots.
	\end{proof}
	
	It follows that the factors \((k+\chi)^{-1}\) and \((k-\chi)^{-1}\), which occur in \eqref{eq:Lambda} and in the elementary eigenfunctions below, are regular at every admissible spatial root. The point \(\chi=0\), however, requires separate consideration.
	
	\begin{remark}
		\label{rem:critical-zero-branch}
		Since \(F(0;\mu)=k^2(2-kR)\), the point \(\chi=0\) is a spatial spectral root precisely when
		\begin{equation}
			\label{eq:critical-zero-condition}
			2-kR=0.
		\end{equation}
		Under this condition,
		\begin{equation}
			\label{eq:F-critical-zero-factorization}
			F(\chi;\mu)=\chi\left[\mu\chi^2-2\chi+k^2(S-\mu)\right].
		\end{equation}
		Thus \(\chi=0\) is a genuine spatial eigenvalue branch for every \(\mu\), although it is removed from the reduced rational relation \(\mu=\Lambda(\chi)\) by cancellation of the common factor \(\chi\) in \(N(\chi)\) and \(D(\chi)\). At \(\mu=S\), this branch intersects a movable branch at \(\chi=0\), and \(F(\chi;S)=\chi^2(S\chi-2)\). Since the resulting double zero is a branch intersection rather than a ramification point of \(\Lambda\), it is excluded from the effective branch-point classification and treated separately below.
	\end{remark}
	
	For a nonzero spatial root \(\chi\), define
	\[
	E(\chi)=\exp\left[{\rm i}\bigl(\chi x+\nu(\chi)t\bigr)\right],\qquad
	\nu(\chi)=\frac{S}{2}+\frac{2-kR}{2\chi},
	\]
	and set \(h_1(\chi)=k/(k+\chi)\) and \(h_2(\chi)=k/(k-\chi)\). An elementary eigenfunction associated with the spectral pair \((\chi,\lambda)\), where \(\mu=\lambda^2\), is
	\begin{equation}
		\label{eq:y-basic}
		{\bf y}(\chi,\lambda)=
		\begin{pmatrix}
			\lambda E(\chi)\\[1mm]
			u_{1,0}h_1(\chi)E(\chi)\\[1mm]
			u_{2,0}h_2(\chi)E(\chi)
		\end{pmatrix}.
	\end{equation}
	
	On the persistent zero branch of Remark~\ref{rem:critical-zero-branch}, the expression for \(\nu(\chi)\) is indeterminate and cannot be evaluated by setting \(\chi=0\). The temporal Lax equation instead gives \(\nu_s(\mu)=\mu/2\). Hence the corresponding elementary eigenfunction is
	\begin{equation}
		\label{eq:zero-branch-eigenfunction}
		{\bf y}_s(\lambda)=\exp\left(\frac{{\rm i}\mu t}{2}\right)
		\begin{pmatrix}
			\lambda\\
			u_{1,0}\\
			u_{2,0}
		\end{pmatrix},\qquad
		\mu=\lambda^2.
	\end{equation}
	In particular, \(\partial_\mu\nu_s=1/2\); this derivative will enter the first-order expansion whenever the persistent branch participates in a real-spectrum limiting construction.
	
	The derivative modes used below are fixed-\(\lambda_0\) derivatives. The following lemma shows directly, without an off-shell residual factorization, why the derivatives up to the multiplicity order solve the same Lax pair.
	
	\begin{lemma}
		\label{lem:fixed-spectral-derivative-chain}
		Let \(m\ge2\), let \(\chi_0\notin\{0,\pm k\}\), and suppose that
		\begin{equation}
			\label{eq:mfold-Lambda-vanishing}
			\Lambda(\chi_0)=\mu_0=\lambda_0^2\ne0,
			\qquad
			\Lambda^{(j)}(\chi_0)=0,
			\quad j=1,\ldots,m-1.
		\end{equation}
		For \(0\le n\le m-1\), define
		\begin{equation}
			\label{eq:fixed-lambda-derivative-modes}
			{\bf Y}_n
			=
			\left.
			\partial_\chi^n{\bf y}(\chi,\lambda_0)
			\right|_{\chi=\chi_0},
		\end{equation}
		where \(\lambda\) is held fixed during differentiation. Then every \({\bf Y}_n\), \(0\le n\le m-1\), satisfies both equations of the background Lax pair at the single spectral point \(\lambda=\lambda_0\).
	\end{lemma}
	
	\begin{proof}
		Choose a simply connected neighborhood \(U\) of \(\chi_0\) which does not meet \(\{0,\pm k\}\) and on which \(\Lambda\) does not vanish. There is a holomorphic branch
		\[
		\lambda(\chi)=\sqrt{\Lambda(\chi)},
		\qquad
		\lambda(\chi_0)=\lambda_0.
		\]
		Differentiating \(\lambda(\chi)^2=\Lambda(\chi)\) once gives
		\(2\lambda_0\lambda'(\chi_0)=\Lambda'(\chi_0)=0\). Inductively, assume that
		\(\lambda^{(j)}(\chi_0)=0\) for \(1\le j<r\le m-1\). The \(r\)-th derivative of \(\lambda^2=\Lambda\) at \(\chi_0\) is
		\[
		2\lambda_0\lambda^{(r)}(\chi_0)
		+
		\sum_{j=1}^{r-1}
		\binom{r}{j}
		\lambda^{(j)}(\chi_0)
		\lambda^{(r-j)}(\chi_0)
		=
		\Lambda^{(r)}(\chi_0).
		\]
		Both the sum and the right-hand side vanish, while \(\lambda_0\ne0\). Hence
		\begin{equation}
			\label{eq:lambda-derivatives-vanish}
			\lambda^{(j)}(\chi_0)=0,
			\qquad j=1,\ldots,m-1.
		\end{equation}
		
		Let
		\[
		\mathcal L_x(\lambda)=\partial_x-U(u_0,\lambda),
		\qquad
		\mathcal L_t(\lambda)=\partial_t-V(u_0,\lambda).
		\]
		For every \(\chi\in U\), the pair \((\chi,\lambda(\chi))\) lies on the spatial spectral curve. Consequently,
		\begin{equation}
			\label{eq:on-curve-Lax-identities}
			\mathcal L_x(\lambda(\chi))
			{\bf y}(\chi,\lambda(\chi))=0,
			\qquad
			\mathcal L_t(\lambda(\chi))
			{\bf y}(\chi,\lambda(\chi))=0.
		\end{equation}
		Fix \(n\le m-1\), differentiate either identity in
		\eqref{eq:on-curve-Lax-identities} \(n\) times with respect to \(\chi\), and evaluate at \(\chi_0\). Every term in which a derivative hits the composed spectral parameter contains at least one factor \(\lambda^{(j)}(\chi_0)\) with \(1\le j\le n\), and therefore vanishes by \eqref{eq:lambda-derivatives-vanish}. Thus
		\[
		\mathcal L_x(\lambda_0)
		\left.
		\frac{d^n}{d\chi^n}
		{\bf y}(\chi,\lambda(\chi))
		\right|_{\chi=\chi_0}=0,
		\]
		and the analogous identity holds for \(\mathcal L_t\). The chain rule and \eqref{eq:lambda-derivatives-vanish} also give
		\[
		\left.
		\frac{d^n}{d\chi^n}
		{\bf y}(\chi,\lambda(\chi))
		\right|_{\chi=\chi_0}
		=
		\left.
		\partial_\chi^n
		{\bf y}(\chi,\lambda_0)
		\right|_{\chi=\chi_0}
		={\bf Y}_n.
		\]
		Hence \(\mathcal L_x(\lambda_0){\bf Y}_n=\mathcal L_t(\lambda_0){\bf Y}_n=0\) for \(0\le n\le m-1\).
	\end{proof}
	
	\begin{proposition}
		\label{prop:confluent-basis-independence}
		The fixed-\(\lambda_0\) confluent modes have the following linear-independence properties.
		\begin{enumerate}
			\item If \(\chi_0\ne0\) is an exact double root of \(F(\chi;\mu_0)\) and \(\chi_s\ne\chi_0\) is the remaining simple root, then \({\bf Y}_0,{\bf Y}_1,{\bf Y}_s\) are linearly independent. Here \({\bf Y}_s\) denotes either \({\bf y}(\chi_s,\lambda_0)\) or, when \(\chi_s=0\) is the persistent branch, the eigenfunction \eqref{eq:zero-branch-eigenfunction} evaluated at \(\lambda_0\).
			\item If \(\chi_0\ne0\) is an exact triple root of \(F(\chi;\mu_0)\), then \({\bf Y}_0,{\bf Y}_1,{\bf Y}_2\) are linearly independent.
		\end{enumerate}
		In either case the three vectors form a fundamental solution matrix of the background Lax pair at \(\lambda=\lambda_0\).
	\end{proposition}
	
	\begin{proof}
		Fix an arbitrary \(t=t_*\). Common nonzero factors depending only on \(t_*\) do not affect linear independence.
		
		In the double-root case, the first components of \({\bf Y}_0\), \({\bf Y}_1\), and \({\bf Y}_s\) are, up to nonzero constant factors, of the form
		\[
		e^{{\rm i}\chi_0x},
		\qquad
		(x+c_*)e^{{\rm i}\chi_0x},
		\qquad
		e^{{\rm i}\chi_sx},
		\]
		where \(c_*\) is independent of \(x\). This remains true for the persistent zero branch, for which \(e^{{\rm i}\chi_sx}=1\). Suppose that a linear combination of the three vector solutions vanishes. Its first component has the form
		\[
		(a+b x)e^{{\rm i}\chi_0x}
		+c e^{{\rm i}\chi_sx}=0
		\qquad\text{for all }x.
		\]
		Applying \((\partial_x-{\rm i}\chi_0)^2\) gives
		\[
		-c(\chi_s-\chi_0)^2e^{{\rm i}\chi_sx}=0.
		\]
		Since \(\chi_s\ne\chi_0\), one has \(c=0\). The remaining identity \(a+bx=0\) then gives \(a=b=0\). Hence \({\bf Y}_0,{\bf Y}_1,{\bf Y}_s\) are linearly independent.
		
		In the triple-root case, the first components of \({\bf Y}_0,{\bf Y}_1,{\bf Y}_2\), after division by the common nonzero factor \(\lambda_0e^{{\rm i}(\chi_0x+\nu(\chi_0)t_*)}\), are polynomials in \(x\) of degrees respectively \(0,1,2\), with nonzero leading coefficients. A vanishing linear combination is therefore a polynomial of degree at most two which vanishes for every \(x\), so all three coefficients are zero. Thus \({\bf Y}_0,{\bf Y}_1,{\bf Y}_2\) are linearly independent.
		
		The spatial Lax equation is a first-order \(3\times3\) linear system. Lemma~\ref{lem:fixed-spectral-derivative-chain} shows that the relevant derivative modes solve it at \(\lambda_0\), while the simple mode in the double-root case is an elementary solution at the same spectral point. Three linearly independent solutions therefore form a fundamental solution matrix. Their simultaneous satisfaction of the temporal Lax equation follows from the same lemma and from the elementary-mode construction.
	\end{proof}
	
	The following elementary recurrence estimate will be used to control complete Taylor tails. It is needed because the individual split-branch coefficients may diverge even though their weighted moments remain finite.
	
	\begin{lemma}
		\label{lem:analytic-moment-tail}
		Let \(\mathcal B\) be a complex Banach space and let
		\[
		f(\delta)=\sum_{n=0}^{\infty}f_n\delta^n
		\]
		be holomorphic as a \(\mathcal B\)-valued function on \(|\delta|<r_0\). Let \(\mathfrak M_n(\eta)\) satisfy, for some integers \(N\ge0\) and \(d\ge1\),
		\begin{equation}
			\label{eq:abstract-moment-recurrence}
			\mathfrak M_n
			=
			\sum_{\ell=1}^{d}
			b_\ell(\eta)\mathfrak M_{n-\ell},
			\qquad n\ge N+d,
		\end{equation}
		where \(b_\ell(\eta)=O(|\eta|)\), and suppose that
		\[
		\mathfrak M_N,\ldots,\mathfrak M_{N+d-1}
		=O(|\eta|^2).
		\]
		Then, for every \(\tau\in(0,r_0)\), after reducing \(|\eta|\) if necessary, there is a constant \(C_\tau\), independent of \(n\) and \(\eta\), such that
		\begin{equation}
			\label{eq:uniform-moment-geometric-bound}
			|\mathfrak M_n(\eta)|
			\le
			C_\tau|\eta|^2\tau^{n-N},
			\qquad n\ge N.
		\end{equation}
		Consequently,
		\begin{equation}
			\label{eq:Banach-analytic-tail}
			\left\|
			\sum_{n=N}^{\infty}
			\mathfrak M_n(\eta)f_n
			\right\|_{\mathcal B}
			=O(|\eta|^2).
		\end{equation}
	\end{lemma}
	
	\begin{proof}
		Fix \(\tau\in(0,r_0)\). Since every \(b_\ell(\eta)=O(|\eta|)\), there exists \(\eta_0>0\) such that
		\[
		\sum_{\ell=1}^{d}
		|b_\ell(\eta)|\tau^{-\ell}
		\le1,
		\qquad 0<|\eta|<\eta_0.
		\]
		Enlarge \(C_\tau\), if necessary, so that \eqref{eq:uniform-moment-geometric-bound} holds for the \(d\) initial indices \(N,\ldots,N+d-1\). If it holds up to index \(n-1\), then \eqref{eq:abstract-moment-recurrence} yields
		\[
		\begin{aligned}
			|\mathfrak M_n|
			&\le
			\sum_{\ell=1}^{d}
			|b_\ell(\eta)|
			C_\tau|\eta|^2
			\tau^{n-\ell-N}
			\\
			&\le
			C_\tau|\eta|^2\tau^{n-N}.
		\end{aligned}
		\]
		Induction proves \eqref{eq:uniform-moment-geometric-bound} for all \(n\ge N\).
		
		Choose \(r\) with \(\tau<r<r_0\). The Banach-valued Cauchy estimate gives
		\[
		\|f_n\|_{\mathcal B}
		\le
		M_r r^{-n},
		\qquad
		M_r=
		\sup_{|\delta|=r}\|f(\delta)\|_{\mathcal B}<\infty.
		\]
		Therefore,
		\[
		\begin{aligned}
			\left\|
			\sum_{n=N}^{\infty}
			\mathfrak M_n f_n
			\right\|_{\mathcal B}
			&\le
			C_\tau M_r|\eta|^2
			\sum_{n=N}^{\infty}
			\tau^{n-N}r^{-n}
			\\
			&=
			\frac{C_\tau M_r r^{-N}}
			{1-\tau/r}
			|\eta|^2.
		\end{aligned}
		\]
		This proves \eqref{eq:Banach-analytic-tail}.
	\end{proof}
	
	\subsection{The one-fold Darboux transformation and the \(J\)-form}
	\label{subsec:onefold-Jform}
	
	For later use, set \(J=\operatorname{diag}(1,-1,-1)\) and \(K=I_3\), and define
	\begin{equation*}
		\begin{aligned}
			\langle{\bf y}_L|J|{\bf y}_R\rangle&=\psi_L^*\psi_R-\phi_L^*\phi_R-\varphi_L^*\varphi_R,\\
			\langle{\bf y}_L|K|{\bf y}_R\rangle&=\psi_L^*\psi_R+\phi_L^*\phi_R+\varphi_L^*\varphi_R.
		\end{aligned}
	\end{equation*}
	The \(J\)-form determines the admissibility of the leading confluent eigenvector in the real-spectrum limit, whereas the positive-definite \(K\)-form contributes to the finite part of the limiting Darboux kernel. We first record the conservation and inertia properties of the \(J\)-form.
	
	\begin{lemma}
		\label{lem:J-form-conservation}
		Let \(\lambda_0\in\mathbb R\setminus\{0\}\), and let \({\bf f}\) and \({\bf g}\) be two solutions of the same Lax pair \eqref{eq:laxpair} at \(\lambda=\lambda_0\). Then
		\begin{equation}
			\label{eq:J-pairing-conservation}
			\partial_x\langle{\bf f}|J|{\bf g}\rangle=0,\qquad
			\partial_t\langle{\bf f}|J|{\bf g}\rangle=0.
		\end{equation}
		If \(\mathcal Y=({\bf y}^{(1)},{\bf y}^{(2)},{\bf y}^{(3)})\) is a fundamental matrix formed by three linearly independent solutions at \(\lambda=\lambda_0\), then the Gram matrix \(G_{\mathcal Y}=\mathcal Y^\dagger J\mathcal Y\) is nonsingular and has inertia \(\operatorname{In}G_{\mathcal Y}=(1,2)\), namely one positive and two negative eigenvalues.
	\end{lemma}
	
	\begin{proof}
		For real \(\lambda_0\), direct inspection of \eqref{eq:U} and \eqref{eq:V} gives
		\[
		U(\lambda_0)^\dagger J+JU(\lambda_0)=0,\qquad
		V(\lambda_0)^\dagger J+JV(\lambda_0)=0.
		\]
		Consequently,
		\[
		\partial_x\langle{\bf f}|J|{\bf g}\rangle={\bf f}^\dagger\bigl(U(\lambda_0)^\dagger J+JU(\lambda_0)\bigr){\bf g}=0,
		\]
		and the corresponding identity involving \(V(\lambda_0)\) gives \(\partial_t\langle{\bf f}|J|{\bf g}\rangle=0\). This proves \eqref{eq:J-pairing-conservation}.
		
		Since \(\mathcal Y\) is fundamental, \(\det\mathcal Y\neq0\). Thus \(G_{\mathcal Y}=\mathcal Y^\dagger J\mathcal Y\) is congruent to \(J\), and Sylvester's law of inertia yields \(\operatorname{In}G_{\mathcal Y}=\operatorname{In}J=(1,2)\).
	\end{proof}
	
	We next recall the one-fold Darboux transformation used throughout the paper.
	
	\begin{proposition}
		\label{prop:one-fold-DT}
		Let \({\bf u}=(u_1,u_2)^{\mathsf T}\) be a solution of \eqref{eq:CFL}, and let \({\bf y}_1=(\psi_1,\phi_1,\varphi_1)^{\mathsf T}\neq{\bf 0}\) solve the Lax pair \eqref{eq:laxpair} at \(\lambda=\lambda_1\), where \(\lambda_1\in\mathbb C\setminus(\mathbb R\cup{\rm i}\mathbb R)\). Define
		\begin{equation}
			\label{eq:m-nonreal}
			m_1=\frac{\langle{\bf y}_1|J|{\bf y}_1\rangle}{\lambda_1^*-\lambda_1}-\frac{\langle{\bf y}_1|K|{\bf y}_1\rangle}{\lambda_1^*+\lambda_1}.
		\end{equation}
		On every domain where \(m_1\neq0\), the transformed fields
		\begin{equation}
			\label{eq:nonreal-DT}
			\begin{pmatrix}
				u_1^{[1]}\\
				u_2^{[1]}
			\end{pmatrix}
			=
			\begin{pmatrix}
				u_1\\
				u_2
			\end{pmatrix}
			+2\frac{\begin{pmatrix}\phi_1\\\varphi_1\end{pmatrix}\psi_1^*}{m_1}
		\end{equation}
		satisfy the coupled Fokas--Lenells system \eqref{eq:CFL}. The transformation is invariant under the rescaling \({\bf y}_1\mapsto c{\bf y}_1\), where \(c\in\mathbb C\setminus\{0\}\).
	\end{proposition}
	
	The generalized Darboux transformation from which Proposition~\ref{prop:one-fold-DT} follows was established by Ling, Feng, and Zhu in \cite{LingFengZhu2018}. Formula \eqref{eq:nonreal-DT} is its one-fold specialization under the coupled Fokas--Lenells reduction. The normalized Lax pair \eqref{eq:laxpair} differs from the representation used in that work only by scalar multiples of the identity, as described in Remark~\ref{rem:scalar-gauge}. These scalar gauge factors multiply the eigenfunction by a nonzero scalar and cancel from the quotient in \eqref{eq:nonreal-DT}. We therefore use Proposition~\ref{prop:one-fold-DT} without repeating its standard proof.
	
	The real-spectrum limits at double and triple spatial roots share the same analytic mechanism. We isolate it once so that the later sections need only construct the leading and first correction vectors.
	
	\begin{theorem}
		\label{thm:unified-real-spectrum-directional-limit}
		Let \({\bf u}=(u_1,u_2)^{\mathsf T}\) be a smooth solution of \eqref{eq:CFL}, let \(\lambda_0>0\), and let \(\lambda(\eta)\) be the local holomorphic branch determined by
		\[
		\lambda(\eta)^2=\lambda_0^2+\eta,
		\qquad
		\lambda(0)=\lambda_0.
		\]
		Suppose that \({\bf y}(\eta)\) is a family of Lax eigenfunctions at \(\lambda=\lambda(\eta)\), defined for sufficiently small nonzero \(\eta\) along every nonreal direction under consideration, and that
		\begin{equation}
			\label{eq:unified-eigenfunction-expansion}
			{\bf y}(\eta)
			=
			{\bf y}_R
			+
			\eta{\bf y}_\mu
			+
			O_{C^2_{\rm loc}}(|\eta|^2).
		\end{equation}
		Assume that \({\bf y}_R\) is a nontrivial solution of the Lax pair at \(\lambda=\lambda_0\) and is \(J\)-null:
		\begin{equation}
			\label{eq:unified-leading-nullity}
			\langle{\bf y}_R|J|{\bf y}_R\rangle=0.
		\end{equation}
		For \(q\in\mathbb C\setminus\mathbb R\), define
		\begin{equation}
			\label{eq:unified-directional-denominator}
			\begin{aligned}
				\mathfrak m_q({\bf y}_R,{\bf y}_\mu)
				={}&
				\frac{2\lambda_0}{q^*-q}
				\left[
				q^*\langle{\bf y}_\mu|J|{\bf y}_R\rangle
				+
				q\langle{\bf y}_R|J|{\bf y}_\mu\rangle
				\right]
				\\
				&-
				\frac{
					\langle{\bf y}_R|K|{\bf y}_R\rangle
				}{2\lambda_0}.
			\end{aligned}
		\end{equation}
		Writing \({\bf y}_R=(\psi_R,\phi_R,\varphi_R)^{\mathsf T}\), set
		\begin{equation}
			\label{eq:unified-directional-field}
			{\bf u}^{(R;q)}
			=
			{\bf u}
			+
			2
			\frac{
				\begin{pmatrix}
					\phi_R\\
					\varphi_R
				\end{pmatrix}
				\psi_R^*
			}{
				\mathfrak m_q({\bf y}_R,{\bf y}_\mu)
			}.
		\end{equation}
		Then the ordinary Darboux denominator \(m(\eta)\) from \eqref{eq:m-nonreal} satisfies, along \(\eta=\varepsilon q\),
		\begin{equation}
			\label{eq:unified-denominator-convergence}
			m(\varepsilon q)
			=
			\mathfrak m_q({\bf y}_R,{\bf y}_\mu)
			+
			O_{C^2_{\rm loc}}(\varepsilon),
			\qquad
			\varepsilon\in\mathbb R,
			\quad \varepsilon\to0.
		\end{equation}
		Moreover,
		\begin{equation}
			\label{eq:unified-negative-real-part}
			\operatorname{Re}
			\mathfrak m_q({\bf y}_R,{\bf y}_\mu)
			=
			-
			\frac{\|{\bf y}_R\|^2}{2\lambda_0}
			<0
		\end{equation}
		at every real \((x,t)\). Consequently, \eqref{eq:unified-directional-field} is globally smooth and is an exact solution of \eqref{eq:CFL}.
	\end{theorem}
	
	\begin{proof}
		The local square-root expansion is
		\[
		\lambda(\eta)
		=
		\lambda_0+
		\frac{\eta}{2\lambda_0}
		+O(|\eta|^2).
		\]
		Along \(\eta=\varepsilon q\), this gives
		\begin{equation}
			\label{eq:unified-lambda-denominator-expansions}
			\lambda(\varepsilon q)^*-
			\lambda(\varepsilon q)
			=
			\frac{\varepsilon(q^*-q)}{2\lambda_0}
			+O(\varepsilon^2),
			\qquad
			\lambda(\varepsilon q)^*+
			\lambda(\varepsilon q)
			=
			2\lambda_0+O(\varepsilon).
		\end{equation}
		Because \(q\notin\mathbb R\), the first denominator in \eqref{eq:unified-lambda-denominator-expansions} is nonzero for all sufficiently small nonzero \(\varepsilon\), and \(\lambda(\varepsilon q)\notin\mathbb R\cup{\rm i}\mathbb R\).
		
		Using \eqref{eq:unified-eigenfunction-expansion}, the Hermitian property of \(J\), and the nullity \eqref{eq:unified-leading-nullity}, we obtain
		\[
		\begin{aligned}
			\langle{\bf y}(\varepsilon q)|J|{\bf y}(\varepsilon q)\rangle
			={}&
			\varepsilon
			\left[
			q^*\langle{\bf y}_\mu|J|{\bf y}_R\rangle
			+
			q\langle{\bf y}_R|J|{\bf y}_\mu\rangle
			\right]
			+
			O_{C^2_{\rm loc}}(\varepsilon^2),
			\\
			\langle{\bf y}(\varepsilon q)|K|{\bf y}(\varepsilon q)\rangle
			={}&
			\langle{\bf y}_R|K|{\bf y}_R\rangle
			+
			O_{C^2_{\rm loc}}(\varepsilon).
		\end{aligned}
		\]
		Substitution into \eqref{eq:m-nonreal}, together with \eqref{eq:unified-lambda-denominator-expansions}, proves \eqref{eq:unified-denominator-convergence} and identifies its limit as \eqref{eq:unified-directional-denominator}.
		
		Set
		\[
		z=\langle{\bf y}_R|J|{\bf y}_\mu\rangle.
		\]
		Then
		\[
		q^*\langle{\bf y}_\mu|J|{\bf y}_R\rangle
		+
		q\langle{\bf y}_R|J|{\bf y}_\mu\rangle
		=
		q^*z^*+qz
		=
		2\operatorname{Re}(qz)
		\in\mathbb R,
		\]
		whereas \(q^*-q\) is nonzero and purely imaginary. Hence the first term in \eqref{eq:unified-directional-denominator} is purely imaginary. Since \(K=I_3\), the real part is
		\[
		\operatorname{Re}\mathfrak m_q
		=-\frac{\|{\bf y}_R\|^2}{2\lambda_0}.
		\]
		The nontrivial Lax solution \({\bf y}_R\) cannot vanish at any point: if \({\bf y}_R(x_0,t_0)=0\), uniqueness for the spatial and temporal linear systems would imply \({\bf y}_R\equiv0\). Thus \(\|{\bf y}_R(x,t)\|^2>0\), proving \eqref{eq:unified-negative-real-part} and the global nonvanishing of \(\mathfrak m_q\).
		
		For each sufficiently small nonzero \(\varepsilon\), Proposition~\ref{prop:one-fold-DT} produces an exact Darboux field from \({\bf y}(\varepsilon q)\). On every compact set, \eqref{eq:unified-negative-real-part} gives a positive lower bound for \(|\mathfrak m_q|\); hence \eqref{eq:unified-eigenfunction-expansion} and \eqref{eq:unified-denominator-convergence} imply local \(C^2\)-convergence of these ordinary Darboux fields to \eqref{eq:unified-directional-field}. Passing to the limit in \eqref{eq:CFL} proves that \({\bf u}^{(R;q)}\) is an exact solution.
	\end{proof}
	
	\section{Geometry and Classification of Multiple Spatial Roots}
	\label{sec:spectral-geometry}
	\label{sec:double-classification}
	
	This section describes the algebraic geometry of the fixed-\(\mu\) spatial equation before any singular Darboux limit is introduced. We first determine the possible distributions of double roots, then distinguish the real and nonreal sectors, and finally characterize the triple-root degeneration.
	
	\subsection{Double-root equation and root distribution}
	\label{subsec:double-root-classification}
	
	Recall that \(\Lambda(\chi)=N(\chi)/D(\chi)\), where \(D(\chi)=\chi(\chi^2-k^2)\). A finite double branch point is a point \(\chi_0\notin\{0,\pm k\}\) such that \(\Lambda'(\chi_0)=0\) and \(\Lambda''(\chi_0)\neq0\). If \(\Lambda''(\chi_0)=0\), the degeneration is of higher order and will be considered in Section~\ref{subsec:triple-root-geometry}.
	
	For \(\chi\notin\{0,\pm k\}\), differentiation gives
	\[
	\Lambda'(\chi)=\frac{H(\chi)}{D(\chi)^2},
	\]
	where \(H(\chi)=N'(\chi)D(\chi)-N(\chi)D'(\chi)\). Substitution of the explicit expressions for \(N\) and \(D\) yields
	\begin{equation}
		\label{eq:H}
		H(\chi)=-2\chi^4+2Sk^2\chi^3+\left(4k^2-3Rk^3\right)\chi^2+Rk^5-2k^4.
	\end{equation}
	The coefficient of \(\chi\) vanishes identically. Since \(D(\chi)\neq0\) on the admissible domain, the finite stationary points of \(\Lambda\) are precisely the zeros of \(H\) outside \(\{0,\pm k\}\).
	
	Because \(H\) has real coefficients, its nonreal roots occur in complex-conjugate pairs. Counting algebraic multiplicities, the possible root distributions are therefore \(4R\), \(2R+1C\), and \(0R+2C\), where \(R\) denotes a real root and \(C\) denotes a nonreal conjugate pair. The last possibility is excluded by the following result.
	
	\begin{theorem}
		\label{thm:root-distribution}
		Assume \(A>0\), \(B>0\), and \(k\neq0\). Then the quartic equation \(H(\chi)=0\) has at least two real roots, counted with multiplicity. Consequently, away from parameter values at which roots coalesce, the only possible root distributions are \(4R\) and \(2R+1C\).
	\end{theorem}
	
	\begin{proof}
		Set \(\chi=ky\). Since \(k\neq0\), substitution into \eqref{eq:H} and division by \(-2k^4\) give the equivalent quartic equation
		\begin{equation}
			\label{eq:P-y}
			P(y)=y^4-(A+B)k\,y^3+\left[\frac32(A-B)k-2\right]y^2+1-\frac12(A-B)k.
		\end{equation}
		At \(y=\pm1\), one has
		\[
		P(1)=-2Bk,\qquad P(-1)=2Ak.
		\]
		
		Suppose first that \(k>0\). Then \(P(-1)>0\) and \(P(1)<0\), so the intermediate value theorem gives a real zero in \((-1,1)\). Since \(P(1)<0\) and \(P(y)\to+\infty\) as \(y\to+\infty\), a second real zero lies in \((1,+\infty)\).
		
		If \(k<0\), then \(P(-1)<0\) and \(P(1)>0\), so there is again a real zero in \((-1,1)\). Moreover, \(P(y)\to+\infty\) as \(y\to-\infty\), while \(P(-1)<0\), and hence a second real zero lies in \((-\infty,-1)\).
		
		Thus \(P\), and equivalently \(H\), always has at least two real roots. Since the coefficients are real, the remaining two roots are either both real or form a nonreal complex-conjugate pair.
	\end{proof}
	
	The equal-amplitude case has a more rigid root structure.
	
	\begin{corollary}
		\label{cor:equal-amplitude-root-distribution}
		Under the assumptions of Theorem~\ref{thm:root-distribution}, suppose in addition that \(A=B\). Then \(H(\chi)=0\) has exactly two real roots, counted with multiplicity, and one nonreal complex-conjugate pair. Thus \(2R+1C\) is the only possible root distribution.
	\end{corollary}
	
	\begin{proof}
		When \(A=B\), equation \eqref{eq:P-y} reduces to
		\[
		P(y)=y^4-2Ak\,y^3-2y^2+1.
		\]
		Set \(\alpha=Ak\). Since \(A>0\) and \(k\neq0\), one has \(\alpha\neq0\).
		
		If \(y\in\mathbb R\) is a root of \(P\), then \(y\neq0\) because \(P(0)=1\). Dividing \(P(y)=0\) by \(y^2\) gives
		\[
		y^2-2\alpha y-2+\frac{1}{y^2}=0,
		\]
		or equivalently
		\[
		\left(y-\frac{1}{y}\right)^2=2\alpha y.
		\]
		The left-hand side is nonnegative, so every real root satisfies \(\alpha y\ge0\). Hence all real roots are positive when \(\alpha>0\) and negative when \(\alpha<0\).
		
		If \(\alpha>0\), the nonzero coefficient signs of \(P\) are \(+\), \(-\), \(-\), and \(+\), which gives two sign changes. Descartes' rule of signs therefore implies that \(P\) has at most two positive real roots, counted with multiplicity. Since it has no negative real roots and Theorem~\ref{thm:root-distribution} guarantees at least two real roots, it has exactly two real roots.
		
		If \(\alpha<0\), the same argument applied to \(P(-y)\) shows that \(P\) has exactly two negative real roots and no positive real roots. In either case, the remaining two roots form a nonreal complex-conjugate pair.
	\end{proof}
	
	\subsection{Real and nonreal double roots}
	\label{subsec:real-nonreal-double}
	
	The reality of a double spatial root is equivalent to the reality of its squared spectral value.
	
	\begin{proposition}
		\label{prop:real-equivalence}
		Let \(\chi_0\notin\{0,\pm k\}\) be an effective double branch point, and define \(\mu_0=\Lambda(\chi_0)\). Then \(\chi_0\in\mathbb R\) if and only if \(\mu_0\in\mathbb R\).
	\end{proposition}
	
	\begin{proof}
		If \(\chi_0\in\mathbb R\), then \(\Lambda\) is finite at \(\chi_0\) and has real coefficients. Hence \(\mu_0=\Lambda(\chi_0)\in\mathbb R\).
		
		Conversely, suppose that \(\mu_0\in\mathbb R\). Then \(F(\chi;\mu_0)\) is a polynomial of degree at most three with real coefficients. Since \(\chi_0\) is a double branch point, \(F(\chi_0;\mu_0)=0\) and \(F_\chi(\chi_0;\mu_0)=0\). If \(\chi_0\notin\mathbb R\), complex conjugation would imply that \(\chi_0^*\) is another double root. The polynomial would then have at least four roots counted with multiplicity, which is impossible. Therefore \(\chi_0\in\mathbb R\).
	\end{proof}
	
	This result separates the two double-root Darboux mechanisms considered later. A nonreal double spatial root corresponds to \(\mu_0\notin\mathbb R\) and can be treated directly by the regular one-fold Darboux transformation. A real double spatial root corresponds to \(\mu_0\in\mathbb R\) and requires a singular real-spectrum limiting construction.
	
	For the real sector, fix an exact double branch point satisfying
	\begin{equation}
		\label{eq:exact-real-double}
		\chi_0\in\mathbb R\setminus\{0,\pm k\},\qquad
		\Lambda'(\chi_0)=0,\qquad
		\Lambda''(\chi_0)\neq0.
	\end{equation}
	The final condition excludes the triple-root degeneration discussed in Section~\ref{subsec:triple-root-geometry}. Set \(\mu_0=\Lambda(\chi_0)\) and \(\lambda_0^2=\mu_0\). The next result shows that the corresponding Darboux spectral point is automatically real.
	
	For this purpose, write
	\begin{equation}
		\label{eq:f-Lambda}
		f(\chi)=2-k^2\left(\frac{A}{\chi+k}+\frac{B}{\chi-k}\right),\qquad
		\Lambda(\chi)=\frac{f(\chi)}{\chi}.
	\end{equation}
	
	\begin{proposition}
		\label{prop:real-double-positive}
		Under the assumptions \eqref{eq:physical-region} and \eqref{eq:exact-real-double}, one has
		\begin{equation}
			\label{eq:mu0-positive-auto}
			\mu_0=\frac{Ak^2}{(\chi_0+k)^2}+\frac{Bk^2}{(\chi_0-k)^2}>0.
		\end{equation}
		Consequently, the Darboux spectral point may be chosen as \(\lambda_0=\sqrt{\mu_0}>0\).
	\end{proposition}
	
	\begin{proof}
		From \eqref{eq:f-Lambda},
		\[
		\Lambda'(\chi)=\frac{\chi f'(\chi)-f(\chi)}{\chi^2}.
		\]
		At \(\chi=\chi_0\), the condition \(\Lambda'(\chi_0)=0\) gives \(f(\chi_0)=\chi_0f'(\chi_0)\). Therefore
		\[
		\mu_0=\Lambda(\chi_0)=\frac{f(\chi_0)}{\chi_0}=f'(\chi_0).
		\]
		Since
		\[
		f'(\chi)=\frac{Ak^2}{(\chi+k)^2}+\frac{Bk^2}{(\chi-k)^2},
		\]
		the assumptions \(A>0\), \(B>0\), \(k\neq0\), and \(\chi_0\neq\pm k\) imply \eqref{eq:mu0-positive-auto}.
	\end{proof}
	
	The remaining simple spatial root of \(F(\chi;\mu_0)=0\) will be denoted by
	\[
	\chi_s=\frac{2}{\mu_0}-2\chi_0,
	\]
	and we set \(L_0=\chi_0-\chi_s\). Since \(\chi_0\) is an exact double root,
	\[
	F(\chi;\mu_0)=\mu_0(\chi-\chi_0)^2(\chi-\chi_s),
	\]
	and hence
	\begin{equation}
		\label{eq:Fchichi-L0}
		F_{\chi\chi}(\chi_0;\mu_0)=2\mu_0L_0.
	\end{equation}
	On the other hand, differentiating the identity \(F(\chi;\mu)=D(\chi)[\mu-\Lambda(\chi)]\) and using \(\mu_0=\Lambda(\chi_0)\) and \(\Lambda'(\chi_0)=0\) gives
	\begin{equation}
		\label{eq:Fchichi-Lambda}
		F_{\chi\chi}(\chi_0;\mu_0)=-D_0\Lambda''(\chi_0),\qquad
		D_0=D(\chi_0)=\chi_0(\chi_0^2-k^2).
	\end{equation}
	Combining \eqref{eq:Fchichi-L0} and \eqref{eq:Fchichi-Lambda}, we obtain
	\begin{equation}
		\label{eq:L0-Lambda2}
		2\mu_0L_0=-D_0\Lambda''(\chi_0)\neq0.
	\end{equation}
	In particular, \(L_0\neq0\), which confirms that the third spatial root remains distinct from the double root.
	
	\subsection{Triple-root conditions}
	\label{subsec:triple-root-geometry}
	
	A triple spatial root is characterized by
	\begin{equation}
		\label{eq:triple-factorization}
		F(\chi;\mu_0)=\mu_0(\chi-\chi_0)^3.
	\end{equation}
	At this stage neither \(\mu_0\) nor \(\chi_0\) is assumed to be real. Their reality must be derived from the coefficient identities in \eqref{eq:triple-factorization}; in particular, a conjugate-root argument cannot be used before \(\mu_0\in\mathbb R\) has been established. The coefficient comparison below gives \(3\mu_0\chi_0=2\), so \(\mu_0\neq0\) and \(\chi_0\neq0\). Proposition~\ref{prop:no-pm-k-root} also excludes \(\chi_0=\pm k\).
	
	The physical triple-root configuration can be characterized entirely in terms of the background parameters.
	
	\begin{theorem}
		\label{thm:triple-characterization}
		Suppose that \(\chi_0\) is a triple root of \(F(\chi;\mu_0)=0\). Then
		\begin{equation}
			\label{eq:mu-triple}
			\mu_0=\frac{2}{3\chi_0},
		\end{equation}
		\[
		S=A+B=\frac{2\chi_0}{k^2}+\frac{2}{3\chi_0},
		\]
		and
		\begin{equation}
			\label{eq:R-triple}
			R=A-B=\frac{2}{k}+\frac{2\chi_0^2}{3k^3}.
		\end{equation}
		In particular, \(R\neq0\), so a physical triple root can occur only on an unequal-amplitude background.
		
		Define \(q_{\rm tri}=\chi_0/k\) and \(\rho=R/S=(A-B)/(A+B)\). Then
		\begin{equation}
			\label{eq:rho-q}
			\rho=\frac{q_{\rm tri}(q_{\rm tri}^2+3)}{3q_{\rm tri}^2+1}.
		\end{equation}
		The unique real solution in \((-1,1)\) is
		\begin{equation}
			\label{eq:q}
			q_{\rm tri}=\tanh\left[\frac13\operatorname{arctanh}\left(\frac{A-B}{A+B}\right)\right].
		\end{equation}
		Consequently,
		\begin{equation}
			\label{eq:k-triple}
			k=\frac{2}{A+B}\left(q_{\rm tri}+\frac{1}{3q_{\rm tri}}\right),\qquad
			\chi_0=q_{\rm tri}k,\qquad
			\mu_0=\lambda_0^2=\frac{2}{3\chi_0}.
		\end{equation}
		The resulting triple branch point satisfies \(0<\chi_0<|k|\) and \(\mu_0>0\), so the Darboux spectral point may be chosen as \(\lambda_0=\sqrt{\mu_0}>0\).
		
		Conversely, the parameters defined by \eqref{eq:q} and \eqref{eq:k-triple} satisfy the triple factorization \eqref{eq:triple-factorization}.
	\end{theorem}
	
	\begin{proof}
		Comparison of the coefficients in
		\[
		F(\chi;\mu_0)=\mu_0(\chi-\chi_0)^3
		\]
		gives
		\[
		3\mu_0\chi_0=2,
		\qquad
		k^2(S-\mu_0)=3\mu_0\chi_0^2,
		\qquad
		k^2(2-kR)=-\mu_0\chi_0^3.
		\]
		In particular, \(\mu_0\chi_0=2/3\). Eliminating \(\mu_0\)
		from the first and third identities yields
		\[
		\chi_0^2=-\frac{3}{2}k^2(2-kR)\in\mathbb R.
		\]
		Hence \(\chi_0\) is either real or purely imaginary. If
		\(\chi_0={\rm i}y\) with \(y\in\mathbb R\setminus\{0\}\), then
		the second coefficient identity gives
		\[
		S=\frac{2\chi_0}{k^2}+\frac{2}{3\chi_0}
		={\rm i}\left(\frac{2y}{k^2}-\frac{2}{3y}\right),
		\]
		which is incompatible with \(S=A+B>0\). Therefore
		\(\chi_0\in\mathbb R\), and then
		\(\mu_0=2/(3\chi_0)\in\mathbb R\).
		The remaining coefficient comparisons give
		\[
		S=\frac{2\chi_0}{k^2}+\frac{2}{3\chi_0},
		\qquad
		R=\frac{2}{k}+\frac{2\chi_0^2}{3k^3},
		\]
		and division gives \eqref{eq:rho-q}. Since \(A>0\) and \(B>0\), one has \(|\rho|<1\). Moreover, \(R\neq0\), so \(\rho\neq0\). The identity
		\[
		\tanh(3z)=\frac{\tanh z\,(\tanh^2z+3)}{3\tanh^2z+1}
		\]
		then shows that \eqref{eq:rho-q} has the unique real solution \eqref{eq:q} in \((-1,1)\). Solving the expression for \(S=A+B\) gives \eqref{eq:k-triple}.
		
		Finally,
		\[
		\chi_0=q_{\rm tri}k=\frac{2}{A+B}\left(q_{\rm tri}^2+\frac13\right)>0.
		\]
		Since \(|q_{\rm tri}|<1\), it follows that \(\chi_0<|k|\), while \eqref{eq:mu-triple} gives \(\mu_0>0\). The converse is verified by substituting \eqref{eq:q} and \eqref{eq:k-triple} into the coefficients of \(F\).
	\end{proof}
	
	\subsection{Summary of the spectral sectors}
	\label{subsec:spectral-sector-summary}
	
	The preceding results divide the relevant multiple-root configurations into the four sectors summarized in Table~\ref{tab:spectral-sectors}. The first three sectors correspond to distinct Darboux mechanisms, whereas the persistent zero-branch intersection is a genuine component of the full polynomial spectral curve but not a ramification point of the reduced map \(\Lambda\).
	
	\begin{table}[htbp]
		\centering
		\caption{Multiple-root spectral sectors and the corresponding Darboux mechanisms.}
		\label{tab:spectral-sectors}
		\begin{tabularx}{\textwidth}{>{\raggedright\arraybackslash}p{0.20\textwidth} >{\raggedright\arraybackslash}p{0.18\textwidth} >{\raggedright\arraybackslash}p{0.18\textwidth} X}
			\toprule
			Spatial configuration & Squared spectral value & Darboux spectral point & Construction \tabularnewline
			\midrule
			Nonreal double root & \(\mu_0\notin\mathbb R\) & \(\lambda_0\notin\mathbb R\cup{\rm i}\mathbb R\) & Regular one-fold Darboux transformation using the double-root derivative chain and the remaining simple mode \tabularnewline
			Real double root & \(\mu_0>0\) & \(\lambda_0\in\mathbb R\setminus\{0\}\) & Directional real-spectrum limit in the mixed double--simple eigenspace \tabularnewline
			Real triple root & \(\mu_0>0\) & \(\lambda_0\in\mathbb R\setminus\{0\}\) & Directional real-spectrum limit based on the full triple confluent chain \tabularnewline
			Persistent zero-branch intersection & Depends on the intersecting movable branch & Depends on \(\mu\) & Genuine branch intersection excluded from the ramification-point classification of \(\Lambda\) \tabularnewline
			\bottomrule
		\end{tabularx}
	\end{table}
	
	\section{Darboux Constructions at Double Spatial Roots}
	\label{sec:double-darboux}
	
	\subsection{Regular construction at a nonreal double root}
	\label{sec:nonreal-double}
	\label{subsec:nonreal-double}
	
	We begin with a nonreal double spatial root. By Proposition~\ref{prop:real-equivalence}, its squared spectral value is nonreal, so either square root \(\lambda_0\) lies outside \(\mathbb R\cup{\rm i}\mathbb R\). The ordinary one-fold Darboux transformation is therefore regular, and no real-spectrum limiting argument is needed.
	
	We first construct the confluent eigenfunctions associated with the double spatial root.
	
	Let \(\chi_0\notin\mathbb R\) be an exact double branch point and set \(\lambda_0^2=\mu_0=\Lambda(\chi_0)\). Since \(\chi_0\) is an exact double root of \(F(\chi;\mu_0)=0\), the cubic factorizes as \(F(\chi;\mu_0)=\mu_0(\chi-\chi_0)^2(\chi-\chi_s)\), where the remaining simple root is \(\chi_s=\frac{2}{\mu_0}-2\chi_0\).
	
	Define
	\begin{equation*}
		{\bf Y}_0
		=
		{\bf y}(\chi_0,\lambda_0),
		\qquad
		{\bf Y}_1
		=
		\left.
		\frac{\partial}{\partial\chi}
		{\bf y}(\chi,\lambda_0)
		\right|_{\chi=\chi_0},
		\qquad
		{\bf Y}_s
		=
		{\bf y}(\chi_s,\lambda_0),
	\end{equation*}
	where the derivative defining \({\bf Y}_1\) is taken at fixed \(\lambda=\lambda_0\). Because an exact double branch point satisfies \(\Lambda'(\chi_0)=0\), Lemma~\ref{lem:fixed-spectral-derivative-chain} with \(m=2\) proves that \({\bf Y}_0\) and \({\bf Y}_1\) satisfy both equations of the background Lax pair at \(\lambda=\lambda_0\). The remaining simple mode \({\bf Y}_s\) is an elementary solution at the same spectral point. Proposition~\ref{prop:confluent-basis-independence} proves, including the persistent-zero-branch case, that \({\bf Y}_0,{\bf Y}_1,{\bf Y}_s\) are linearly independent. They therefore form a fundamental solution basis of the three-dimensional Lax system at \(\lambda=\lambda_0\).
	
	The resulting confluent basis can be inserted directly into the one-fold Darboux transformation.
	
	\begin{theorem}
		\label{thm:complex-double}
		Let \(\chi_0\notin\mathbb R\) be an effective double branch point, set \(\mu_0=\Lambda(\chi_0)\), and choose the square root \(\lambda_0^2=\mu_0\) with \(\operatorname{Re}\lambda_0>0\). For arbitrary constants \(c_0,c_1,c_s\in\mathbb C\), not all zero, define \({\bf y}_D=c_0{\bf Y}_0 + c_1{\bf Y}_1 + c_s{\bf Y}_s\). Write \({\bf y}_D=(\psi_D,\phi_D,\varphi_D)^T\), and define
		\begin{equation}
			\label{eq:mD}
			m_D
			=
			\frac{
				\langle{\bf y}_D|J|{\bf y}_D\rangle
			}{
				\lambda_0^*-\lambda_0
			}
			-
			\frac{
				\langle{\bf y}_D|K|{\bf y}_D\rangle
			}{
				\lambda_0^*+\lambda_0
			}.
		\end{equation}
		Then
		\begin{equation}
			\label{eq:complex-double-solution}
			\begin{pmatrix}
				u_1^{(C,2)}\\
				u_2^{(C,2)}
			\end{pmatrix}
			=
			\begin{pmatrix}
				u_{1,0}\\
				u_{2,0}
			\end{pmatrix}
			+
			2
			\frac{
				\begin{pmatrix}
					\phi_D\\
					\varphi_D
				\end{pmatrix}
				\psi_D^*
			}{
				m_D
			}
		\end{equation}
		is an exact solution of \eqref{eq:CFL}.
		
		Since \(K=I_3\), one has \(\operatorname{Re}m_D=-\langle{\bf y}_D|{\bf y}_D\rangle/(2\operatorname{Re}\lambda_0)<0\). Hence \(m_D\) has no zeros on the real \((x,t)\)-plane, and the solution \eqref{eq:complex-double-solution} is globally nonsingular.
	\end{theorem}
	
	\begin{proof}
		By Proposition~\ref{prop:real-equivalence}, \(\chi_0\notin\mathbb R\) implies \(\mu_0\notin\mathbb R\). Hence a square root \(\lambda_0\) of \(\mu_0\) is neither real nor purely imaginary. Therefore, \(\lambda_0^*-\lambda_0\neq0\) and \(\lambda_0^*+\lambda_0\neq0\).
		
		By Lemma~\ref{lem:fixed-spectral-derivative-chain}, \({\bf Y}_0\) and \({\bf Y}_1\) satisfy the background Lax pair at \(\lambda=\lambda_0\), and \({\bf Y}_s\) is the elementary solution associated with the remaining simple root. Their linear combination \({\bf y}_D\) is therefore a valid eigenfunction at the same nonreal spectral point. Substitution into the standard one-fold Darboux formula \eqref{eq:nonreal-DT} gives \eqref{eq:complex-double-solution}, which is consequently an exact solution of \eqref{eq:CFL}.
		
		It remains to establish nonsingularity. Since \(J\) and \(K\) are Hermitian,
		\[
		\langle{\bf y}_D|J|{\bf y}_D\rangle\in\mathbb R, \qquad \langle{\bf y}_D|K|{\bf y}_D\rangle\in\mathbb R.
		\]
		The denominator
		\[
		\lambda_0^*-\lambda_0=-2{\rm i}\operatorname{Im}\lambda_0
		\]
		is purely imaginary. Hence the first term in \(m_D\) is purely imaginary.
		
		In the present reduction \(K=I_3\), while
		\[
		\lambda_0^*+\lambda_0=2\operatorname{Re}\lambda_0.
		\]
		Thus
		\[
		\operatorname{Re}m_D=- \frac{ \langle{\bf y}_D|{\bf y}_D\rangle }{ 2\operatorname{Re}\lambda_0 }.
		\]
		By the branch choice \(\operatorname{Re}\lambda_0>0\), this quantity is strictly negative for every nonzero eigenfunction \({\bf y}_D\). Therefore \(m_D(x,t)\neq0\) for all real \((x,t)\).
	\end{proof}
	
	\begin{remark}
		\label{rem:lambda-sign-complex-double}
		The two choices of the square root, \(\lambda_0\) and \(-\lambda_0\), generate the same Darboux solution. Indeed, replacing \(\lambda_0\) by \(-\lambda_0\) changes the sign of \(\psi_D\) but leaves \(\phi_D\) and \(\varphi_D\) unchanged. At the same time, both denominators in \eqref{eq:mD} change sign, so that \(m_D\mapsto-m_D\). Consequently, the quotient in \eqref{eq:complex-double-solution} remains unchanged.
	\end{remark}
	
	\begin{remark}
		\label{rem:complex-double-coefficients}
		Multiplication of \({\bf y}_D\) by a nonzero constant does not change the Darboux field, because both the numerator and denominator in \eqref{eq:complex-double-solution} are multiplied by the same positive factor. Thus
		\[
		(c_0,c_1,c_s) \sim \gamma(c_0,c_1,c_s), \qquad \gamma\in\mathbb C\setminus\{0\}.
		\]
	\end{remark}
	
	\subsection{The mixed double--simple confluent eigenspace and its null cone}
	\label{sec:real-double}
	\label{subsec:double-null-cone}
	
	We next consider the exact real double root fixed in \eqref{eq:exact-real-double}. Proposition~\ref{prop:real-double-positive} permits the choice \(\lambda_0=\sqrt{\mu_0}>0\). Since the ordinary Darboux kernel becomes singular as the spectral point approaches the real axis, we first determine all projective leading vectors for which the singular \(J\)-term cancels.
	
	\begin{lemma}
		\label{lem:J-null}
		Let \(\lambda=\sqrt{\Lambda(\chi)}\in\mathbb R\). Then the elementary eigenfunction \eqref{eq:y-basic} satisfies
		\begin{equation}
			\label{eq:J-null-identity}
			\langle{\bf y}(\chi,\lambda)
			|J|
			{\bf y}(\chi,\lambda)\rangle
			=
			-\chi\Lambda'(\chi).
		\end{equation}
		Consequently,
		\begin{equation}
			\label{eq:Y0-J-null}
			\langle{\bf Y}_0|J|{\bf Y}_0\rangle=0,
			\qquad
			{\bf Y}_0={\bf y}(\chi_0,\lambda_0).
		\end{equation}
	\end{lemma}
	
	\begin{proof}
		For real \(\chi\) and real \(\lambda\), one has
		\[
		|E(\chi)|=1, \qquad h_1(\chi),h_2(\chi)\in\mathbb R, \qquad |\lambda|^2=\lambda^2.
		\]
		Therefore
		\[
		\begin{aligned}
			\langle{\bf y}|J|{\bf y}\rangle
			&=
			\lambda^2
			-Ah_1(\chi)^2
			-Bh_2(\chi)^2\\
			&=
			\frac{f(\chi)}{\chi}
			-
			\frac{Ak^2}{(\chi+k)^2}
			-
			\frac{Bk^2}{(\chi-k)^2}\\
			&=
			\frac{f(\chi)}{\chi}-f'(\chi)
			=
			-\chi\Lambda'(\chi).
		\end{aligned}
		\]
		Evaluating this identity at \(\chi=\chi_0\), where \(\Lambda'(\chi_0)=0\), yields \eqref{eq:Y0-J-null}.
	\end{proof}
	
	\begin{remark}
		The identity \eqref{eq:J-null-identity} is stated only in the real branch-point setting. For complex \(\chi\), the Hermitian inner product contains \(|E|^2\), \(|\lambda|^2\), and \(|h_j|^2\), and hence the same algebraic identity does not hold in general.
	\end{remark}
	
	The potentially singular contribution to the Darboux kernel is
	\[
	\frac{\langle{\bf y}|J|{\bf y}\rangle}{\lambda^*-\lambda}.
	\]
	A nontrivial finite real-spectrum correction therefore requires the leading vector to be \(J\)-null. If its leading \(J\)-norm is nonzero, the singular term dominates and the Darboux correction converges to zero, leaving only the seed solution.
	
	Define the fixed-\(\lambda_0\) derivative vectors
	\begin{equation*}
		{\bf Y}_n
		=
		\left.
		\frac{\partial^n}{\partial\chi^n}
		{\bf y}(\chi,\lambda_0)
		\right|_{\chi=\chi_0},
		\qquad
		n=0,1,2,3.
	\end{equation*}
	Let
	\begin{equation}
		\label{eq:Ys-real-double}
		{\bf Y}_s
		=
		{\bf y}(\chi_s,\lambda_0)
	\end{equation}
	denote the elementary eigenfunction associated with the remaining simple spatial root \(\chi_s\). If \(\chi_s=0\), the right-hand side of \eqref{eq:Ys-real-double} is understood in the sense of \eqref{eq:zero-branch-eigenfunction}.
	
	The full fixed-\(\lambda_0\) confluent eigenspace at an exact real double branch point is \(\mathcal E_D=\operatorname{span}_{\mathbb C} \{ {\bf Y}_0,\, {\bf Y}_1,\, {\bf Y}_s \}\). Accordingly, the most general leading vector is
	\begin{equation}
		\label{eq:general-double-leading}
		{\bf y}_R
		=
		M_0{\bf Y}_0
		+
		M_1{\bf Y}_1
		+
		c_s{\bf Y}_s.
	\end{equation}
	
	For brevity, write \([{\bf f},{\bf g}]_J:=\langle{\bf f}|J|{\bf g}\rangle\). The conserved \(J\)-pairing imposes a particularly simple Gram structure on this eigenspace.
	
	\begin{lemma}
		\label{lem:double-Gram-structure}
		At an exact real double branch point, one has
		\begin{equation}
			\label{eq:double-Gram-relations}
			[{\bf Y}_0,{\bf Y}_0]_J=0,
			\qquad
			[{\bf Y}_0,{\bf Y}_s]_J=0,
			\qquad
			[{\bf Y}_1,{\bf Y}_s]_J=0.
		\end{equation}
		Define \(a_J=[{\bf Y}_0,{\bf Y}_1]_J\), \(b_J=[{\bf Y}_1,{\bf Y}_1]_J\), and \(d_s=[{\bf Y}_s,{\bf Y}_s]_J\). Then
		\[
		a_J
		=-Ah_1(\chi_0)h_1'(\chi_0)
		-Bh_2(\chi_0)h_2'(\chi_0)
		=-\frac{\chi_0}{2}\Lambda''(\chi_0),
		\]
		and
		\[
		b_J
		=-A[h_1'(\chi_0)]^2
		-B[h_2'(\chi_0)]^2<0.
		\]
		Moreover,
		\begin{equation}
			\label{eq:ds-simple-root-location}
			d_s
			=
			\frac{\mu_0L_0^2}{\chi_s^2-k^2}
			<0,
			\qquad
			L_0=\chi_0-\chi_s,
		\end{equation}
		and consequently
		\begin{equation}
			\label{eq:simple-root-inside-carrier}
			|\chi_s|<|k|.
		\end{equation}
		In particular, \(a_J\neq0\), \(b_J\neq0\), and \(d_s\neq0\).
	\end{lemma}
	
	\begin{proof}
		The first relation in \eqref{eq:double-Gram-relations} follows from Lemma~\ref{lem:J-null}.
		
		The vectors \({\bf Y}_0\) and \({\bf Y}_s\) correspond to the distinct real spatial roots \(\chi_0\neq\chi_s\) at the same real spectral value \(\lambda_0\). After removal of the common plane-wave gauge, their \(J\)-pairing has the form
		\[
		[{\bf Y}_0,{\bf Y}_s]_J=C_{0s} \exp \left\{ {\rm i} \left[ (\chi_s-\chi_0)x + (\nu_s-\nu_0)t \right] \right\},
		\]
		where \(C_{0s}\) is independent of \(x,t\). By Lemma~\ref{lem:J-form-conservation}, this pairing must be independent of \(x\). Since \(\chi_s-\chi_0\neq0\), one must have \(C_{0s}=0\). Hence \([{\bf Y}_0,{\bf Y}_s]_J=0\).
		
		Likewise, because \({\bf Y}_1\) is obtained by differentiating the \(\chi_0\)-mode, its pairing with the simple mode has the form
		\[
		[{\bf Y}_1,{\bf Y}_s]_J=\left( C_{1s}^{(0)}+C_{1s}^{(1)}x \right) \exp \left\{ {\rm i} \left[ (\chi_s-\chi_0)x + (\nu_s-\nu_0)t \right] \right\}.
		\]
		The only such function that is independent of \(x\) is the zero function. Therefore \([{\bf Y}_1,{\bf Y}_s]_J=0\). The same argument remains valid when \(\chi_s=0\), provided that the zero-branch eigenfunction \eqref{eq:zero-branch-eigenfunction} is used.
		
		Set \(X=x+\nu'(\chi_0)t\). Writing
		\[
		{\bf Y}_1={\rm i}X{\bf Y}_0 + E_0
		\begin{pmatrix}
			0\\
			u_{1,0}h_1'(\chi_0)\\
			u_{2,0}h_2'(\chi_0)
		\end{pmatrix},
		\]
		and using \([{\bf Y}_0,{\bf Y}_0]_J=0\), one obtains \(a_J=-Ah_1(\chi_0)h_1'(\chi_0) - Bh_2(\chi_0)h_2'(\chi_0)\), and \(b_J=-A[h_1'(\chi_0)]^2 - B[h_2'(\chi_0)]^2\). The latter is strictly negative because \(A,B>0\), \(h_1'(\chi_0)\neq0\), and \(h_2'(\chi_0)\neq0\).
		
		To relate \(a_J\) to the spectral geometry, recall
		\[
		\Lambda(\chi)=\frac{f(\chi)}{\chi}.
		\]
		At a stationary point, \(\Lambda'(\chi_0)=0\), and direct differentiation gives
		\[
		\Lambda''(\chi_0)=\frac{f''(\chi_0)}{\chi_0}.
		\]
		On the other hand,
		\[
		a_J=-\frac12f''(\chi_0),
		\]
		which proves
		\[
		a_J=-\frac{\chi_0}{2}\Lambda''(\chi_0).
		\]
		Since \(\chi_0\neq0\) and \(\Lambda''(\chi_0)\neq0\), one has \(a_J\neq0\).
		
		Finally, in the ordered basis \(\{ {\bf Y}_0, {\bf Y}_1, {\bf Y}_s \}\), the \(J\)-Gram matrix is
		\begin{equation*}
			G_D
			=
			\begin{pmatrix}
				0&a_J&0\\
				a_J&b_J&0\\
				0&0&d_s
			\end{pmatrix}.
		\end{equation*}
		The vectors form a fundamental solution basis at \(\lambda=\lambda_0\): \({\bf Y}_0,{\bf Y}_1\) span the double-root confluent sector, while \({\bf Y}_s\) corresponds to the remaining simple spatial root. Hence Lemma~\ref{lem:J-form-conservation} implies that \(G_D\) has inertia \((1,2)\).
		
		The first \(2\times2\) block has determinant \(-a_J^2<0\), and therefore has one positive and one negative eigenvalue. The remaining eigenvalue must consequently be negative. Thus \(d_s<0\).
		
		It remains to derive the explicit formula for \(d_s\). First suppose that \(\chi_s\neq0\). Lemma~\ref{lem:J-null} gives
		\[
		d_s=-\chi_s\Lambda'(\chi_s).
		\]
		Differentiating the exact factorization
		\[
		F(\chi;\mu_0)
		=
		\mu_0(\chi-\chi_0)^2(\chi-\chi_s)
		\]
		at \(\chi=\chi_s\) gives
		\begin{equation}
			\label{eq:Fchi-simple-factorized}
			F_\chi(\chi_s;\mu_0)
			=
			\mu_0(\chi_s-\chi_0)^2
			=
			\mu_0L_0^2.
		\end{equation}
		On the other hand, differentiating
		\(F(\chi;\mu)=D(\chi)[\mu-\Lambda(\chi)]\) at the simple root yields
		\[
		F_\chi(\chi_s;\mu_0)
		=-D(\chi_s)\Lambda'(\chi_s)
		=(\chi_s^2-k^2)d_s.
		\]
		Comparison with \eqref{eq:Fchi-simple-factorized} proves \eqref{eq:ds-simple-root-location} when \(\chi_s\neq0\).
		
		If \(\chi_s=0\) is the persistent zero branch, formula \eqref{eq:zero-branch-eigenfunction} gives
		\[
		d_s=\mu_0-S.
		\]
		The factorization now reads
		\[
		F(\chi;\mu_0)
		=
		\mu_0(\chi-\chi_0)^2\chi.
		\]
		Thus
		\[
		F_\chi(0;\mu_0)
		=
		\mu_0\chi_0^2
		=
		\mu_0L_0^2.
		\]
		The polynomial expression \eqref{eq:F} also gives
		\(F_\chi(0;\mu_0)=k^2(S-\mu_0)\). Hence
		\[
		d_s
		=
		\mu_0-S
		=-\frac{\mu_0L_0^2}{k^2},
		\]
		which is again \eqref{eq:ds-simple-root-location} with \(\chi_s=0\).
		
		Finally, \(\mu_0>0\) and \(L_0\neq0\), so the numerator in \eqref{eq:ds-simple-root-location} is strictly positive. Since \(d_s<0\), its denominator must be negative: \(\chi_s^2-k^2<0\). This proves \eqref{eq:simple-root-inside-carrier}.
	\end{proof}
	
	The resulting Gram matrix allows the complete projective null cone to be described explicitly.
	
	\begin{proposition}
		\label{prop:admissible-double}
		A nonzero vector of the form \eqref{eq:general-double-leading} can be \(J\)-null only if \(M_0\neq0\). After the projective normalization \(M_0=1\), the null condition is
		\begin{equation}
			\label{eq:double-null-general}
			b_J|M_1|^2
			+
			2a_J\operatorname{Re}M_1
			+
			d_s|c_s|^2
			=
			0.
		\end{equation}
		Equivalently,
		\begin{equation}
			\label{eq:double-null-circle}
			\left(
			\operatorname{Re}M_1+\frac{a_J}{b_J}
			\right)^2
			+
			\left(
			\operatorname{Im}M_1
			\right)^2
			=
			\left(
			\frac{a_J}{b_J}
			\right)^2
			-
			\frac{d_s}{b_J}|c_s|^2.
		\end{equation}
		Consequently, the simple branch can occur in a null leading vector if and only if
		\begin{equation}
			\label{eq:cs-bound-double}
			|c_s|^2
			\le
			\frac{a_J^2}{b_Jd_s}.
		\end{equation}
	\end{proposition}
	
	\begin{proof}
		Using Lemma~\ref{lem:double-Gram-structure}, one obtains
		\[
		[{\bf y}_R,{\bf y}_R]_J=2a_J\operatorname{Re}(M_0^*M_1) + b_J|M_1|^2 + d_s|c_s|^2.
		\]
		If \(M_0=0\), then \([{\bf y}_R,{\bf y}_R]_J=b_J|M_1|^2+d_s|c_s|^2<0\) for every nonzero pair \((M_1,c_s)\). Hence every nonzero null vector has \(M_0\neq0\).
		
		After normalizing \(M_0=1\), the null condition becomes \eqref{eq:double-null-general}. Completing the square in \(M_1\) gives \eqref{eq:double-null-circle}. The circle is nonempty precisely when its squared radius is nonnegative, which gives \eqref{eq:cs-bound-double}.
	\end{proof}
	
	We first construct the limit in the generic mixed double--simple sector, where \(M_1c_s\neq0\) and \(b_J|M_1|^2+2a_J\operatorname{Re}M_1+d_s|c_s|^2=0\). The pure double-root and lower coefficient sectors will then follow by direct specialization.
	
	\subsection{Symmetric splitting of a real double root}
	\label{subsec:double-root-splitting}
	
	\label{subsec:symmetric-root-splitting}
	
	The double degeneration has already been defined at the fixed endpoint \(\mu_0\) by the multiplicity-two equation \(F(\chi;\mu_0)=0\). We now displace this single endpoint to construct a regular nonreal-spectrum approximating family for the singular Darboux kernel. This auxiliary displacement does not introduce, or coalesce, several distinct Darboux poles.
	
	Let \(\mu=\mu_0+\eta\), where \(\eta\in\mathbb C\) and \(0<|\eta|\ll1\). Since \(\mu_0=\lambda_0^2>0\), there exists a unique local holomorphic square-root branch satisfying \(\lambda(\eta)^2=\mu_0+\eta\) and \(\lambda(0)=\lambda_0\). Define \(D_0=\chi_0^3-k^2\chi_0\) and \(D_1=3\chi_0^2-k^2\). Since \(\chi_0\notin\{0,\pm k\}\), one has \(D_0\neq0\).
	
	At \(\mu=\mu_0\), the exact-double-root assumption gives \(F(\chi;\mu_0)=\mu_0(\chi-\chi_0)^2(\chi-\chi_s)\). Since \(F\) is affine in \(\mu\),
	\begin{equation}
		\label{eq:F-mu-perturbed}
		F(\chi;\mu_0+\eta)
		=
		F(\chi;\mu_0)
		+
		\eta(\chi^3-k^2\chi).
	\end{equation}
	Introduce \(\chi=\chi_0+\delta\) and \(L_0=\chi_0-\chi_s\). Substitution into \eqref{eq:F-mu-perturbed} gives the exact cubic equation
	\begin{equation}
		\label{eq:delta-cubic-exact}
		G(\delta,\eta)
		:=
		(\mu_0+\eta)\delta^3
		+
		(\mu_0L_0+3\chi_0\eta)\delta^2
		+
		D_1\eta\,\delta
		+
		D_0\eta
		=
		0.
	\end{equation}
	At \(\eta=0\), \(G(\delta,0)=\mu_0\delta^2(\delta+L_0)\), so \(\delta=0\) is a double root and \(\delta=-L_0\) is simple. The condition \(L_0\neq0\) excludes a triple spatial root. We first separate these roots uniformly for small nonzero \(\eta\).
	
	\begin{lemma}
		\label{lem:rouche-double-splitting}
		For all sufficiently small nonzero \(\eta\), the equation \(G(\delta,\eta)=0\) has exactly two roots, counted with multiplicity, in a disk of radius \(O(|\eta|^{1/2})\) centered at \(\delta=0\), and exactly one root in a fixed neighborhood of \(\delta=-L_0\).
	\end{lemma}
	
	\begin{proof}
		Write \(G(\delta,\eta)=G_0(\delta)+\eta H(\delta)\), where \(G_0(\delta)=\mu_0\delta^2(\delta+L_0)\), and \(H(\delta)=\delta^3 + 3\chi_0\delta^2 + D_1\delta + D_0\).
		
		We first isolate the root continuing from \(-L_0\). Choose
		\[
		0<\rho<\frac{|L_0|}{3}.
		\]
		On the circle \(|\delta+L_0|=\rho\), one has
		\[
		\begin{aligned}
			|G_0(\delta)|
			&=
			\mu_0|\delta|^2|\delta+L_0|\\
			&=
			\mu_0\rho|\delta|^2.
		\end{aligned}
		\]
		Moreover, \(|\delta| \ge |L_0|-|\delta+L_0|=|L_0|-\rho\), and hence \(|G_0(\delta)| \ge \mu_0\rho(|L_0|-\rho)^2\). Since \(H\) is bounded on this circle, \(|\eta H(\delta)|<|G_0(\delta)|\) for all sufficiently small \(|\eta|\). Rouch\'e's theorem therefore implies that \(G(\cdot,\eta)\) and \(G_0\) have the same number of zeros in \(|\delta+L_0|<\rho\). Since \(G_0\) has exactly one simple zero there, namely \(\delta=-L_0\), the perturbed polynomial has exactly one zero in this disk.
		
		We next determine the scale of the two roots near zero. Define the leading-order comparison polynomial \(P_\eta(\delta)=\mu_0L_0\delta^2+D_0\eta\). Choose \(R>0\) such that
		\begin{equation}
			\label{eq:R-choice-double}
			|\mu_0L_0|R^2>|D_0|.
		\end{equation}
		On the circle \(|\delta|=R|\eta|^{1/2}\), the reverse triangle inequality gives
		\begin{align}
			|P_\eta(\delta)|
			&\ge
			|\mu_0L_0||\delta|^2
			-
			|D_0||\eta|
			\nonumber\\
			&=
			\left(
			|\mu_0L_0|R^2-|D_0|
			\right)|\eta|.
			\label{eq:P-eta-lower}
		\end{align}
		The coefficient in parentheses is strictly positive by \eqref{eq:R-choice-double}.
		
		On the same circle,
		\begin{align*}
			G(\delta,\eta)-P_\eta(\delta)
			&=
			\mu_0\delta^3
			+
			\eta
			\left(
			D_1\delta
			+
			3\chi_0\delta^2
			+
			\delta^3
			\right).
		\end{align*}
		Since \(|\delta|=R|\eta|^{1/2}\), the terms on the right-hand side have respective orders
		\[
		O(|\eta|^{3/2}), \qquad O(|\eta|^{3/2}), \qquad O(|\eta|^2), \qquad O(|\eta|^{5/2}).
		\]
		Consequently,
		\begin{equation}
			\label{eq:G-minus-P-order}
			|G(\delta,\eta)-P_\eta(\delta)|
			=
			O(|\eta|^{3/2})
		\end{equation}
		uniformly on the circle. Since \(|\eta|^{3/2}=o(|\eta|)\), equations \eqref{eq:P-eta-lower} and \eqref{eq:G-minus-P-order} imply that \(|G-P_\eta|<|P_\eta|\) for all sufficiently small nonzero \(\eta\).
		
		By Rouch\'e's theorem, \(G\) and \(P_\eta\) have the same number of zeros in \(|\delta|<R|\eta|^{1/2}\). The two zeros of \(P_\eta\) are
		\[
		\delta=\pm \sqrt{ -\frac{D_0\eta}{\mu_0L_0} },
		\]
		and condition \eqref{eq:R-choice-double} ensures that both lie strictly inside this disk. Hence \(G\) also has exactly two zeros there, counted with multiplicity. These roots satisfy \(|\delta_j|<R|\eta|^{1/2}\) for \(j=1,2\). The third root is the unique root previously isolated near \(-L_0\).
	\end{proof}
	
	By Lemma~\ref{lem:rouche-double-splitting}, the roots may now be labeled unambiguously so that \(\chi_j(\eta)=\chi_0+\delta_j(\eta)\) for \(j=1,2\), with \(\delta_1,\delta_2=O(|\eta|^{1/2}) \longrightarrow0\), whereas the remaining root satisfies \(\delta_3\longrightarrow-L_0\).
	
	The individual branches \(\delta_1\) and \(\delta_2\) depend on a choice of square-root sheet. Their symmetric sum and product are independent of the labeling and are single-valued analytic functions of \(\eta\). These symmetric quantities are therefore the natural variables for the confluent construction.
	
	\begin{lemma}
		\label{lem:moment-double}
		The symmetric combinations of the two roots splitting from \(\delta=0\) satisfy \(\delta_1+\delta_2=s_1\eta+O(|\eta|^2)\) and \(\delta_1\delta_2=p_1\eta+O(|\eta|^2)\), where \(p_1=\frac{D_0}{\mu_0L_0}=-\frac{2}{\Lambda''(\chi_0)}\), and \(s_1=\frac{p_1-D_1/\mu_0}{L_0}=-\frac{2\Lambda'''(\chi_0)} {3[\Lambda''(\chi_0)]^2}\). Equivalently, after choosing a local branch of \(\eta^{1/2}\), the individual split roots have the Puiseux expansions
		\begin{equation}
			\label{eq:Puiseux-double}
			\delta_{1,2}
			=
			\pm
			\sqrt{
				\frac{2\eta}{\Lambda''(\chi_0)}
			}
			-
			\frac{\Lambda'''(\chi_0)}
			{3[\Lambda''(\chi_0)]^2}\eta
			+
			O(|\eta|^{3/2}).
		\end{equation}
	\end{lemma}
	
	\begin{proof}
		The root \(\delta=-L_0\) of \(G(\delta,0)=0\) is simple, since \(G_\delta(-L_0,0)=\mu_0L_0^2 \neq0\). Therefore, the holomorphic implicit function theorem yields a unique analytic root \(\delta_3(\eta)\) near \(\eta=0\) such that \(G(\delta_3(\eta),\eta)=0\) and \(\delta_3(0)=-L_0\). In particular, \(\delta_3\) admits the Taylor expansion
		\begin{equation}
			\label{eq:delta3-first-order}
			\delta_3(\eta)
			=
			-L_0+\kappa_1\eta+O(|\eta|^2),
		\end{equation}
		where the implicit-function differentiation formula gives
		\[
		\kappa_1=\delta_3'(0)=-\frac{G_\eta(-L_0,0)} {G_\delta(-L_0,0)}.
		\]
		From \(G_\eta(\delta,\eta)=\delta^3 + 3\chi_0\delta^2 + D_1\delta + D_0\), we obtain \(G_\eta(-L_0,0)=-L_0^3 + 3\chi_0L_0^2 - D_1L_0 + D_0\). Hence
		\begin{equation}
			\label{eq:kappa1-third-root}
			\kappa_1
			=
			\frac{
				L_0^3
				-
				3\chi_0L_0^2
				+
				D_1L_0
				-
				D_0
			}{
				\mu_0L_0^2
			}.
		\end{equation}
		
		Vieta's product formula applied to \eqref{eq:delta-cubic-exact} gives the exact identity \(\delta_1\delta_2\delta_3=-\frac{D_0\eta}{\mu_0+\eta}\). Therefore,
		\begin{equation}
			\label{eq:p-analytic-exact}
			\delta_1\delta_2
			=
			-\frac{D_0\eta}
			{(\mu_0+\eta)\delta_3(\eta)}.
		\end{equation}
		The right-hand side is analytic in \(\eta\) near zero because \(\mu_0\neq0\) and \(\delta_3(0)=-L_0\neq0\). Using
		\[
		\frac{1}{\mu_0+\eta}=\frac{1}{\mu_0}+O(|\eta|),
		\]
		and
		\[
		\frac{1}{\delta_3(\eta)}=-\frac{1}{L_0}+O(|\eta|),
		\]
		in \eqref{eq:p-analytic-exact}, we obtain
		\[
		\delta_1\delta_2=\frac{D_0}{\mu_0L_0}\eta + O(|\eta|^2).
		\]
		Thus
		\[
		p_1=\frac{D_0}{\mu_0L_0}.
		\]
		
		We next use the total-sum Vieta relation. For \eqref{eq:delta-cubic-exact}, it reads \(\delta_1+\delta_2+\delta_3=-\frac{ \mu_0L_0+3\chi_0\eta }{ \mu_0+\eta }\). Consequently,
		\begin{equation}
			\label{eq:s-analytic-total-sum}
			\delta_1+\delta_2
			=
			-\frac{
				\mu_0L_0+3\chi_0\eta
			}{
				\mu_0+\eta
			}
			-
			\delta_3(\eta).
		\end{equation}
		Both terms on the right-hand side are analytic near \(\eta=0\); hence \(\delta_1+\delta_2\) is analytic in \(\eta\).
		
		Expanding the rational term gives
		\begin{align}
			-\frac{
				\mu_0L_0+3\chi_0\eta
			}{
				\mu_0+\eta
			}
			&=
			-\left(
			L_0+\frac{3\chi_0}{\mu_0}\eta
			\right)
			\left(
			1-\frac{\eta}{\mu_0}
			+
			O(|\eta|^2)
			\right)
			\nonumber\\
			&=
			-L_0
			-
			\frac{3\chi_0-L_0}{\mu_0}\eta
			+
			O(|\eta|^2).
			\label{eq:Vieta-total-sum-expansion}
		\end{align}
		Combining \eqref{eq:delta3-first-order}, \eqref{eq:s-analytic-total-sum}, and \eqref{eq:Vieta-total-sum-expansion}, we obtain
		\[
		\delta_1+\delta_2=\left[ -\frac{3\chi_0-L_0}{\mu_0} - \kappa_1 \right]\eta + O(|\eta|^2).
		\]
		Therefore, \(s_1=-\frac{3\chi_0-L_0}{\mu_0} - \kappa_1\). Substituting \eqref{eq:kappa1-third-root} gives
		\begin{align}
			s_1
			&=
			-\frac{3\chi_0-L_0}{\mu_0}
			-
			\frac{
				L_0^3
				-
				3\chi_0L_0^2
				+
				D_1L_0
				-
				D_0
			}{
				\mu_0L_0^2
			}
			\nonumber\\
			&=
			\frac{
				-(3\chi_0-L_0)L_0^2
				-
				L_0^3
				+
				3\chi_0L_0^2
				-
				D_1L_0
				+
				D_0
			}{
				\mu_0L_0^2
			}
			\nonumber\\
			&=
			\frac{D_0-D_1L_0}
			{\mu_0L_0^2}.
			\label{eq:s1-D0-D1}
		\end{align}
		Since
		\[
		p_1=\frac{D_0}{\mu_0L_0},
		\]
		equation \eqref{eq:s1-D0-D1} becomes
		\[
		s_1=\frac{p_1-D_1/\mu_0}{L_0}.
		\]
		This proves the first expression for \(s_1\).
		
		The identity \eqref{eq:L0-Lambda2} gives \(2\mu_0L_0=-D_0\Lambda''(\chi_0)\). Therefore,
		\[
		p_1=\frac{D_0}{\mu_0L_0}=-\frac{2}{\Lambda''(\chi_0)}.
		\]
		
		To derive the alternative expression for \(s_1\), differentiate the identity \(F(\chi;\mu)=D(\chi)\bigl[\mu-\Lambda(\chi)\bigr]\) three times with respect to \(\chi\), and evaluate it at \((\chi,\mu)=(\chi_0,\mu_0)\). Since \(\mu_0=\Lambda(\chi_0)\) and \(\Lambda'(\chi_0)=0\), one obtains \(F_{\chi\chi\chi}(\chi_0;\mu_0)=-D_0\Lambda'''(\chi_0) - 3D_1\Lambda''(\chi_0)\). On the other hand, the cubic form of \(F\) gives \(F_{\chi\chi\chi}(\chi_0;\mu_0)=6\mu_0\). Hence \(6\mu_0=-D_0\Lambda'''(\chi_0) - 3D_1\Lambda''(\chi_0)\). Combining this identity with \(2\mu_0L_0=-D_0\Lambda''(\chi_0)\), and the first formula for \(s_1\), yields
		\[
		s_1=-\frac{2\Lambda'''(\chi_0)} {3[\Lambda''(\chi_0)]^2}.
		\]
		
		Finally, Taylor expansion of the dispersion map gives
		\begin{align*}
			\eta
			&=
			\Lambda(\chi_0+\delta)
			-
			\Lambda(\chi_0)
			\\
			&=
			\frac{\Lambda''(\chi_0)}{2}\delta^2
			+
			\frac{\Lambda'''(\chi_0)}{6}\delta^3
			+
			O(\delta^4).
		\end{align*}
		After choosing a local square-root branch, substitute \(\delta=\alpha\eta^{1/2} + \beta\eta + O(|\eta|^{3/2})\). Matching the coefficients of \(\eta\) and \(\eta^{3/2}\) gives
		\[
		\alpha^2=\frac{2}{\Lambda''(\chi_0)}, \qquad \beta=-\frac{\Lambda'''(\chi_0)} {3[\Lambda''(\chi_0)]^2}.
		\]
		The two choices of the square root give \eqref{eq:Puiseux-double}.
	\end{proof}
	
	For a prescribed mixed-sector coefficient \(M_1\), choose \(C_1\) and \(C_2\) so that \(C_1+C_2=1\) and \(C_1\delta_1+C_2\delta_2=M_1\). For all sufficiently small nonzero \(\eta\), \(\delta_1-\delta_2\neq0\), and hence
	\begin{equation}
		\label{eq:C12-double}
		C_1
		=
		\frac{M_1-\delta_2}{\delta_1-\delta_2},
		\qquad
		C_2
		=
		\frac{\delta_1-M_1}{\delta_1-\delta_2}.
	\end{equation}
	
	For \(j=1,2\), let \({\bf y}_j(\eta)={\bf y}(\chi_j(\eta),\lambda(\eta))\). Then
	\begin{align*}
		C_1{\bf y}_1+C_2{\bf y}_2
		={}&
		\frac{
			\delta_1{\bf y}_2-\delta_2{\bf y}_1
		}{
			\delta_1-\delta_2
		}
		+
		M_1
		\frac{
			{\bf y}_1-{\bf y}_2
		}{
			\delta_1-\delta_2
		}.
	\end{align*}
	The first term tends to \({\bf Y}_0\), while the second is a divided difference tending to \(M_1{\bf Y}_1\). Therefore, \(C_1{\bf y}_1+C_2{\bf y}_2 \longrightarrow {\bf Y}_0+M_1{\bf Y}_1\).
	
	Since \(M_1\neq0\) in the generic mixed sector, \(C_1,C_2=O(|\eta|^{-1/2})\). The divergent parts cancel on the common limiting eigenvector \({\bf Y}_0\), leaving the finite associated-vector contribution \(M_1{\bf Y}_1\).
	
	Define the weighted moments
	\[
	\mathfrak M_n=C_1\delta_1^n+C_2\delta_2^n,
	\qquad n\ge0.
	\]
	The following strengthened estimate is uniform in the moment index and will justify the complete analytic Taylor tail.
	
	\begin{lemma}
		\label{lem:double-weighted-moment-tail}
		The prescribed moments satisfy
		\begin{align}
			\mathfrak M_0&=1,
			&\mathfrak M_1&=M_1,
			\notag\\
			\mathfrak M_2
			&=
			\mathcal M_2\eta+O(|\eta|^2),
			&
			\mathfrak M_3
			&=
			\mathcal M_3\eta+O(|\eta|^2),
			\label{eq:double-low-moment-expansion}
		\end{align}
		where
		\[
		\mathcal M_2=s_1M_1-p_1,
		\qquad
		\mathcal M_3=-p_1M_1.
		\]
		Moreover, for every \(\tau>0\), after reducing \(|\eta|\) if necessary, there is a constant \(C_\tau>0\) such that
		\begin{equation}
			\label{eq:double-uniform-moment-tail}
			|\mathfrak M_n|
			\le
			C_\tau|\eta|^2\tau^{n-4},
			\qquad n\ge4.
		\end{equation}
	\end{lemma}
	
	\begin{proof}
		Set
		\[
		s(\eta)=\delta_1+\delta_2,
		\qquad
		p(\eta)=\delta_1\delta_2.
		\]
		Lemma~\ref{lem:moment-double} gives
		\[
		s(\eta)=s_1\eta+O(|\eta|^2),
		\qquad
		p(\eta)=p_1\eta+O(|\eta|^2).
		\]
		Because \(\delta_1\) and \(\delta_2\) are the roots of
		\(z^2-s(\eta)z+p(\eta)=0\), multiplication by
		\(C_j\delta_j^{n-2}\) and summation over \(j=1,2\) gives the exact recurrence
		\begin{equation}
			\label{eq:double-exact-moment-recurrence}
			\mathfrak M_n
			=
			s(\eta)\mathfrak M_{n-1}
			-
			p(\eta)\mathfrak M_{n-2},
			\qquad n\ge2.
		\end{equation}
		The defining equations for \(C_1,C_2\) give
		\(\mathfrak M_0=1\) and \(\mathfrak M_1=M_1\). Hence
		\[
		\begin{aligned}
			\mathfrak M_2
			&=s(\eta)M_1-p(\eta)
			=(s_1M_1-p_1)\eta+O(|\eta|^2),
			\\
			\mathfrak M_3
			&=s(\eta)\mathfrak M_2-p(\eta)M_1
			=-p_1M_1\eta+O(|\eta|^2),
		\end{aligned}
		\]
		which proves \eqref{eq:double-low-moment-expansion}. Applying \eqref{eq:double-exact-moment-recurrence} twice more yields
		\[
		\mathfrak M_4
		=s(\eta)\mathfrak M_3-p(\eta)\mathfrak M_2
		=O(|\eta|^2),
		\]
		and
		\[
		\mathfrak M_5
		=s(\eta)\mathfrak M_4-p(\eta)\mathfrak M_3
		=O(|\eta|^2).
		\]
		The recurrence \eqref{eq:double-exact-moment-recurrence} has the form of \eqref{eq:abstract-moment-recurrence} with
		\[
		N=4,
		\qquad d=2,
		\qquad b_1(\eta)=s(\eta)=O(|\eta|),
		\qquad b_2(\eta)=-p(\eta)=O(|\eta|).
		\]
		Lemma~\ref{lem:analytic-moment-tail} therefore gives the uniform estimate \eqref{eq:double-uniform-moment-tail}.
	\end{proof}
	
	\subsection{The general finite correction vector}
	\label{subsec:double-correction}
	
	\label{subsec:general-double-correction}
	
	Let \(\chi_s(\eta)=\chi_0+\delta_3(\eta)\) be the analytic simple spatial branch satisfying \(\chi_s(0)=\chi_s\). By the holomorphic implicit function theorem, \(\chi_s(\eta)=\chi_s+\kappa_s\eta+O(|\eta|^2)\), where
	\begin{equation*}
		\kappa_s
		=
		-\frac{
			D(\chi_s)
		}{
			F_\chi(\chi_s;\mu_0)
		}
		=
		-\frac{
			D(\chi_s)
		}{
			\mu_0L_0^2
		}
		=
		-\frac{2}{\mu_0^2}-s_1.
	\end{equation*}
	
	Let \({\bf y}_s(\eta)\) denote the elementary eigenfunction along this simple branch: \({\bf y}_s(\eta)={\bf y}(\chi_s(\eta),\lambda(\eta))\). When \(\chi_s=0\), the right-hand side is understood using the zero-branch expression \eqref{eq:zero-branch-eigenfunction}.
	
	Define the total simple-branch correction \({\bf Y}_{s,\mu}=\left. \frac{d}{d\eta} {\bf y}_s(\eta) \right|_{\eta=0}\). For \(\chi_s\neq0\), this vector is
	\begin{equation*}
		{\bf Y}_{s,\mu}
		=
		\frac{1}{2\lambda_0}
		\begin{pmatrix}
			E_s\\
			0\\
			0
		\end{pmatrix}
		+
		\kappa_s
		\left.
		\frac{\partial}{\partial\chi}
		{\bf y}(\chi,\lambda_0)
		\right|_{\chi=\chi_s},
		\qquad
		E_s=E(\chi_s).
	\end{equation*}
	For the critical zero branch \(\chi_s=0\), one instead has
	\begin{equation*}
		{\bf Y}_{s,\mu}
		=
		e^{{\rm i}\mu_0t/2}
		\left[
		\frac{1}{2\lambda_0}
		\begin{pmatrix}
			1\\0\\0
		\end{pmatrix}
		+
		\frac{{\rm i}t}{2}
		\begin{pmatrix}
			\lambda_0\\
			u_{1,0}\\
			u_{2,0}
		\end{pmatrix}
		\right].
	\end{equation*}
	
	We now combine the two coalescing modes with the analytic continuation of the remaining simple mode. Consider the split-root eigenfunction
	\begin{equation}
		\label{eq:y-eta-double-general}
		{\bf y}(\eta)
		=
		C_1{\bf y}(\chi_1,\lambda(\eta))
		+
		C_2{\bf y}(\chi_2,\lambda(\eta))
		+
		c_s{\bf y}_s(\eta),
	\end{equation}
	where \(C_1\) and \(C_2\) are given by \eqref{eq:C12-double}, while the simple-branch coefficient \(c_s\) is held fixed as \(\eta\to0\). This fixed-moment, fixed-simple-coefficient choice is the canonical lifting used throughout the remainder of the section.
	
	For the phase derivatives appearing below, we use the notation
	\begin{equation}
		\label{eq:En-phase-derivatives}
		E_n
		=
		\left.
		\partial_\chi^nE(\chi)
		\right|_{\chi=\chi_0},
		\qquad n\ge0.
	\end{equation}
	
	\begin{theorem}
		\label{thm:y-eta-C2-local}
		The eigenfunction \eqref{eq:y-eta-double-general} satisfies
		\begin{equation}
			\label{eq:y-eta-expansion}
			{\bf y}(\eta)
			=
			{\bf y}_R
			+
			\eta{\bf y}_\mu
			+
			O_{C^2_{\rm loc}}(|\eta|^2),
		\end{equation}
		where \({\bf y}_R={\bf Y}_0 + M_1{\bf Y}_1 + c_s{\bf Y}_s\), and
		\begin{equation}
			\label{eq:ymu-double}
			\begin{aligned}
				{\bf y}_\mu
				={}&
				\frac{1}{2\lambda_0}
				\begin{pmatrix}
					E_0+M_1E_1\\
					0\\
					0
				\end{pmatrix}
				+
				\frac{\mathcal M_2}{2}{\bf Y}_2
				+
				\frac{\mathcal M_3}{6}{\bf Y}_3
				+
				c_s{\bf Y}_{s,\mu},
			\end{aligned}
		\end{equation}
		with \(\mathcal M_2=s_1M_1-p_1\) and \(\mathcal M_3=-p_1M_1\).
		
		More precisely, for every compact set \(\Omega\subset\mathbb R^2\), there exist constants \(C_\Omega>0\) and \(\eta_\Omega>0\) such that
		\begin{equation*}
			\sup_{(x,t)\in\Omega}
			\left\|
			\partial_x^{\alpha_1}
			\partial_t^{\alpha_2}
			\left[
			{\bf y}(\eta)
			-
			{\bf y}_R
			-
			\eta{\bf y}_\mu
			\right]
			\right\|
			\le
			C_\Omega|\eta|^2
		\end{equation*}
		for all \(\alpha_1,\alpha_2\in\mathbb N_0\) and \(\alpha_1+\alpha_2\le2\), and all sufficiently small nonzero \(\eta\).
	\end{theorem}
	
	\begin{proof}
		Fix a compact set \(\Omega\subset\mathbb R^2\), and set
		\[
		\mathcal B=C^2(\Omega;\mathbb C^3)
		\]
		with its standard maximum norm over all \(x,t\) derivatives of total order at most two. Choose \(r_0>0\) so that the closed disk \(|\delta|\le r_0\) centered at \(\chi_0\) does not meet \(\{0,\pm k\}\). The map
		\[
		f(\delta)
		={\bf y}(\chi_0+\delta,\lambda_0)
		\]
		is holomorphic from \(|\delta|<r_0\) into \(\mathcal B\), and its Taylor coefficients are
		\[
		f_n=\frac{{\bf Y}_n}{n!}.
		\]
		Therefore the fixed-\(\lambda_0\) part of the two-branch combination has the convergent expansion
		\begin{equation}
			\label{eq:double-fixed-lambda-full-series}
			C_1f(\delta_1)+C_2f(\delta_2)
			=
			\sum_{n=0}^{\infty}
			\frac{\mathfrak M_n}{n!}{\bf Y}_n
			\qquad\text{in }\mathcal B.
		\end{equation}
		By \eqref{eq:double-low-moment-expansion}, the terms with \(n=0,1,2,3\) equal
		\[
		{\bf Y}_0+M_1{\bf Y}_1
		+
		\eta
		\left(
		\frac{\mathcal M_2}{2}{\bf Y}_2
		+
		\frac{\mathcal M_3}{6}{\bf Y}_3
		\right)
		+O_{\mathcal B}(|\eta|^2).
		\]
		For the infinite tail in \eqref{eq:double-fixed-lambda-full-series}, choose \(\tau\in(0,r_0)\). Lemma~\ref{lem:double-weighted-moment-tail} supplies the uniform geometric bound \eqref{eq:double-uniform-moment-tail}, and Lemma~\ref{lem:analytic-moment-tail}, applied with \(N=4\), gives
		\[
		\left\|
		\sum_{n=4}^{\infty}
		\frac{\mathfrak M_n}{n!}{\bf Y}_n
		\right\|_{\mathcal B}
		=O(|\eta|^2).
		\]
		This proves, in the full \(C^2(\Omega)\) norm,
		\begin{equation}
			\label{eq:double-fixed-lambda-C2-expansion}
			\begin{aligned}
				&C_1{\bf y}(\chi_1,\lambda_0)
				+C_2{\bf y}(\chi_2,\lambda_0)
				\\
				&\qquad=
				{\bf Y}_0+M_1{\bf Y}_1
				+
				\eta
				\left(
				\frac{\mathcal M_2}{2}{\bf Y}_2
				+
				\frac{\mathcal M_3}{6}{\bf Y}_3
				\right)
				+O_{\mathcal B}(|\eta|^2).
			\end{aligned}
		\end{equation}
		
		It remains to include the explicit spectral variation. From \eqref{eq:y-basic},
		\[
		{\bf y}(\chi,\lambda(\eta))
		-
		{\bf y}(\chi,\lambda_0)
		=
		\bigl[\lambda(\eta)-\lambda_0\bigr]
		\begin{pmatrix}
			E(\chi)\\0\\0
		\end{pmatrix}.
		\]
		The scalar phase \(E(\chi_0+\delta)\) is also holomorphic as a \(C^2(\Omega)\)-valued function. Its moment expansion, together with \(\mathfrak M_2,\mathfrak M_3=O(|\eta|)\) and the uniformly controlled tail, gives
		\[
		C_1E(\chi_1)+C_2E(\chi_2)
		=
		E_0+M_1E_1+O_{C^2(\Omega)}(|\eta|).
		\]
		Since
		\[
		\lambda(\eta)-\lambda_0
		=
		\frac{\eta}{2\lambda_0}+O(|\eta|^2),
		\]
		the complete spectral-variation contribution is
		\[
		\frac{\eta}{2\lambda_0}
		\begin{pmatrix}
			E_0+M_1E_1\\0\\0
		\end{pmatrix}
		+O_{\mathcal B}(|\eta|^2).
		\]
		
		Finally, if \(\chi_s\ne0\), the implicit-function branch \(\chi_s(\eta)\), the square-root branch \(\lambda(\eta)\), and the elementary eigenfunction are jointly holomorphic in \(\eta\) as \(\mathcal B\)-valued functions. Hence
		\[
		{\bf y}_s(\eta)
		={\bf Y}_s+
		\eta{\bf Y}_{s,\mu}
		+O_{\mathcal B}(|\eta|^2).
		\]
		When \(\chi_s=0\), the same estimate follows directly from the separate analytic formula \eqref{eq:zero-branch-eigenfunction}. Multiplying by the fixed coefficient \(c_s\) and adding this expansion to \eqref{eq:double-fixed-lambda-C2-expansion} proves \eqref{eq:y-eta-expansion}--\eqref{eq:ymu-double}. Since \(\Omega\) was arbitrary, the remainder is \(O_{C^2_{\rm loc}}(|\eta|^2)\).
	\end{proof}
	
	\subsection{The directional real-spectrum limit}
	\label{subsec:double-directional-limit}
	
	Having identified the finite leading and first-order vectors, we approach the real spectral value \(\mu_0\) along
	\begin{equation*}
		\eta=\varepsilon q,
		\qquad
		\varepsilon\in\mathbb R,
		\qquad
		\varepsilon\to0,
		\qquad
		q\in\mathbb C\setminus\mathbb R.
	\end{equation*}
	The restriction \(q\notin\mathbb R\) is essential. For real \(q\), both \(\mu_0+\varepsilon q\) and the local square-root branch \(\lambda(\eta)\) remain real, so \(\lambda(\eta)^*-\lambda(\eta)=0\) identically and the ordinary nonreal-spectrum Darboux transformation cannot be used as the approximating family.
	
	\begin{theorem}
		\label{thm:real-double-directional}
		Let \(\chi_0\) satisfy \eqref{eq:exact-real-double}. Let \(M_1,c_s\in\mathbb C\) satisfy
		\begin{equation}
			\label{eq:generic-double-theorem-assumption}
			b_J|M_1|^2
			+
			2a_J\operatorname{Re}M_1
			+
			d_s|c_s|^2
			=
			0.
		\end{equation}
		Define \({\bf y}_R={\bf Y}_0 + M_1{\bf Y}_1 + c_s{\bf Y}_s\), and let \({\bf y}_\mu\) be given by \eqref{eq:ymu-double}.
		
		For any fixed approach direction \(q\in\mathbb C\setminus\mathbb R\), define
		\begin{equation}
			\label{eq:mq-double}
			\begin{aligned}
				m_q
				={}&
				\frac{
					2\lambda_0
					\left[
					q^*
					\langle{\bf y}_\mu|J|{\bf y}_R\rangle
					+
					q
					\langle{\bf y}_R|J|{\bf y}_\mu\rangle
					\right]
				}{
					q^*-q
				}
				\\
				&-
				\frac{
					\langle{\bf y}_R|K|{\bf y}_R\rangle
				}{
					2\lambda_0
				}.
			\end{aligned}
		\end{equation}
		Writing \({\bf y}_R=(\psi_R,\phi_R,\varphi_R)^T\), the limiting field
		\begin{equation}
			\label{eq:real-double-solution-q}
			\begin{pmatrix}
				u_1^{(R,2;q)}\\
				u_2^{(R,2;q)}
			\end{pmatrix}
			=
			\begin{pmatrix}
				u_{1,0}\\
				u_{2,0}
			\end{pmatrix}
			+
			2
			\frac{
				\begin{pmatrix}
					\phi_R\\
					\varphi_R
				\end{pmatrix}
				\psi_R^*
			}{
				m_q
			}.
		\end{equation}
		Then
		\begin{equation}
			\label{eq:mq-negative-real-part}
			\operatorname{Re}m_q
			=
			-
			\frac{
				\langle{\bf y}_R|K|{\bf y}_R\rangle
			}{
				2\lambda_0
			}
			<0.
		\end{equation}
		Consequently, \(m_q\) has no zeros on the real \((x,t)\)-plane, and \eqref{eq:real-double-solution-q} is a globally smooth exact solution of \eqref{eq:CFL}.
	\end{theorem}
	
	\begin{proof}
		Theorem~\ref{thm:y-eta-C2-local} gives the expansion \eqref{eq:unified-eigenfunction-expansion} with the vectors \({\bf y}_R\) and \({\bf y}_\mu\) specified above. The constraint \eqref{eq:generic-double-theorem-assumption} is exactly the nullity condition \eqref{eq:unified-leading-nullity}. By Proposition~\ref{prop:confluent-basis-independence}, the normalized vector \({\bf y}_R={\bf Y}_0+M_1{\bf Y}_1+c_s{\bf Y}_s\) is nontrivial. Hence all hypotheses of Theorem~\ref{thm:unified-real-spectrum-directional-limit} are satisfied.
		
		The universal denominator \eqref{eq:unified-directional-denominator} is precisely \(m_q\) in \eqref{eq:mq-double}, and the universal limiting field \eqref{eq:unified-directional-field} is \eqref{eq:real-double-solution-q}. Equations \eqref{eq:mq-negative-real-part}, global nonvanishing of \(m_q\), and exact solvability therefore follow directly from \eqref{eq:unified-negative-real-part} and Theorem~\ref{thm:unified-real-spectrum-directional-limit}.
	\end{proof}
	
	\begin{remark}
		Multiplying \(q\) by a nonzero real constant does not change \(m_q\). Therefore the limiting solution depends only on the direction of approach in the complex \(\mu\)-plane. We may normalize \(q=e^{{\rm i}\alpha}\) and \(\sin\alpha\neq0\). If \(z=\langle{\bf y}_R|J|{\bf y}_\mu\rangle=z_{\rm R}+{\rm i}z_{\rm I}\), then
		\begin{equation*}
			m_\alpha
			=
			2{\rm i}\lambda_0
			\left(
			z_{\rm R}\cot\alpha-z_{\rm I}
			\right)
			-
			\frac{
				\langle{\bf y}_R|K|{\bf y}_R\rangle
			}{
				2\lambda_0
			}.
		\end{equation*}
		Thus the approach direction modifies only the imaginary part of the limiting denominator; its real part remains strictly negative. Also, \(m_{-q}=m_q\), so opposite orientations along the same line in the \(\mu\)-plane give the same limiting solution.
	\end{remark}
	
	\begin{corollary}
		\label{cor:real-double-imag}
		For the imaginary-axis approach
		\[
		\eta={\rm i}\varepsilon, \qquad \varepsilon\in\mathbb R, \qquad \varepsilon\to0,
		\]
		one has \(q={\rm i}\), and
		\begin{equation}
			\label{eq:mR-double}
			m_R
			=
			\lambda_0
			\left(
			\langle{\bf y}_\mu|J|{\bf y}_R\rangle
			-
			\langle{\bf y}_R|J|{\bf y}_\mu\rangle
			\right)
			-
			\frac{
				\langle{\bf y}_R|K|{\bf y}_R\rangle
			}{
				2\lambda_0
			}.
		\end{equation}
		The corresponding solution is
		\begin{equation}
			\label{eq:real-double-solution}
			\begin{pmatrix}
				u_1^{(R,2)}\\
				u_2^{(R,2)}
			\end{pmatrix}
			=
			\begin{pmatrix}
				u_{1,0}\\
				u_{2,0}
			\end{pmatrix}
			+
			2
			\frac{
				\begin{pmatrix}
					\phi_R\\
					\varphi_R
				\end{pmatrix}
				\psi_R^*
			}{
				m_R
			}.
		\end{equation}
		Moreover,
		\[
		\operatorname{Re}m_R=-\frac{ \langle{\bf y}_R|K|{\bf y}_R\rangle }{ 2\lambda_0 } <0,
		\]
		so \eqref{eq:real-double-solution} is globally smooth.
	\end{corollary}
	
	\begin{proof}
		Substitute
		\[
		q={\rm i}, \qquad q^*=-{\rm i}, \qquad q^*-q=-2{\rm i}
		\]
		into \eqref{eq:mq-double}. The result is \eqref{eq:mR-double}. The remaining conclusions follow directly from Theorem~\ref{thm:real-double-directional}.
	\end{proof}
	
	\begin{corollary}
		\label{cor:pure-double-sector}
		If \(c_s=0\), then the general null condition \eqref{eq:double-null-general} reduces to
		\begin{equation}
			\label{eq:M1-circle-double}
			\left|
			M_1+\frac{a_J}{b_J}
			\right|
			=
			\left|
			\frac{a_J}{b_J}
			\right|.
		\end{equation}
		For every \(M_1\) on this circle, the leading and correction vectors become \({\bf y}_R={\bf Y}_0+M_1{\bf Y}_1\), and
		\[
		{\bf y}_\mu=\frac{1}{2\lambda_0}
		\begin{pmatrix}
			E_0+M_1E_1\\
			0\\
			0
		\end{pmatrix}
		+ \frac{s_1M_1-p_1}{2}{\bf Y}_2 - \frac{p_1M_1}{6}{\bf Y}_3.
		\]
		The limiting Darboux formula is obtained from Theorem~\ref{thm:real-double-directional} by setting \(c_s=0\).
	\end{corollary}
	
	\begin{corollary}
		\label{cor:elementary-double-sector}
		If \(M_1=0\), then the null condition forces \(c_s=0\). The only admissible leading vector in this case is \({\bf y}_R={\bf Y}_0\), and
		\[
		{\bf y}_\mu=\frac{1}{2\lambda_0}
		\begin{pmatrix}
			E_0\\
			0\\
			0
		\end{pmatrix}
		- \frac{p_1}{2}{\bf Y}_2.
		\]
	\end{corollary}
	
	\begin{remark}
		\label{rem:double-special-coefficients}
		In the pure double-root sector, the two real points of \eqref{eq:M1-circle-double} are
		\[
		M_1=0, \qquad M_1=-\frac{2a_J}{b_J}.
		\]
		The second value gives the canonical nontrivial real confluent representative used in the explicit rational formulas below.
		
		In the mixed sector, the maximal simple-branch amplitude is
		\[
		|c_s|^2=\frac{a_J^2}{b_Jd_s}.
		\]
		At this saturated boundary, the admissible \(M_1\)-circle collapses to
		\[
		M_1=-\frac{a_J}{b_J}.
		\]
		Both cases are direct specializations of the general mixed formula and require no separate limiting argument.
	\end{remark}
	
	\subsection{The pure double-root null-circle family}
	\label{subsec:pure-double-family}
	
	\label{subsec:explicit-rational-double}
	
	We now specialize the real-double-root construction to the pure sector \(c_s=0\), for which the limiting fields can be written explicitly as background-normalized rational functions. Let \({\bf y}_R={\bf Y}_0+r{\bf Y}_1\), where \(r\in\mathbb C\setminus\{0\}\) satisfies the pure-sector null condition
	\begin{equation}
		\label{eq:r-general-null-condition}
		b_J|r|^2
		+
		2a_J\operatorname{Re}r
		=
		0.
	\end{equation}
	Equivalently, the complete nonzero null circle is parameterized by \(r=r(\vartheta)=\frac{a_J}{b_J}\left(e^{{\rm i}\vartheta}-1\right)\), where \(\vartheta\in(0,2\pi)\). The value \(\vartheta=\pi\) gives the canonical nontrivial real representative \(r_{\rm c}=-\frac{2a_J}{b_J}\in\mathbb R\setminus\{0\}\).
	
	Fix one admissible real exact double branch \(\chi_0\) satisfying
	\[
	\chi_0\notin\{0,\pm k\}, \qquad \mu_0=\Lambda(\chi_0)>0, \qquad \Lambda''(\chi_0)\neq0,
	\]
	and choose \(\lambda_0=\sqrt{\mu_0}>0\). All quantities \(\chi_s\), \(L_0\), \(D_0\), \(D_1\), \(p_1\), and \(s_1\) are evaluated at this selected branch.
	
	The family is determined by the continuous parameters \(S\), \(R\), \(k\), \(\alpha\), and \(\vartheta\), together with the discrete choice of an admissible real double branch \(\chi_0\).
	
	The two moment coefficients entering the first correction vector are \(M_{2,\eta}=s_1r-p_1\) and \(M_{3,\eta}=-p_1r\).
	
	For later use, set \(\Gamma_\nu=2-kR\), and define
	\begin{equation*}
		\begin{aligned}
			\nu_0
			&=
			\frac{S}{2}
			+
			\frac{\Gamma_\nu}{2\chi_0},
			\\
			\nu_1
			&=
			-\frac{\Gamma_\nu}{2\chi_0^2},
			\\
			\nu_2
			&=
			\frac{\Gamma_\nu}{\chi_0^3},
			\\
			\nu_3
			&=
			-\frac{3\Gamma_\nu}{\chi_0^4}.
		\end{aligned}
	\end{equation*}
	Introduce the characteristic coordinate \(\xi=x+\nu_1t\), and the exponential derivative polynomials
	\begin{equation*}
		\begin{aligned}
			\mathcal E_0&=1,
			\\
			\mathcal E_1&={\rm i}\xi,
			\\
			\mathcal E_2&=-\xi^2+{\rm i}\nu_2t,
			\\
			\mathcal E_3
			&=
			-{\rm i}\xi^3
			-
			3\nu_2\xi t
			+
			{\rm i}\nu_3t.
		\end{aligned}
	\end{equation*}
	They satisfy \(\left. \frac{\partial^nE}{\partial\chi^n} \right|_{\chi=\chi_0}=E_0\mathcal E_n\) for \(n=0,1,2,3\), where
	\[
	E_0=\exp \left[ {\rm i} \left( \chi_0x+\nu_0t \right) \right].
	\]
	
	For \(j=1,2\), define
	\begin{equation*}
		h_{1,n}
		=
		\frac{
			(-1)^n n!k
		}{
			(k+\chi_0)^{n+1}
		},
		\qquad
		h_{2,n}
		=
		\frac{
			n!k
		}{
			(k-\chi_0)^{n+1}
		},
		\qquad
		n=0,1,2,3.
	\end{equation*}
	In particular,
	\[
	h_{1,0}=\frac{k}{k+\chi_0}, \qquad h_{2,0}=\frac{k}{k-\chi_0}.
	\]
	Define
	\begin{equation*}
		\begin{aligned}
			\mathcal G_{j,1}
			&=
			h_{j,1}
			+
			h_{j,0}\mathcal E_1,
			\\
			\mathcal G_{j,2}
			&=
			h_{j,2}
			+
			2h_{j,1}\mathcal E_1
			+
			h_{j,0}\mathcal E_2,
			\\
			\mathcal G_{j,3}
			&=
			h_{j,3}
			+
			3h_{j,2}\mathcal E_1
			+
			3h_{j,1}\mathcal E_2
			+
			h_{j,0}\mathcal E_3.
		\end{aligned}
	\end{equation*}
	Thus
	\[
	\left. \partial_\chi^n \bigl[ h_j(\chi)E(\chi) \bigr] \right|_{\chi=\chi_0}=E_0\mathcal G_{j,n}.
	\]
	
	Define the polynomial components of the leading vector by
	\begin{equation*}
		\begin{aligned}
			P_0
			&=
			\lambda_0
			\left(
			1+r\mathcal E_1
			\right),
			\\
			P_j
			&=
			h_{j,0}
			+
			r\mathcal G_{j,1},
			\qquad
			j=1,2.
		\end{aligned}
	\end{equation*}
	Equivalently,
	\begin{equation}
		\label{eq:Pj-linear-explicit-rational}
		\begin{aligned}
			P_0
			&=
			\lambda_0
			\left(
			1+{\rm i}r\xi
			\right),
			\\
			P_j
			&=
			h_{j,0}
			+
			rh_{j,1}
			+
			{\rm i}rh_{j,0}\xi,
			\qquad
			j=1,2.
		\end{aligned}
	\end{equation}
	The leading eigenvector is
	\begin{equation*}
		{\bf y}_R
		=
		E_0
		\begin{pmatrix}
			P_0\\
			u_{1,0}P_1\\
			u_{2,0}P_2
		\end{pmatrix}.
	\end{equation*}
	
	The polynomial components of the correction vector are
	\begin{equation}
		\label{eq:Qj-explicit-rational}
		\begin{aligned}
			Q_0
			={}&
			\frac{1}{2\lambda_0}
			\left(
			1+r\mathcal E_1
			\right)
			+
			\frac{\lambda_0M_{2,\eta}}{2}\mathcal E_2
			+
			\frac{\lambda_0M_{3,\eta}}{6}\mathcal E_3,
			\\
			Q_j
			={}&
			\frac{M_{2,\eta}}{2}\mathcal G_{j,2}
			+
			\frac{M_{3,\eta}}{6}\mathcal G_{j,3},
			\qquad
			j=1,2.
		\end{aligned}
	\end{equation}
	Thus
	\begin{equation*}
		{\bf y}_\mu
		=
		E_0
		\begin{pmatrix}
			Q_0\\
			u_{1,0}Q_1\\
			u_{2,0}Q_2
		\end{pmatrix}.
	\end{equation*}
	
	Define
	\begin{equation}
		\label{eq:Z-explicit-rational}
		\mathcal Z(x,t)
		=
		P_0^*Q_0
		-
		AP_1^*Q_1
		-
		BP_2^*Q_2.
	\end{equation}
	
	For the explicit polynomial coefficients, introduce the real constants
	\begin{equation*}
		\begin{aligned}
			c_J
			&=
			-Ah_{1,0}h_{1,2}
			-
			Bh_{2,0}h_{2,2},
			\\
			d_J
			&=
			-Ah_{1,0}h_{1,3}
			-
			Bh_{2,0}h_{2,3},
			\\
			e_J
			&=
			-Ah_{1,1}h_{1,2}
			-
			Bh_{2,1}h_{2,2},
			\\
			f_J
			&=
			-Ah_{1,1}h_{1,3}
			-
			Bh_{2,1}h_{2,3}.
		\end{aligned}
	\end{equation*}
	The root-moment coefficients satisfy the identity
	\begin{equation}
		\label{eq:moment-Gram-identity-general-r}
		1
		+
		a_Js_1
		+
		b_Jp_1
		-
		c_Jp_1
		=
		0.
	\end{equation}
	
	\begin{lemma}
		\label{lem:aJ-p1-simplification}
		At every exact real double branch point,
		\begin{equation}
			\label{eq:aJ-p1-chi0}
			a_Jp_1=\chi_0.
		\end{equation}
	\end{lemma}
	
	\begin{proof}
		The identities already obtained in Lemma~\ref{lem:double-Gram-structure} and Lemma~\ref{lem:moment-double} are
		\[
		a_J
		=-\frac{\chi_0}{2}\Lambda''(\chi_0),
		\qquad
		p_1
		=-\frac{2}{\Lambda''(\chi_0)}.
		\]
		Because the root is exactly double, \(\Lambda''(\chi_0)\neq0\). Multiplication of the two identities gives \eqref{eq:aJ-p1-chi0}.
	\end{proof}
	
	Write \(r=u+{\rm i}v\), where \(u=\operatorname{Re}r\), \(v=\operatorname{Im}r\), and \(\varrho=|r|^2\). The following calculation identifies the complete polynomial structure of the mixed \(J\)-pairing that enters the limiting denominator.
	
	\begin{lemma}
		\label{lem:Z-general-complex-r}
		Let \(r\neq0\) satisfy \eqref{eq:r-general-null-condition}. Then
		\begin{equation*}
			\mathcal Z(x,t)
			=
			\zeta_0
			+
			{\rm i}
			\left(
			\gamma_3\xi^3
			+
			\gamma_2\xi^2
			+
			\gamma_1\xi
			+
			\gamma_t t
			+
			\gamma_0
			\right),
		\end{equation*}
		where all coefficients are real and
		\begin{equation}
			\label{eq:gamma3-general-r}
			\gamma_3
			=
			-\frac{\chi_0}{3}\varrho,
		\end{equation}
		\begin{equation}
			\label{eq:gamma2-general-r}
			\gamma_2
			=
			\chi_0v,
		\end{equation}
		\begin{equation*}
			\begin{aligned}
				\gamma_1
				={}&
				-\chi_0
				+
				a_Js_1u
				-
				b_Jp_1u
				+
				b_Js_1\varrho
				\\
				&-
				\frac{c_Js_1}{2}\varrho
				+
				\frac{d_Jp_1}{6}\varrho
				-
				\frac{e_Jp_1}{2}\varrho,
			\end{aligned}
		\end{equation*}
		\begin{equation*}
			\gamma_t
			=
			\frac{a_J\varrho}{6}
			\left(
			3\nu_2s_1-\nu_3p_1
			\right),
		\end{equation*}
		\begin{equation*}
			\gamma_0
			=
			v
			\left(
			\frac{c_Js_1}{2}
			-
			\frac{d_Jp_1}{6}
			+
			\frac{e_Jp_1}{2}
			\right),
		\end{equation*}
		and
		\begin{equation}
			\label{eq:zeta0-general-r}
			\begin{aligned}
				\zeta_0
				={}&
				\frac12
				-
				\frac{c_Jp_1}{2}
				+
				\frac{c_Js_1u}{2}
				-
				\frac{d_Jp_1u}{6}
				-
				\frac{e_Jp_1u}{2}
				\\
				&+
				\frac{e_Js_1}{2}\varrho
				-
				\frac{f_Jp_1}{6}\varrho.
			\end{aligned}
		\end{equation}
		In particular, \(\gamma_3\neq0\).
	\end{lemma}
	
	\begin{proof}
		Substitution of \eqref{eq:Pj-linear-explicit-rational} and \eqref{eq:Qj-explicit-rational} into \eqref{eq:Z-explicit-rational} initially produces terms of total degree at most four.
		
		The quartic term is proportional to \(\langle{\bf Y}_0|J|{\bf Y}_0\rangle\) and therefore vanishes. The possible mixed terms proportional to \(\xi t\) cancel by the symmetric moment relations \(M_{2,\eta}=s_1r-p_1\) and \(M_{3,\eta}=-p_1r\). The real coefficients of \(\xi^2\) and \(\xi\) are proportional to \(1+a_Js_1+b_Jp_1-c_Jp_1\), and hence vanish by \eqref{eq:moment-Gram-identity-general-r}. The remaining coefficients, simplified with \eqref{eq:aJ-p1-chi0}, give \eqref{eq:gamma3-general-r}--\eqref{eq:zeta0-general-r}.
		
		Finally,
		\[
		\chi_0\neq0, \qquad r\neq0,
		\]
		and hence \(\gamma_3\neq0\).
	\end{proof}
	
	The real part of the denominator is controlled by the positive quadratic polynomial defined next. Set \(\mathcal L_r(\xi)=|1+{\rm i}r\xi|^2=1 - 2\operatorname{Im}(r)\xi + |r|^2\xi^2\). Since a nonzero admissible \(r\) cannot be purely imaginary,
	\[
	\operatorname{Re}r\neq0.
	\]
	Moreover,
	\begin{equation}
		\label{eq:Lr-positive-general}
		\mathcal L_r(\xi)
		=
		|r|^2
		\left(
		\xi-\frac{\operatorname{Im}r}{|r|^2}
		\right)^2
		+
		\frac{(\operatorname{Re}r)^2}{|r|^2}
		>0.
	\end{equation}
	
	Combining this positivity with the polynomial representation of \(\mathcal Z\) yields the explicit rational family.
	
	\begin{proposition}
		\label{prop:explicit-rational-double}
		Let \(r\neq0\) satisfy \eqref{eq:r-general-null-condition}. For \(q=e^{{\rm i}\alpha}\) with \(\sin\alpha\neq0\), define
		\begin{equation}
			\label{eq:Dalpha-explicit-rational}
			\begin{aligned}
				\mathcal D_\alpha^{(r)}(x,t)
				={}&
				-\lambda_0\mathcal L_r(\xi)
				\\
				&+
				2{\rm i}\lambda_0
				\left[
				\zeta_0\cot\alpha
				-
				\gamma_3\xi^3
				-
				\gamma_2\xi^2
				-
				\gamma_1\xi
				-
				\gamma_t t
				-
				\gamma_0
				\right].
			\end{aligned}
		\end{equation}
		For \(j=1,2\), define \(\mathcal N_j^{(r)}(x,t)=2P_jP_0^*\). Equivalently,
		\begin{equation}
			\label{eq:Nj-expanded-rational}
			\begin{aligned}
				\mathcal N_j^{(r)}
				=
				2\lambda_0
				\Big[
				&
				h_{j,0}
				+
				rh_{j,1}
				-
				2\operatorname{Im}(r)h_{j,0}\xi
				\\
				&
				+
				|r|^2h_{j,0}\xi^2
				-
				{\rm i}|r|^2h_{j,1}\xi
				\Big].
			\end{aligned}
		\end{equation}
		Then
		\begin{equation}
			\label{eq:explicit-rational-double-solution}
			u_j^{(R,2;\alpha,r)}
			=
			u_{j,0}
			\left[
			1+
			\frac{
				\mathcal N_j^{(r)}(x,t)
			}{
				\mathcal D_\alpha^{(r)}(x,t)
			}
			\right],
			\qquad
			j=1,2.
		\end{equation}
		The polynomial degrees satisfy \(\deg\mathcal N_j^{(r)}\le2\) and \(\deg\mathcal D_\alpha^{(r)}\le3\). Moreover,
		\begin{equation}
			\label{eq:Dalpha-real-part-rational}
			\operatorname{Re}
			\mathcal D_\alpha^{(r)}(x,t)
			=
			-\lambda_0\mathcal L_r(\xi)
			<0
		\end{equation}
		for all real \((x,t)\). Hence the fields \eqref{eq:explicit-rational-double-solution} are globally regular background-normalized rational functions.
	\end{proposition}
	
	\begin{proof}
		The pure-sector null condition gives \(|P_0|^2 - A|P_1|^2 - B|P_2|^2=0\). Consequently, \(|P_0|^2 + A|P_1|^2 + B|P_2|^2=2|P_0|^2\). Since \(P_0=\lambda_0(1+{\rm i}r\xi)\), one has \(|P_0|^2=\lambda_0^2\mathcal L_r(\xi)\). Substitution of Lemma~\ref{lem:Z-general-complex-r} into the directional limiting kernel gives \eqref{eq:Dalpha-explicit-rational}.
		
		Direct multiplication of \(2P_jP_0^*\) gives \eqref{eq:Nj-expanded-rational}. The strict inequality in \eqref{eq:Dalpha-real-part-rational} follows from \eqref{eq:Lr-positive-general}. This proves global regularity.
	\end{proof}
	
	\begin{corollary}
		\label{cor:explicit-rational-imaginary}
		For
		\[
		q={\rm i}, \qquad \alpha=\frac{\pi}{2},
		\]
		the denominator becomes
		\begin{equation}
			\label{eq:Dpi2-explicit-rational}
			\begin{aligned}
				\mathcal D_{\pi/2}^{(r)}(x,t)
				={}&
				-\lambda_0\mathcal L_r(\xi)
				\\
				&-
				2{\rm i}\lambda_0
				\left(
				\gamma_3\xi^3
				+
				\gamma_2\xi^2
				+
				\gamma_1\xi
				+
				\gamma_t t
				+
				\gamma_0
				\right).
			\end{aligned}
		\end{equation}
		The corresponding solution is
		\[
		u_j^{(R,2;{\rm i},r)}=u_{j,0} \left[ 1+ \frac{ \mathcal N_j^{(r)} }{ \mathcal D_{\pi/2}^{(r)} } \right], \qquad j=1,2,
		\]
		and is globally regular.
	\end{corollary}
	
	\begin{remark}
		\label{rem:scope-explicit-double}
		When \(r\in\mathbb R\), the null condition permits only
		\[
		r=0 \qquad\text{or}\qquad r=r_{\rm c}=-\frac{2a_J}{b_J}.
		\]
		For the nontrivial real representative \(r=r_{\rm c}\), one has \(\gamma_2=\gamma_0=0\), and the imaginary polynomial reduces to an odd cubic polynomial in \(\xi\) plus the term \(\gamma_t t\).
		
		For a genuinely complex null coefficient,
		\[
		\operatorname{Im}r\neq0,
		\]
		the new coefficients \(\gamma_2\) and \(\gamma_0\) are generally nonzero. They shift the center and lower-order geometry of the rational structure but do not affect its global regularity or its leading cubic far-field balance.
	\end{remark}
	
	\subsection{Translation equivalences}
	\label{subsec:double-translations}
	
	For a fixed nonzero complex null coefficient \(r\), the approach direction enters only through the constant term \(\zeta_0\cot\alpha\) in the imaginary part of the denominator. When \(\gamma_t\neq0\), this dependence can be absorbed by a translation that preserves the characteristic coordinate \(\xi\).
	
	\begin{proposition}
		\label{prop:alpha-translation}
		Let \(r\neq0\) satisfy \eqref{eq:r-general-null-condition}, and assume \(\gamma_t\neq0\). For \(q=e^{{\rm i}\alpha}\) with \(\sin\alpha\neq0\), define \(\tau_\alpha^{(r)}=\frac{\zeta_0}{\gamma_t}\cot\alpha\). Then
		\begin{equation}
			\label{eq:m-alpha-translation}
			\mathcal D_\alpha^{(r)}(x,t)
			=
			\mathcal D_{\pi/2}^{(r)}
			\left(
			x+\nu_1\tau_\alpha^{(r)},
			t-\tau_\alpha^{(r)}
			\right).
		\end{equation}
		Moreover,
		\begin{equation*}
			\begin{aligned}
				u_j^{(R,2;\alpha,r)}(x,t)
				={}&
				\exp
				\left\{
				{\rm i}
				\left(
				\omega_j-k_j\nu_1
				\right)
				\tau_\alpha^{(r)}
				\right\}
				\\
				&\times
				u_j^{(R,2;{\rm i},r)}
				\left(
				x+\nu_1\tau_\alpha^{(r)},
				t-\tau_\alpha^{(r)}
				\right).
			\end{aligned}
		\end{equation*}
		Consequently,
		\begin{equation*}
			\left|
			u_j^{(R,2;\alpha,r)}(x,t)
			\right|^2
			=
			\left|
			u_j^{(R,2;{\rm i},r)}
			\left(
			x+\nu_1\tau_\alpha^{(r)},
			t-\tau_\alpha^{(r)}
			\right)
			\right|^2.
		\end{equation*}
	\end{proposition}
	
	\begin{proof}
		Set \(x'=x+\nu_1\tau_\alpha^{(r)}\) and \(t'=t-\tau_\alpha^{(r)}\). Then \(x'+\nu_1t'=x+\nu_1t=\xi\). Hence \(\mathcal L_r(\xi)\) and \(\mathcal N_j^{(r)}(\xi)\) are invariant under this translation. Furthermore, \(\gamma_t t'=\gamma_t t - \zeta_0\cot\alpha\). Substitution into \eqref{eq:Dpi2-explicit-rational} gives \eqref{eq:m-alpha-translation}.
		
		The phase relation follows from
		\[
		u_{j,0}(x,t)=\exp \left\{ {\rm i} \left( \omega_j-k_j\nu_1 \right) \tau_\alpha^{(r)} \right\} u_{j,0}(x',t').
		\]
	\end{proof}
	
	\begin{remark}
		\label{rem:ct-zero-translation}
		If \(\gamma_t=0\), then
		\[
		\mathcal D_\alpha^{(r)}=\mathcal D_{\pi/2}^{(r)} + 2{\rm i}\lambda_0\zeta_0\cot\alpha.
		\]
		If also \(\zeta_0=0\), the limiting solution is independent of the approach angle. If \(\zeta_0\neq0\), the angle produces a constant imaginary offset that cannot be removed by a translation preserving \(\xi\).
	\end{remark}
	
	The null-circle coefficient itself also has a geometric interpretation: every nonzero complex point on the pure-sector null circle is related to the canonical real representative by a rigid translation and an irrelevant scalar normalization.
	
	\begin{proposition}
		\label{prop:r-null-circle-translation}
		Let
		\[
		r_{\rm c}=-\frac{2a_J}{b_J},
		\]
		and let \(r\neq0\) satisfy \eqref{eq:r-general-null-condition}. Define \(\delta_r=\operatorname{Im}\frac{1}{r}\) and \(\kappa_r=\left| \frac{r}{r_{\rm c}} \right|^2\). The null condition implies
		\begin{equation}
			\label{eq:inverse-r-line}
			\operatorname{Re}\frac1r
			=
			\frac1{r_{\rm c}}.
		\end{equation}
		Hence
		\[
		\frac1r=\frac1{r_{\rm c}} + {\rm i}\delta_r.
		\]
		
		Let the subscript \({\rm c}\) denote the coefficients associated with \(r=r_{\rm c}\). Then
		\begin{equation}
			\label{eq:P-translation-r}
			P_\ell^{(r)}(\xi)
			=
			\frac{r}{r_{\rm c}}
			P_\ell^{({\rm c})}
			(\xi+\delta_r),
			\qquad
			\ell=0,1,2,
		\end{equation}
		and therefore
		\begin{equation}
			\label{eq:N-translation-r}
			\mathcal N_j^{(r)}(\xi)
			=
			\kappa_r
			\mathcal N_j^{({\rm c})}
			(\xi+\delta_r).
		\end{equation}
		
		Moreover,
		\begin{equation}
			\label{eq:Z-translation-r}
			\mathcal Z_r(\xi,t)
			=
			\kappa_r
			\mathcal Z_{\rm c}
			(\xi+\delta_r,t)
			+
			A_r+{\rm i}B_r,
		\end{equation}
		where
		\begin{equation}
			\label{eq:AB-r-translation}
			A_r
			=
			\zeta_0(r)
			-
			\kappa_r\zeta_0(r_{\rm c}),
		\end{equation}
		and
		\begin{equation}
			\label{eq:Br-translation}
			B_r
			=
			\gamma_0(r)
			-
			\kappa_r
			\left[
			\gamma_3(r_{\rm c})\delta_r^3
			+
			\gamma_1(r_{\rm c})\delta_r
			\right].
		\end{equation}
		
		Assume \(\gamma_t(r_{\rm c})\neq0\), and define
		\begin{equation}
			\label{eq:tau-r-alpha}
			\tau_{r,\alpha}
			=
			\frac{
				A_r\cot\alpha-B_r
			}{
				\kappa_r\gamma_t(r_{\rm c})
			}.
		\end{equation}
		Then
		\begin{equation}
			\label{eq:D-r-canonical-translation}
			\mathcal D_\alpha^{(r)}(x,t)
			=
			\kappa_r
			\mathcal D_\alpha^{(r_{\rm c})}
			\left(
			x+\delta_r+\nu_1\tau_{r,\alpha},
			t-\tau_{r,\alpha}
			\right).
		\end{equation}
		Consequently,
		\begin{equation}
			\label{eq:u-r-canonical-translation}
			\begin{aligned}
				u_j^{(R,2;\alpha,r)}(x,t)
				={}&
				\exp
				\left\{
				{\rm i}
				\left[
				-k_j\delta_r
				+
				(\omega_j-k_j\nu_1)
				\tau_{r,\alpha}
				\right]
				\right\}
				\\
				&\times
				u_j^{(R,2;\alpha,r_{\rm c})}
				\left(
				x+\delta_r+\nu_1\tau_{r,\alpha},
				t-\tau_{r,\alpha}
				\right).
			\end{aligned}
		\end{equation}
		In particular, every nonzero complex point of the pure-sector null circle has the same modulus profile as the canonical real representative, up to a rigid space--time translation.
	\end{proposition}
	
	\begin{proof}
		Dividing the null condition by \(|r|^2\) gives
		\[
		b_J + 2a_J \operatorname{Re}\frac1r=0,
		\]
		which proves \eqref{eq:inverse-r-line}. It follows directly that
		\[
		1+{\rm i}r\xi=\frac{r}{r_{\rm c}} \left[ 1+{\rm i}r_{\rm c} (\xi+\delta_r) \right].
		\]
		The same identity applied to the lower components proves \eqref{eq:P-translation-r} and \eqref{eq:N-translation-r}.
		
		Substitution of the explicit coefficients \eqref{eq:gamma3-general-r}-- \eqref{eq:zeta0-general-r} shows that the difference in \eqref{eq:Z-translation-r} is independent of \(\xi\) and \(t\). Evaluating it at \(\xi=t=0\) gives \eqref{eq:AB-r-translation} and \eqref{eq:Br-translation}.
		
		Finally, \(\gamma_t(r)=\kappa_r\gamma_t(r_{\rm c})\), and the choice \eqref{eq:tau-r-alpha} removes the remaining constant contribution \(A_r\cot\alpha-B_r\). Equations \eqref{eq:D-r-canonical-translation} and \eqref{eq:u-r-canonical-translation} follow.
	\end{proof}
	
	\section{Darboux Constructions at Triple Spatial Roots}
	\label{sec:triple-darboux}
	\label{sec:triple}
	
	Throughout this section, \(\chi_0\), \(\mu_0\), and \(\lambda_0\) denote the real triple-root data characterized in Theorem~\ref{thm:triple-characterization}.
	
	\subsection{The triple confluent eigenspace and its null cone}
	\label{subsec:triple-null-cone}
	\label{subsec:full-triple-chain}
	
	Define
	\begin{equation*}
		{\bf Y}_n
		=
		\left.
		\frac{\partial^n}{\partial\chi^n}
		{\bf y}(\chi,\lambda_0)
		\right|_{\chi=\chi_0},
		\qquad
		n=0,\ldots,5,
	\end{equation*}
	At a triple branch point, \(\Lambda'(\chi_0)=\Lambda''(\chi_0)=0\). Lemma~\ref{lem:fixed-spectral-derivative-chain} with \(m=3\) therefore shows that \({\bf Y}_0,{\bf Y}_1,{\bf Y}_2\) all satisfy the background Lax pair at the fixed spectral point \(\lambda=\lambda_0\). Proposition~\ref{prop:confluent-basis-independence} proves that these three solutions are linearly independent, so they form a fundamental solution basis. The higher derivatives \({\bf Y}_3,{\bf Y}_4,{\bf Y}_5\) are used only as Taylor coefficients in the first correction vector; they are not asserted to be independent fixed-\(\lambda_0\) Lax eigenfunctions.
	
	Set \(g_{mn} = \langle{\bf Y}_m|J|{\bf Y}_n\rangle, \qquad 0\le m,n\le2\).
	\begin{lemma}
		\label{lem:triple-Gram-structure}
		At a real triple branch point, the Gram coefficients \(g_{mn} = [{\bf Y}_m,{\bf Y}_n]_J, \qquad 0\le m,n\le2\), satisfy
		\begin{equation}
			\label{eq:triple-Gram-relations}
			g_{00}=0,
			\qquad
			g_{01}=0,
			\qquad
			g_{02}=2g_{11}<0.
		\end{equation}
		More explicitly, \(g_{11} = -A[h_1'(\chi_0)]^2 - B[h_2'(\chi_0)]^2 <0\), \(g_{12} = -Ah_1'(\chi_0)h_1''(\chi_0) - Bh_2'(\chi_0)h_2''(\chi_0) \in\mathbb R\), and \(g_{22} = -A[h_1''(\chi_0)]^2 - B[h_2''(\chi_0)]^2 <0\). The lower Gram block
		\begin{equation*}
			G_T^-
			=
			\begin{pmatrix}
				g_{11}&g_{12}\\
				g_{12}&g_{22}
			\end{pmatrix}
		\end{equation*}
		is strictly negative definite.
	\end{lemma}
	
	\begin{proof}
		The identity \(g_{00}=0\) follows from Lemma~\ref{lem:J-null}. As in the double-root calculation,
		\[
		g_{01} = -Ah_1(\chi_0)h_1'(\chi_0) - Bh_2(\chi_0)h_2'(\chi_0) = -\frac{\chi_0}{2}\Lambda''(\chi_0).
		\]
		At a triple branch point, \(\Lambda''(\chi_0)=0\), and hence \(g_{01}=0\).
		
		Let \(X=x+\nu'(\chi_0)t\), and introduce the diagonal unitary factor
		\[
		\mathscr U_0 = \operatorname{diag} \left( E_0, \frac{u_{1,0}}{a_1}E_0, \frac{u_{2,0}}{a_2}E_0 \right).
		\]
		Set
		\[
		{\bf v} =
		\begin{pmatrix}
			\lambda_0\\
			a_1h_1(\chi_0)\\
			a_2h_2(\chi_0)
		\end{pmatrix},
		\qquad {\bf v}^{(n)} =
		\begin{pmatrix}
			0\\
			a_1h_1^{(n)}(\chi_0)\\
			a_2h_2^{(n)}(\chi_0)
		\end{pmatrix},
		\quad n\ge1.
		\]
		The first two derivative modes have the exact forms
		\[
		{\bf Y}_1 = \mathscr U_0 \left( {\bf v}^{(1)} + {\rm i}X{\bf v} \right),
		\]
		and
		\[
		{\bf Y}_2 = \mathscr U_0 \left[ {\bf v}^{(2)} + 2{\rm i}X{\bf v}^{(1)} + \left( {\rm i}\nu''(\chi_0)t-X^2 \right) {\bf v} \right].
		\]
		Since \(\mathscr U_0\) is unitary and commutes with \(J\), it does not affect the Gram coefficients.
		
		Using \(g_{00}=g_{01}=0\), all \(x,t\)-dependent contributions cancel from the Gram coefficients. One obtains \(g_{11} = -A[h_1'(\chi_0)]^2 - B[h_2'(\chi_0)]^2\), and \(g_{02} = -Ah_1(\chi_0)h_1''(\chi_0) - Bh_2(\chi_0)h_2''(\chi_0)\). Since \(h_j(\chi)h_j''(\chi) = 2[h_j'(\chi)]^2, \qquad j=1,2\), it follows that \(g_{02}=2g_{11}<0\). The same cancellation gives \(g_{12} = -Ah_1'(\chi_0)h_1''(\chi_0) - Bh_2'(\chi_0)h_2''(\chi_0)\), and \(g_{22} = -A[h_1''(\chi_0)]^2 - B[h_2''(\chi_0)]^2\).
		
		It remains to prove strict negative definiteness. The first leading principal minor satisfies \(g_{11}<0\). The determinant is
		\begin{align*}
			\det G_T^-
			&=
			g_{11}g_{22}-|g_{12}|^2
			\\
			&=
			AB
			\left[
			h_1'(\chi_0)h_2''(\chi_0)
			-
			h_2'(\chi_0)h_1''(\chi_0)
			\right]^2.
		\end{align*}
		Using
		\[
		h_1' = -\frac{k}{(k+\chi)^2}, \qquad h_1'' = \frac{2k}{(k+\chi)^3},
		\]
		and
		\[
		h_2' = \frac{k}{(k-\chi)^2}, \qquad h_2'' = \frac{2k}{(k-\chi)^3},
		\]
		we obtain
		\[
		h_1'h_2''-h_2'h_1'' = -\frac{ 4k^3 }{ (k+\chi_0)^3(k-\chi_0)^3 } \neq0.
		\]
		Therefore \(\det G_T^->0\). Sylvester's criterion now shows that \(G_T^-\) is strictly negative definite.
	\end{proof}
	
	The full triple-chain leading vector is
	\begin{equation}
		\label{eq:projective-triple-leading}
		{\bf y}_R
		=
		M_0{\bf Y}_0
		+
		M_1{\bf Y}_1
		+
		\frac{M_2}{2}{\bf Y}_2.
	\end{equation}
	
	\begin{proposition}
		\label{prop:admissible-triple}
		A nonzero vector of the form \eqref{eq:projective-triple-leading} can be \(J\)-null only if \(M_0\neq0\). After the projective normalization \(M_0=1\), the admissible coefficients satisfy
		\begin{equation}
			\label{eq:triple-null-general}
			g_{11}|M_1|^2
			+
			g_{02}\operatorname{Re}M_2
			+
			\operatorname{Re}
			\left(
			g_{12}M_1^*M_2
			\right)
			+
			\frac{g_{22}}{4}|M_2|^2
			=
			0.
		\end{equation}
	\end{proposition}
	
	\begin{proof}
		If \(M_0=0\), then
		\[
		{\bf y}_R = M_1{\bf Y}_1+\frac{M_2}{2}{\bf Y}_2.
		\]
		The negative definiteness of \(G_T^-\) gives \(\langle{\bf y}_R|J|{\bf y}_R\rangle<0\) for every nonzero pair \((M_1,M_2)\). Hence \(M_0\neq0\) for every nonzero null vector. After setting \(M_0=1\), direct expansion using \eqref{eq:triple-Gram-relations} gives \eqref{eq:triple-null-general}.
	\end{proof}
	
	For fixed \(M_1\), introduce \(\beta(M_1) = g_{02}+g_{12}M_1^*, \qquad \Delta_T(M_1) = |\beta(M_1)|^2 - g_{11}g_{22}|M_1|^2\). Then \eqref{eq:triple-null-general} is equivalent to
	\begin{equation}
		\label{eq:M2-circle-general-triple}
		\left|
		M_2+\frac{2\beta(M_1)^*}{g_{22}}
		\right|^2
		=
		\frac{4\Delta_T(M_1)}{g_{22}^2}.
	\end{equation}
	Thus a null vector with the prescribed \(M_1\) exists if and only if \(\Delta_T(M_1)\ge0\).
	
	\subsection{Cubic root splitting}
	\label{subsec:triple-splitting}
	
	The triple degeneration is defined at the fixed endpoint \(\mu_0\) by \(F(\chi;\mu_0)=\mu_0(\chi-\chi_0)^3\). As in the real double-root case, the following displacement of that single endpoint is introduced only to regularize the real-spectrum Darboux kernel and to resolve the nearby spatial branches. It is not a coalescence of distinct Darboux poles.
	
	Let \(\mu= \lambda(\eta)^2=\mu_0+\eta, \qquad \lambda(0)=\lambda_0\). Since \(\lambda_0\neq0\), the holomorphic implicit function theorem gives a unique local holomorphic branch satisfying \(\lambda(\eta) = \lambda_0 + \frac{\eta}{2\lambda_0} + O(|\eta|^2)\).
	
	Write \(\chi=\chi_0+\delta\). Since \(F\) is affine in \(\mu\), \(F(\chi;\mu_0+\eta) = F(\chi;\mu_0) + \eta D(\chi)\), where \(D(\chi)=\chi^3-k^2\chi\). Define \(D_0 = D(\chi_0) = \chi_0^3-k^2\chi_0, \qquad D_1 = D'(\chi_0) = 3\chi_0^2-k^2\). Using \eqref{eq:triple-factorization}, one obtains the exact perturbed cubic
	\begin{equation}
		\label{eq:delta-cubic-triple}
		G_3(\delta,\eta)
		:=
		(\mu_0+\eta)\delta^3
		+
		3\chi_0\eta\,\delta^2
		+
		D_1\eta\,\delta
		+
		D_0\eta
		=
		0.
	\end{equation}
	
	For the physical triple point, \(D_0 = \chi_0(\chi_0^2-k^2) = \chi_0k^2(q_{\rm tri}^2-1) <0\), and in particular \(D_0\neq0\).
	
	\begin{lemma}
		\label{lem:triple-root-localization}
		For all sufficiently small nonzero \(\eta\), all three roots of \eqref{eq:delta-cubic-triple} lie in a disk \(|\delta|<R|\eta|^{1/3}\), where \(R>0\) is independent of \(\eta\). Thus \(\delta_j=O(|\eta|^{1/3}), \qquad j=1,2,3\). After choosing a local branch \(t=\eta^{1/3}\), the roots may be labeled so that
		\begin{equation}
			\label{eq:triple-Puiseux-leading}
			\delta_j
			=
			\alpha\omega_jt
			+
			O(t^2),
			\qquad
			\alpha^3=-\frac{D_0}{\mu_0},
		\end{equation}
		where \(\omega_j^3=1, \qquad \omega_j\neq\omega_\ell \quad (j\neq\ell)\). In particular, the three roots are distinct for sufficiently small nonzero \(\eta\).
	\end{lemma}
	
	\begin{proof}
		Define \(P_\eta(\delta) = \mu_0\delta^3+D_0\eta\). Choose \(R>0\) such that \(\mu_0R^3>|D_0|\). On \(|\delta|=R|\eta|^{1/3}\), the reverse triangle inequality gives
		\[
		|P_\eta(\delta)| \ge \left( \mu_0R^3-|D_0| \right)|\eta|.
		\]
		Moreover,
		\[
		\begin{aligned}
			G_3(\delta,\eta)-P_\eta(\delta)
			&=
			\eta\delta^3
			+
			3\chi_0\eta\delta^2
			+
			D_1\eta\delta.
		\end{aligned}
		\]
		On the same circle, the three terms on the right-hand side are respectively \(O(|\eta|^2), \qquad O(|\eta|^{5/3}), \qquad O(|\eta|^{4/3})\). Hence \(|G_3-P_\eta| = o(|\eta|) < |P_\eta|\) for all sufficiently small nonzero \(\eta\). Rouch\'e's theorem implies that \(G_3\) and \(P_\eta\) have the same number of zeros in the disk. Since \(P_\eta\) has three zeros there, so does \(G_3\). As \(G_3\) is cubic, these are all its zeros.
		
		To obtain the local branches, set \(\eta=t^3, \qquad \delta=tz\). Dividing \eqref{eq:delta-cubic-triple} by \(t^3\) gives \((\mu_0+t^3)z^3 + 3\chi_0t^2z^2 + D_1tz + D_0 = 0\). At \(t=0\), this reduces to \(\mu_0z^3+D_0=0\). Its three roots \(z_j(0)=\alpha\omega_j\) are nonzero and simple because \(D_0\neq0\). The holomorphic implicit function theorem therefore produces three analytic functions \(z_j(t)\) satisfying \(z_j(t) = \alpha\omega_j+O(t)\). Multiplication by \(t\) yields \eqref{eq:triple-Puiseux-leading}.
	\end{proof}
	
	\subsection{Moment recurrences and the full triple-chain correction vector}
	\label{subsec:triple-correction}
	
	We now consider the generic full triple-chain sector \(M_1M_2\neq0, \qquad \eqref{eq:triple-null-general}\ \text{holds}\).
	
	Let the three split roots be \(\chi_j=\chi_0+\delta_j, \qquad j=1,2,3\). The coefficients \(C_j\) are fixed by the moment conditions \(\sum_{j=1}^3C_j=1, \qquad \sum_{j=1}^3C_j\delta_j=M_1, \qquad \sum_{j=1}^3C_j\delta_j^2=M_2\). Solving the Vandermonde system gives
	\begin{equation*}
		C_j
		=
		\frac{
			M_2
			-
			(\delta_\ell+\delta_m)M_1
			+
			\delta_\ell\delta_m
		}{
			(\delta_j-\delta_\ell)
			(\delta_j-\delta_m)
		},
		\qquad
		\{j,\ell,m\}=\{1,2,3\}.
	\end{equation*}
	Since \(\delta_j=O(|\eta|^{1/3})\), one generically has \(C_j=O(|\eta|^{-2/3})\). The singular coefficient parts cancel among the three coalescing eigenfunctions and leave the finite vector
	\[
	{\bf Y}_0 + M_1{\bf Y}_1 + \frac{M_2}{2}{\bf Y}_2.
	\]
	
	Define \(\mathfrak M_n = \sum_{j=1}^3C_j\delta_j^n\). Then
	\[
	\mathfrak M_0=1, \qquad \mathfrak M_1=M_1, \qquad \mathfrak M_2=M_2.
	\]
	
	\begin{proposition}
		\label{prop:triple-general-moments}
		For fixed \(M_1,M_2\in\mathbb C\),
		\begin{align*}
			\mathfrak M_3
			&=
			\mathcal M_3\eta+O(|\eta|^2),
			\\
			\mathfrak M_4
			&=
			\mathcal M_4\eta+O(|\eta|^2),
			\\
			\mathfrak M_5
			&=
			\mathcal M_5\eta+O(|\eta|^2),
		\end{align*}
		where
		\begin{align}
			\mathcal M_3
			&=
			-\frac{
				3\chi_0M_2+D_1M_1+D_0
			}{
				\mu_0
			},
			\label{eq:M3eta-triple-general}
			\\
			\mathcal M_4
			&=
			-\frac{
				D_1M_2+D_0M_1
			}{
				\mu_0
			}, \notag\\
			\mathcal M_5
			&=
			-\frac{D_0M_2}{\mu_0}.
			\label{eq:M5eta-triple-general}
		\end{align}
		Moreover,
		for every \(\tau>0\), after reducing \(|\eta|\) if necessary, there is a constant \(C_\tau>0\), independent of \(n\) and \(\eta\), such that
		\begin{equation}
			\label{eq:higher-triple-moments}
			|\mathfrak M_n|
			\le
			C_\tau|\eta|^2\tau^{n-6},
			\qquad
			n\ge6.
		\end{equation}
	\end{proposition}
	
	\begin{proof}
		Each split root satisfies
		\begin{equation}
			\label{eq:triple-normalized-root-equation}
			\delta_j^3
			+a_2(\eta)\delta_j^2
			+a_1(\eta)\delta_j
			+a_0(\eta)=0,
		\end{equation}
		where
		\[
		a_2(\eta) = \frac{3\chi_0\eta}{\mu_0+\eta},
		\qquad
		a_1(\eta) = \frac{D_1\eta}{\mu_0+\eta},
		\qquad
		a_0(\eta) = \frac{D_0\eta}{\mu_0+\eta}.
		\]
		Multiplying \eqref{eq:triple-normalized-root-equation} by \(C_j\delta_j^{n-3}\) and summing over \(j=1,2,3\) gives the exact recurrence
		\begin{equation}
			\label{eq:triple-exact-moment-recurrence}
			\mathfrak M_n
			+a_2(\eta)\mathfrak M_{n-1}
			+a_1(\eta)\mathfrak M_{n-2}
			+a_0(\eta)\mathfrak M_{n-3}
			=0,
			\qquad n\ge3.
		\end{equation}
		Because
		\[
		a_2(\eta)=\frac{3\chi_0}{\mu_0}\eta+O(|\eta|^2),
		\qquad
		a_1(\eta)=\frac{D_1}{\mu_0}\eta+O(|\eta|^2),
		\qquad
		a_0(\eta)=\frac{D_0}{\mu_0}\eta+O(|\eta|^2),
		\]
		setting \(n=3\) in \eqref{eq:triple-exact-moment-recurrence} and using \(\mathfrak M_0=1\), \(\mathfrak M_1=M_1\), \(\mathfrak M_2=M_2\) gives
		\[
		\mathfrak M_3
		=-\frac{3\chi_0M_2+D_1M_1+D_0}{\mu_0}\eta
		+O(|\eta|^2).
		\]
		Next, the cases \(n=4\) and \(n=5\) give respectively
		\[
		\mathfrak M_4
		=-\frac{D_1M_2+D_0M_1}{\mu_0}\eta
		+O(|\eta|^2),
		\]
		and
		\[
		\mathfrak M_5
		=-\frac{D_0M_2}{\mu_0}\eta
		+O(|\eta|^2).
		\]
		These are \eqref{eq:M3eta-triple-general}--\eqref{eq:M5eta-triple-general}.
		
		Since \(a_0,a_1,a_2=O(|\eta|)\) and \(\mathfrak M_3,\mathfrak M_4,\mathfrak M_5=O(|\eta|)\), the recurrence gives
		\[
		\mathfrak M_6,
		\mathfrak M_7,
		\mathfrak M_8
		=O(|\eta|^2).
		\]
		For \(n\ge9\), equation \eqref{eq:triple-exact-moment-recurrence} has the form of \eqref{eq:abstract-moment-recurrence} with
		\[
		N=6,
		\qquad d=3,
		\qquad
		b_1=-a_2,
		\quad b_2=-a_1,
		\quad b_3=-a_0.
		\]
		Lemma~\ref{lem:analytic-moment-tail} therefore yields the uniform geometric estimate \eqref{eq:higher-triple-moments} for every prescribed \(\tau>0\).
	\end{proof}
	
	Consider
	\begin{equation}
		\label{eq:y-eta-triple}
		{\bf y}(\eta)
		=
		\sum_{j=1}^3
		C_j
		{\bf y}(\chi_0+\delta_j,\lambda(\eta)).
	\end{equation}
	
	\begin{proposition}
		\label{prop:triple-eigenfunction-expansion}
		The combination \eqref{eq:y-eta-triple} satisfies
		\begin{equation}
			\label{eq:y-triple-expansion}
			{\bf y}(\eta)
			=
			{\bf y}_R
			+
			\eta{\bf y}_\mu
			+
			O_{C^2_{\rm loc}}(|\eta|^2),
		\end{equation}
		where
		\begin{equation}
			\label{eq:yR-triple}
			{\bf y}_R
			=
			{\bf Y}_0
			+
			M_1{\bf Y}_1
			+
			\frac{M_2}{2}{\bf Y}_2.
		\end{equation}
		and
		\begin{equation}
			\label{eq:ymu-triple}
			\begin{aligned}
				{\bf y}_\mu
				={}&
				\frac{1}{2\lambda_0}
				\begin{pmatrix}
					E_0+M_1E_1+\dfrac{M_2}{2}E_2\\
					0\\
					0
				\end{pmatrix}
				\\
				&+
				\frac{\mathcal M_3}{6}{\bf Y}_3
				+
				\frac{\mathcal M_4}{24}{\bf Y}_4
				+
				\frac{\mathcal M_5}{120}{\bf Y}_5.
			\end{aligned}
		\end{equation}
		More precisely, for every compact set \(\Omega\subset\mathbb R^2\), there exist constants \(C_\Omega>0\) and \(\eta_\Omega>0\) such that
		\[
		\left\|
		{\bf y}(\eta)-{\bf y}_R-\eta{\bf y}_\mu
		\right\|_{C^2(\Omega)}
		\le
		C_\Omega|\eta|^2,
		\qquad 0<|\eta|<\eta_\Omega.
		\]
	\end{proposition}
	
	\begin{proof}
		Fix a compact set \(\Omega\subset\mathbb R^2\) and let
		\[
		\mathcal B=C^2(\Omega;\mathbb C^3).
		\]
		Choose \(r_0>0\) so that \(|\chi-\chi_0|\le r_0\) avoids \(\{0,\pm k\}\). Then
		\[
		f(\delta)={\bf y}(\chi_0+\delta,\lambda_0)
		\]
		is holomorphic on \(|\delta|<r_0\) as a \(\mathcal B\)-valued function, with Taylor coefficients \(f_n={\bf Y}_n/n!\). Hence
		\begin{equation}
			\label{eq:triple-fixed-lambda-full-series}
			\sum_{j=1}^3
			C_j{\bf y}(\chi_0+\delta_j,\lambda_0)
			=
			\sum_{n=0}^{\infty}
			\frac{\mathfrak M_n}{n!}{\bf Y}_n
			\qquad\text{in }\mathcal B.
		\end{equation}
		The prescribed identities \(\mathfrak M_0=1\), \(\mathfrak M_1=M_1\), and \(\mathfrak M_2=M_2\), together with \eqref{eq:M3eta-triple-general}--\eqref{eq:M5eta-triple-general}, show that the terms with \(0\le n\le5\) in \eqref{eq:triple-fixed-lambda-full-series} equal
		\[
		{\bf y}_R
		+
		\eta
		\left(
		\frac{\mathcal M_3}{6}{\bf Y}_3
		+
		\frac{\mathcal M_4}{24}{\bf Y}_4
		+
		\frac{\mathcal M_5}{120}{\bf Y}_5
		\right)
		+O_{\mathcal B}(|\eta|^2).
		\]
		Choose \(\tau\in(0,r_0)\). The uniform estimate \eqref{eq:higher-triple-moments} and Lemma~\ref{lem:analytic-moment-tail}, applied with \(N=6\), give
		\[
		\left\|
		\sum_{n=6}^{\infty}
		\frac{\mathfrak M_n}{n!}{\bf Y}_n
		\right\|_{\mathcal B}
		=O(|\eta|^2).
		\]
		Thus the complete fixed-\(\lambda_0\) expansion, including its infinite Taylor tail, is controlled in \(C^2(\Omega)\).
		
		To account for the varying square root, use the exact identity
		\[
		{\bf y}(\chi,\lambda(\eta))
		-
		{\bf y}(\chi,\lambda_0)
		=
		\bigl[\lambda(\eta)-\lambda_0\bigr]
		\begin{pmatrix}
			E(\chi)\\0\\0
		\end{pmatrix}.
		\]
		The moment expansion of the phase gives
		\[
		\sum_{j=1}^3C_jE(\chi_0+\delta_j)
		=
		E_0+M_1E_1+\frac{M_2}{2}E_2
		+O_{C^2(\Omega)}(|\eta|).
		\]
		Indeed, the moments of orders three through five are \(O(|\eta|)\), and the full tail beginning with order six is \(O(|\eta|^2)\) by the same Cauchy estimate. Since
		\[
		\lambda(\eta)-\lambda_0
		=
		\frac{\eta}{2\lambda_0}+O(|\eta|^2),
		\]
		the spectral variation contributes
		\[
		\frac{\eta}{2\lambda_0}
		\begin{pmatrix}
			E_0+M_1E_1+\dfrac{M_2}{2}E_2\\0\\0
		\end{pmatrix}
		+O_{\mathcal B}(|\eta|^2).
		\]
		Combining the fixed-\(\lambda_0\) and spectral-variation parts gives \eqref{eq:y-triple-expansion} and \eqref{eq:ymu-triple}, with a remainder bounded by \(C_\Omega|\eta|^2\) in \(C^2(\Omega)\). Since \(\Omega\) was arbitrary, the remainder is \(O_{C^2_{\rm loc}}(|\eta|^2)\).
	\end{proof}
	
	\subsection{The directional real-spectrum limit}
	\label{subsec:triple-directional-limit}
	
	Let
	\[
	\eta=\varepsilon\zeta, \qquad \varepsilon\in\mathbb R, \qquad \varepsilon\to0, \qquad \zeta\in\mathbb C\setminus\mathbb R.
	\]
	
	\begin{theorem}
		\label{thm:triple-solution}
		Let \(\chi_0\) be the real triple branch point determined by \eqref{eq:mu-triple}--\eqref{eq:k-triple}. Let \(M_1,M_2\in\mathbb C\) satisfy the full null constraint \eqref{eq:triple-null-general}, and let \({\bf y}_R\) and \({\bf y}_\mu\) be given by \eqref{eq:yR-triple} and \eqref{eq:ymu-triple}.
		
		Define
		\begin{equation}
			\label{eq:mq-triple}
			\begin{aligned}
				m_\zeta^{(3)}
				={}&
				\frac{2\lambda_0}{\zeta^*-\zeta}
				\left[
				\zeta^*
				\langle{\bf y}_\mu|J|{\bf y}_R\rangle
				+
				\zeta
				\langle{\bf y}_R|J|{\bf y}_\mu\rangle
				\right]
				\\
				&-
				\frac{
					\langle{\bf y}_R|K|{\bf y}_R\rangle
				}{
					2\lambda_0
				}.
			\end{aligned}
		\end{equation}
		Writing \({\bf y}_R = (\psi_R,\phi_R,\varphi_R)^T\), define
		\begin{equation}
			\label{eq:triple-solution}
			\begin{pmatrix}
				u_1^{(R,3;\zeta)}\\
				u_2^{(R,3;\zeta)}
			\end{pmatrix}
			=
			\begin{pmatrix}
				u_{1,0}\\
				u_{2,0}
			\end{pmatrix}
			+
			2
			\frac{
				\begin{pmatrix}
					\phi_R\\
					\varphi_R
				\end{pmatrix}
				\psi_R^*
			}{
				m_\zeta^{(3)}
			}.
		\end{equation}
		Then
		\[
		\operatorname{Re}m_\zeta^{(3)} = - \frac{ \langle{\bf y}_R|{\bf y}_R\rangle }{ 2\lambda_0 } <0.
		\]
		Consequently, \eqref{eq:triple-solution} is a globally nonsingular exact solution of \eqref{eq:CFL}.
	\end{theorem}
	
	\begin{proof}
		Proposition~\ref{prop:triple-eigenfunction-expansion} gives the expansion \eqref{eq:unified-eigenfunction-expansion}. The full constraint \eqref{eq:triple-null-general} is exactly \eqref{eq:unified-leading-nullity}. Since the coefficient of \({\bf Y}_0\) in \({\bf y}_R\) is normalized to one and \({\bf Y}_0,{\bf Y}_1,{\bf Y}_2\) are linearly independent by Proposition~\ref{prop:confluent-basis-independence}, the leading vector is nontrivial.
		
		Theorem~\ref{thm:unified-real-spectrum-directional-limit}, applied with \(q=\zeta\), now gives \eqref{eq:mq-triple}, the strict negativity of its real part, and the global smoothness and exact solvability of \eqref{eq:triple-solution}.
	\end{proof}
	
	For the imaginary-axis approach \(\eta={\rm i}\varepsilon\), one obtains
	\begin{equation}
		\label{eq:mR-triple}
		m_R^{(3)}
		=
		\lambda_0
		\left(
		\langle{\bf y}_\mu|J|{\bf y}_R\rangle
		-
		\langle{\bf y}_R|J|{\bf y}_\mu\rangle
		\right)
		-
		\frac{
			\langle{\bf y}_R|K|{\bf y}_R\rangle
		}{
			2\lambda_0
		}.
	\end{equation}
	
	\subsection{Lower coefficient sectors}
	\label{subsec:triple-special-sectors}
	
	\begin{corollary}[Pure second-derivative sector]
		\label{cor:pure-second-triple}
		If \(M_1=0, \qquad M_2\neq0\), then the null constraint becomes
		\begin{equation}
			\label{eq:M2-circle-triple}
			\left|
			M_2+\frac{2g_{02}}{g_{22}}
			\right|
			=
			\left|
			\frac{2g_{02}}{g_{22}}
			\right|.
		\end{equation}
		The leading vector is
		\[
		{\bf y}_R = {\bf Y}_0+\frac{M_2}{2}{\bf Y}_2,
		\]
		and the correction vector is obtained from \eqref{eq:ymu-triple} with
		\begin{align*}
			\mathcal M_3
			&=
			-\frac{3\chi_0M_2+D_0}{\mu_0},
			\\
			\mathcal M_4
			&=
			-\frac{D_1M_2}{\mu_0},
			\\
			\mathcal M_5
			&=
			-\frac{D_0M_2}{\mu_0}.
		\end{align*}
		The same limiting kernel and Darboux formula remain valid.
	\end{corollary}
	
	\begin{corollary}[Canonical triple-root sector]
		\label{cor:canonical-triple}
		If \(M_1=M_2=0\), then \({\bf y}_R={\bf Y}_0\) is automatically \(J\)-null. The corresponding correction vector is
		\begin{equation*}
			{\bf y}_\mu^{(0)}
			=
			\frac{1}{2\lambda_0}
			\begin{pmatrix}
				E_0\\
				0\\
				0
			\end{pmatrix}
			-
			\frac{D_0}{6\mu_0}{\bf Y}_3.
		\end{equation*}
		Substitution of \({\bf y}_R={\bf Y}_0\) and \({\bf y}_\mu={\bf y}_\mu^{(0)}\) into \eqref{eq:mq-triple} and \eqref{eq:triple-solution} gives the canonical real-spectrum triple-root solution.
	\end{corollary}
	
	\begin{remark}
		\label{rem:triple-special-sectors}
		If \(M_2=0\), the null condition reduces to \(g_{11}|M_1|^2=0\). Since \(g_{11}<0\), the case \(M_2=0,\qquad M_1\neq0\) is inadmissible. Thus the canonical vector \({\bf Y}_0\) is the only admissible vector without a second-derivative contribution.
		
		For fixed \(M_1\neq0\), the full-chain circle \eqref{eq:M2-circle-general-triple} collapses to a single point when \(\Delta_T(M_1)=0\). At this saturated boundary,
		\[
		M_2 = -\frac{2\beta(M_1)^*}{g_{22}}.
		\]
		This case is already contained in the general full-chain formula and requires no separate limiting construction.
	\end{remark}
	
	\begin{remark}
		\label{rem:admissible-summary}
		Up to an overall nonzero complex factor, the possible leading coefficient combinations are summarized below. For the full triple-chain sector, we use
		\[
		\beta(M_1)=g_{02}+g_{12}M_1^*, \qquad \Delta_T(M_1) = |\beta(M_1)|^2-g_{11}g_{22}|M_1|^2.
		\]
		
		\begin{center}
			\begingroup
			\renewcommand{\arraystretch}{1.35}
			\setlength{\tabcolsep}{3pt}
			\small
			
			\begin{tabular}{
					>{\centering\arraybackslash}m{0.08\textwidth}
					|
					>{\centering\arraybackslash}m{0.25\textwidth}
					|
					>{\centering\arraybackslash}m{0.27\textwidth}
					|
					>{\centering\arraybackslash}m{0.31\textwidth}
				}
				\hline
				\textbf{Branch type}
				&
				\textbf{Coefficient choice}
				&
				\textbf{Null constraint}
				&
				\textbf{Interpretation}
				\\ \hline
				
				\multirow{7}{*}{\textbf{double}}
				&
				\(M_0=0\)
				&
				No nonzero solution
				&
				Inadmissible
				\\[1.5pt]
				&
				\(M_0=1,\ c_s=0\)
				&
				\(M_1\) lies on \eqref{eq:M1-circle-double}
				&
				General pure double-root sector
				\\[1.5pt]
				&
				\(M_0=1,\ c_s=0,\ M_1=0\)
				&
				Special real point on \eqref{eq:M1-circle-double}
				&
				\({\bf y}_R={\bf Y}_0\)
				\\[1.5pt]
				&
				\(M_0=1,\ c_s=0,\ M_1=-2a_J/b_J\)
				&
				The second real point on \eqref{eq:M1-circle-double}
				&
				Nontrivial real confluent branch
				\\[1.5pt]
				&
				\(M_0=1,\ M_1=0,\ c_s\neq0\)
				&
				\(d_s|c_s|^2=0\)
				&
				Inadmissible
				\\[1.5pt]
				&
				\(M_0=1,\ M_1c_s\neq0\)
				&
				Subject to \eqref{eq:double-null-general}
				&
				General mixed double--simple sector
				\\[1.5pt]
				&
				\(
				\begin{gathered}
					M_0=1,\quad M_1c_s\neq0,\\
					|c_s|^2=\dfrac{a_J^2}{b_Jd_s}
				\end{gathered}
				\)
				&
				\(M_1=-a_J/b_J\)
				&
				Saturated boundary of the mixed sector; the admissible
				circle collapses to one point
				\\ \hline
				
				\multirow{7}{*}{\textbf{triple}}
				&
				\(M_0=0\)
				&
				No nonzero solution
				&
				Inadmissible
				\\[1.5pt]
				&
				\(M_0=1,\ M_1=M_2=0\)
				&
				Automatically null
				&
				Canonical sector
				\\[1.5pt]
				&
				\(M_0=1,\ M_2=0,\ M_1\neq0\)
				&
				\(g_{11}|M_1|^2=0\)
				&
				Inadmissible
				\\[1.5pt]
				&
				\(M_0=1,\ M_1=0,\ M_2\neq0\)
				&
				\(M_2\) lies on \eqref{eq:M2-circle-triple}
				&
				General pure second-derivative sector
				\\[1.5pt]
				&
				\(
				\begin{gathered}
					M_0=1,\quad M_1=0,\\
					M_2=-\dfrac{4g_{02}}{g_{22}}\in\mathbb R
				\end{gathered}
				\)
				&
				Unique nonzero real point on
				\eqref{eq:M2-circle-triple}
				&
				Special real second-derivative branch
				\\[1.5pt]
				&
				\(M_0=1,\ M_1M_2\neq0\)
				&
				Subject to \eqref{eq:triple-null-general}
				&
				General full triple-chain sector
				\\[1.5pt]
				&
				\(
				\begin{gathered}
					M_0=1,\quad M_1M_2\neq0,\\
					\Delta_T(M_1)=0
				\end{gathered}
				\)
				&
				\(M_2=-2\beta(M_1)^*/g_{22}\)
				&
				Saturated boundary of the full triple-chain sector; the
				admissible \(M_2\)-circle collapses to one point
				\\ \hline
			\end{tabular}
			
			\endgroup
		\end{center}
		
		The rows describing special real points and saturated boundaries are subcases of the corresponding general sectors and are retained to make the real coefficient branches explicit. In particular, the pure double-root circle contains the two real points \(M_1=0\) and \(M_1=-2a_J/b_J\), while the pure second-derivative circle contains the canonical point \(M_2=0\) and the unique nonzero real point \(M_2=-4g_{02}/g_{22}\).
		
		No independent simple-branch coefficient exists at a pure triple branch point, because all three spatial roots have already coalesced into \(\chi_0\).
	\end{remark}
	
	\section{Nonlocalized Asymptotic Curve}
	\label{sec:nonlocal-asymptotics}
	
	\subsection{Cubic kernel level sets for the pure double-root family}
	\label{subsec:double-nonlocal}
	\label{subsec:nonlocalized-curves}

	We now study the far-field behavior of the general nonzero pure double-root family \(c_s=0, \qquad r\in\mathbb C\setminus\{0\}, \qquad b_J|r|^2+2a_J\operatorname{Re}r=0\). The analysis does not involve the mixed double--simple sector \(c_s\neq0\).
	
	For \(q=e^{{\rm i}\alpha}, \qquad \sin\alpha\neq0\), the imaginary part of the limiting denominator is \(\operatorname{Im} \mathcal D_\alpha^{(r)} = 2\lambda_0 \left[ \zeta_0\cot\alpha - \Gamma_r(\xi,t) \right]\), where \(\Gamma_r(\xi,t) = \gamma_3\xi^3 + \gamma_2\xi^2 + \gamma_1\xi + \gamma_t t + \gamma_0\).
	
	\begin{definition}[Kernel level-set family]
		\label{def:NL-curve-family}
		For each \(\rho\in\mathbb R\), define
		\begin{equation*}
			\mathcal C_\rho^{(\alpha,r)}
			=
			\left\{
			(x,t)\in\mathbb R^2:
			\operatorname{Im}
			\mathcal D_\alpha^{(r)}(x,t)
			=
			\rho
			\right\}.
		\end{equation*}
		Equivalently,
		\begin{equation}
			\label{eq:C-rho-implicit}
			\gamma_3\xi^3
			+
			\gamma_2\xi^2
			+
			\gamma_1\xi
			+
			\gamma_t t
			+
			\gamma_0
			=
			\zeta_0\cot\alpha
			-
			\frac{\rho}{2\lambda_0}.
		\end{equation}
		
		An unbounded connected component is called an \emph{asymptotically nonlocalized component} if \(|\xi|\longrightarrow\infty\) along that component.
	\end{definition}
	
	\begin{lemma}
		\label{lem:C-rho-unbounded-graph}
		Assume \(\gamma_t\neq0\). Then every level set \(\mathcal C_\rho^{(\alpha,r)}\) is a connected graph parameterized by \(\xi\in\mathbb R\):
		\begin{equation}
			\label{eq:C-rho-parametric}
			\begin{aligned}
				t_\rho(\xi)
				={}&
				\frac{
					\zeta_0\cot\alpha
					-\rho/(2\lambda_0)
					-\gamma_0
					-\gamma_1\xi
					-\gamma_2\xi^2
					-\gamma_3\xi^3
				}{
					\gamma_t
				},
				\\
				x_\rho(\xi)
				={}&
				\xi-\nu_1t_\rho(\xi).
			\end{aligned}
		\end{equation}
		Since \(\gamma_3\neq0\), one has
		\begin{equation*}
			\begin{aligned}
				t_\rho(\xi)
				&=
				-\frac{\gamma_3}{\gamma_t}\xi^3
				+
				O(\xi^2),
				\\
				x_\rho(\xi)
				&=
				\frac{\nu_1\gamma_3}{\gamma_t}\xi^3
				+
				O(\xi^2)
			\end{aligned}
		\end{equation*}
		as \(|\xi|\to\infty\).
	\end{lemma}
	
	\begin{proof}
		When \(\gamma_t\neq0\), \eqref{eq:C-rho-implicit} can be solved uniquely for \(t\) for every \(\xi\in\mathbb R\), giving \eqref{eq:C-rho-parametric}. The graph is continuous and contains points with arbitrarily large \(|\xi|\).
		
		The cubic asymptotics follow from
		\[
		\gamma_3 = -\frac{\chi_0}{3}|r|^2 \neq0.
		\]
	\end{proof}
	
	For the imaginary-axis approach \(\alpha=\pi/2\), the central level set is \(\mathcal C_0^{(\pi/2,r)}: \qquad \gamma_3\xi^3 + \gamma_2\xi^2 + \gamma_1\xi + \gamma_t t + \gamma_0 = 0\).
	
	\begin{proposition}
		\label{prop:NL-asymptotic-limit}
		Assume \(\gamma_t\neq0\). Along every level-set graph \(\mathcal C_\rho^{(\alpha,r)}\),
		\begin{equation}
			\label{eq:NL-asymptotic-normalized}
			\lim_{\substack{|\xi|\to\infty\\
					(x,t)\in\mathcal C_\rho^{(\alpha,r)}}}
			\frac{
				u_j^{(R,2;\alpha,r)}(x,t)
			}{
				u_{j,0}(x,t)
			}
			=
			\mathcal P_j^{(D)}
			:=
			1-2h_j(\chi_0),
			\qquad
			j=1,2.
		\end{equation}
		Thus the pure double-root modulation does not approach the plane-wave background along these unbounded curves.
	\end{proposition}
	
	\begin{proof}
		As \(|\xi|\to\infty\), \(\mathcal L_r(\xi) = |r|^2\xi^2+O(|\xi|)\). Therefore, on a level set, \(\mathcal D_\alpha^{(r)} = -\lambda_0|r|^2\xi^2 + O(|\xi|) + {\rm i}\rho\). Moreover, \(\mathcal N_j^{(r)} = 2\lambda_0 |r|^2h_j(\chi_0)\xi^2 + O(|\xi|)\). Dividing these two expressions gives
		\[
		\frac{ \mathcal N_j^{(r)} }{ \mathcal D_\alpha^{(r)} } \longrightarrow -2h_j(\chi_0),
		\]
		which proves \eqref{eq:NL-asymptotic-normalized}.
	\end{proof}
	
	\begin{corollary}
		\label{cor:double-plateau-reciprocity}
		The two background-normalized double-root plateau factors satisfy
		\begin{equation}
			\label{eq:double-plateau-explicit-factors}
			\mathcal P_1^{(D)}
			=
			\frac{\chi_0-k}{\chi_0+k},
			\qquad
			\mathcal P_2^{(D)}
			=
			\frac{\chi_0+k}{\chi_0-k},
		\end{equation}
		and hence obey the exact reciprocal law
		\begin{equation}
			\label{eq:double-plateau-reciprocal-law}
			\mathcal P_1^{(D)}\mathcal P_2^{(D)}=1.
		\end{equation}
		Neither factor equals \(1\) or \(-1\). Consequently,
		\begin{equation}
			\label{eq:double-intensity-plateau-product}
			|\mathcal P_1^{(D)}|^2
			|\mathcal P_2^{(D)}|^2
			=1,
		\end{equation}
		and precisely one normalized amplitude plateau is larger than one while the other is smaller than one.
	\end{corollary}
	
	\begin{proof}
		Substitution of
		\[
		h_1(\chi_0)=\frac{k}{k+\chi_0},
		\qquad
		h_2(\chi_0)=\frac{k}{k-\chi_0}
		\]
		into \(\mathcal P_j^{(D)}=1-2h_j(\chi_0)\) gives
		\[
		\begin{aligned}
			\mathcal P_1^{(D)}
			&=
			1-\frac{2k}{k+\chi_0}
			=
			\frac{\chi_0-k}{\chi_0+k},
			\\
			\mathcal P_2^{(D)}
			&=
			1-\frac{2k}{k-\chi_0}
			=
			\frac{\chi_0+k}{\chi_0-k}.
		\end{aligned}
		\]
		The denominators are nonzero because \(\chi_0\neq\pm k\). Multiplication gives \eqref{eq:double-plateau-reciprocal-law}.
		
		If either factor were \(1\), then \(k=0\), contrary to \eqref{eq:physical-region}. If either factor were \(-1\), then \(\chi_0=0\), which is excluded for an effective double branch point. Thus both factors are real and have modulus different from one. The reciprocal law then implies that one modulus is greater than one and the other is less than one; squaring gives \eqref{eq:double-intensity-plateau-product}.
	\end{proof}
	
	Consequently,
	\begin{equation*}
		\lim_{\substack{|\xi|\to\infty\\
				(x,t)\in\mathcal C_\rho^{(\alpha,r)}}}
		\left|
		u_j^{(R,2;\alpha,r)}(x,t)
		\right|
		=
		a_j
		\left|
		1-2h_j(\chi_0)
		\right|,
		\qquad
		j=1,2.
	\end{equation*}
	The plateau is independent of \(r\), \(\alpha\), and \(\rho\). These parameters modify only the position and lower-order geometry of the level-set graph.
	
	\begin{remark}
		\label{rem:double-linear-directions}
		Assume \(\gamma_t\neq0\). Along a straight ray \(x=vt+O(1), \qquad |t|\to\infty\), if \(v+\nu_1\neq0\), then \(\Gamma_r(\xi,t)=O(t^3)\), whereas the numerator is \(O(t^2)\). Hence the normalized correction is \(O(t^{-1})\).
		
		If \(v=-\nu_1\), then \(\xi=O(1)\), while \(\Gamma_r(\xi,t) = \gamma_t t+O(1)\), and the correction again tends to zero. Thus the solution returns to the plane-wave background along generic linear directions while retaining a finite nonbackground plateau on the cubic level-set graphs.
	\end{remark}
	
	\begin{remark}
		\label{rem:exceptional-NL-level-sets}
		If \(\gamma_t=0\), the level-set equation reduces to a cubic equation for \(\xi\):
		\[
		\gamma_3\xi^3 + \gamma_2\xi^2 + \gamma_1\xi + \gamma_0 = \zeta_0\cot\alpha - \frac{\rho}{2\lambda_0}.
		\]
		Since \(\gamma_3\neq0\), it has only finitely many real roots. Each root gives a characteristic line with constant \(\xi\), so no component satisfies \(|\xi|\to\infty\). The plateau theorem above therefore applies precisely to the generic case \(\gamma_t\neq0\).
	\end{remark}
	
	\begin{remark}
		\label{rem:approach-direction-NL-curves}
		For the real representative \(r=r_{\rm c}\), one has \(\gamma_2=\gamma_0=0\), and the level-set equation is centered and odd in \(\xi\), apart from the term \(\gamma_t t\).
		
		For a genuinely complex \(r\), the terms
		\[
		\gamma_2\xi^2 \qquad\text{and}\qquad \gamma_0
		\]
		shift the curve and destroy its odd symmetry. Nevertheless, the leading cubic coefficient remains nonzero, and the asymptotic plateau is unchanged.
	\end{remark}
	
	\subsection{The characteristic-line plateau for triple-root solutions}
	\label{subsec:triple-nonlocal}
	\label{subsec:triple-nonlocal-curves}

	We now restrict attention to the purely imaginary approach
	\[
	\eta={\rm i}\varepsilon, \qquad \varepsilon\in\mathbb R, \qquad \varepsilon\to0,
	\]
	and consider an admissible triple-root leading vector \({\bf y}_R = {\bf Y}_0 + M_1{\bf Y}_1 + \frac{M_2}{2}{\bf Y}_2, \qquad M_2\neq0\), where \(M_1\) may vanish and the coefficients satisfy the null condition \eqref{eq:triple-null-general}.
	
	Let
	\[
	q_{\rm tri}=\frac{\chi_0}{k}.
	\]
	The triple-root relations imply \(\nu_1 = \nu'(\chi_0) = \frac{1}{3k^2}, \qquad \nu_2 = \nu''(\chi_0) = -\frac{2}{3q_{\rm tri}k^3} \neq0\). We introduce the characteristic coordinate \(\xi = x+\nu_1t = x+\frac{t}{3k^2}\). The distinguished triple-root characteristic line is therefore
	\begin{equation*}
		\Gamma_T
		=
		\left\{
		(x,t)\in\mathbb R^2:
		\xi=0
		\right\}
		=
		\left\{
		(x,t)\in\mathbb R^2:
		x+\frac{t}{3k^2}=0
		\right\}.
	\end{equation*}
	
	For the asymptotic calculation, define \(\nu_\ell = \nu^{(\ell)}(\chi_0), \qquad \ell=1,\ldots,5\). Since
	\[
	\nu(\chi) = \frac{S}{2} + \frac{2-kR}{2\chi},
	\]
	one has
	\begin{equation*}
		\begin{aligned}
			\nu_2
			&=
			\frac{2-kR}{\chi_0^3},
			&
			\nu_3
			&=
			-\frac{3(2-kR)}{\chi_0^4},
			\\
			\nu_4
			&=
			\frac{12(2-kR)}{\chi_0^5},
			&
			\nu_5
			&=
			-\frac{60(2-kR)}{\chi_0^6}.
		\end{aligned}
	\end{equation*}
	
	Set
	\begin{equation*}
		{\bf v}_0
		=
		\begin{pmatrix}
			\lambda_0\\
			a_1h_1(\chi_0)\\
			a_2h_2(\chi_0)
		\end{pmatrix},
		\qquad
		{\bf v}_n
		=
		\begin{pmatrix}
			0\\
			a_1h_1^{(n)}(\chi_0)\\
			a_2h_2^{(n)}(\chi_0)
		\end{pmatrix},
		\qquad
		n\ge1,
	\end{equation*}
	and introduce the diagonal unitary matrix
	\begin{equation*}
		\mathscr U(x,t)
		=
		\operatorname{diag}
		\left(
		E_0,\,
		\frac{u_{1,0}}{a_1}E_0,\,
		\frac{u_{2,0}}{a_2}E_0
		\right),
		\qquad
		E_0=E(\chi_0).
	\end{equation*}
	Since \(\mathscr U\) is diagonal and unitary, it commutes with both \(J\) and \(K\), and hence it does not affect any \(J\)- or \(K\)-inner product.
	
	Along \(\Gamma_T\), where \(\xi=0\), the normalized exponential derivatives satisfy
	\begin{equation}
		\label{eq:triple-E-line}
		\begin{aligned}
			\frac{E_1}{E_0}
			&=0,
			\\
			\frac{E_2}{E_0}
			&=
			{\rm i}\nu_2t,
			\\
			\frac{E_3}{E_0}
			&=
			{\rm i}\nu_3t,
			\\
			\frac{E_4}{E_0}
			&=
			-3\nu_2^2t^2
			+
			{\rm i}\nu_4t,
			\\
			\frac{E_5}{E_0}
			&=
			-10\nu_2\nu_3t^2
			+
			{\rm i}\nu_5t.
		\end{aligned}
	\end{equation}
	Consequently, \({\bf Y}_n = \mathscr U \sum_{m=0}^n \binom{n}{m} \frac{E_{n-m}}{E_0} {\bf v}_m\).
	
	Define \(c_T = \frac{{\rm i}\nu_2M_2}{2}, \qquad {\bf r}_T = {\bf v}_0 + M_1{\bf v}_1 + \frac{M_2}{2}{\bf v}_2\). Then the leading eigenvector has the exact representation
	\begin{equation}
		\label{eq:yR-line-expansion}
		{\bf y}_R\big|_{\Gamma_T}
		=
		\mathscr U
		\left(
		c_Tt\,{\bf v}_0
		+
		{\bf r}_T
		\right).
	\end{equation}
	Because \(M_2\neq0\) and \(\nu_2\neq0\), \(c_T\neq0\).
	
	To expand the correction vector \({\bf y}_\mu\), define
	\begin{equation*}
		\begin{aligned}
			{\bf p}_{T,2}
			={}&
			-
			\left(
			\frac{\nu_2^2\mathcal M_4}{8}
			+
			\frac{\nu_2\nu_3\mathcal M_5}{12}
			\right)
			{\bf v}_0
			-
			\frac{\nu_2^2\mathcal M_5}{8}
			{\bf v}_1,
		\end{aligned}
	\end{equation*}
	and
	\begin{equation*}
		\begin{aligned}
			{\bf p}_{T,1}
			={}&
			\frac{{\rm i}\nu_2M_2}{4\lambda_0}
			\begin{pmatrix}
				1\\0\\0
			\end{pmatrix}
			\\
			&+
			\frac{{\rm i}\mathcal M_3}{6}
			\left(
			\nu_3{\bf v}_0
			+
			3\nu_2{\bf v}_1
			\right)
			\\
			&+
			\frac{{\rm i}\mathcal M_4}{24}
			\left(
			\nu_4{\bf v}_0
			+
			4\nu_3{\bf v}_1
			+
			6\nu_2{\bf v}_2
			\right)
			\\
			&+
			\frac{{\rm i}\mathcal M_5}{120}
			\left(
			\nu_5{\bf v}_0
			+
			5\nu_4{\bf v}_1
			+
			10\nu_3{\bf v}_2
			+
			10\nu_2{\bf v}_3
			\right).
		\end{aligned}
	\end{equation*}
	Using \eqref{eq:ymu-triple} and \eqref{eq:triple-E-line}, one obtains \({\bf y}_\mu\big|_{\Gamma_T} = \mathscr U \left( t^2{\bf p}_{T,2} + t{\bf p}_{T,1} + {\bf p}_{T,0} \right)\), where \({\bf p}_{T,0}\) is independent of \(t\). Its explicit form is not needed for the leading asymptotic balance.
	
	Define \(\mathcal Z_T(t) = \left. \langle{\bf y}_R|J|{\bf y}_\mu\rangle \right|_{\Gamma_T}\). At a triple branch point,
	\[
	\langle{\bf v}_0|J|{\bf v}_0\rangle=0, \qquad \langle{\bf v}_0|J|{\bf v}_1\rangle=0.
	\]
	Since \({\bf p}_{T,2}\) belongs to \(\operatorname{span}\{{\bf v}_0,{\bf v}_1\}\), the nominal \(t^3\) term in \(\mathcal Z_T\) vanishes: \(\langle{\bf v}_0|J|{\bf p}_{T,2}\rangle=0\). It follows that
	\begin{equation}
		\label{eq:ZT-asymptotic}
		\mathcal Z_T(t)
		=
		\zeta_{T,2}t^2
		+
		O(|t|),
		\qquad
		|t|\to\infty,
	\end{equation}
	where \(\zeta_{T,2} = c_T^* \langle{\bf v}_0|J|{\bf p}_{T,1}\rangle + \langle{\bf r}_T|J|{\bf p}_{T,2}\rangle\).
	
	We also define \(\kappa_0 = \langle{\bf v}_0|K|{\bf v}_0\rangle = \mu_0 + Ah_1(\chi_0)^2 + Bh_2(\chi_0)^2\). The \(J\)-nullity of \({\bf v}_0\) gives \(\mu_0 = Ah_1(\chi_0)^2 + Bh_2(\chi_0)^2\), and therefore \(\kappa_0=2\mu_0>0\). Equation~\eqref{eq:yR-line-expansion} then gives
	\begin{equation}
		\label{eq:KR-line-asymptotic}
		\left.
		\langle{\bf y}_R|K|{\bf y}_R\rangle
		\right|_{\Gamma_T}
		=
		|c_T|^2\kappa_0t^2
		+
		O(|t|).
	\end{equation}
	
	For later use, set \(d_T = - \frac{|c_T|^2\kappa_0}{2\lambda_0} - 2{\rm i}\lambda_0 \operatorname{Im}\zeta_{T,2}\). Its real part is \(\operatorname{Re}d_T = - \frac{|c_T|^2\kappa_0}{2\lambda_0} <0\), and hence \(d_T\neq0\).
	
	\begin{proposition}
		\label{prop:triple-nonlocal-corridor}
		Let \(M_2\neq0\), and suppose that \(M_1,M_2\) satisfy the triple-root null condition \eqref{eq:triple-null-general}. For the imaginary-axis real-spectrum limit, one has
		\begin{equation}
			\label{eq:mR-triple-line-asymptotic}
			\left.
			m_R^{(3)}
			\right|_{\Gamma_T}
			=
			d_Tt^2
			+
			O(|t|),
			\qquad
			|t|\to\infty.
		\end{equation}
		
		Let \(\Phi_{1,R}=\phi_R, \qquad \Phi_{2,R}=\varphi_R\). Then
		\begin{equation}
			\label{eq:triple-numerator-line}
			\left.
			\frac{
				2\Phi_{j,R}\psi_R^*
			}{
				u_{j,0}
			}
			\right|_{\Gamma_T}
			=
			2|c_T|^2\lambda_0h_j(\chi_0)t^2
			+
			O(|t|),
			\qquad
			j=1,2.
		\end{equation}
		Consequently,
		\begin{equation}
			\label{eq:triple-line-normalized-limit}
			\lim_{\substack{|t|\to\infty\\(x,t)\in\Gamma_T}}
			\frac{
				u_j^{(R,3;{\rm i})}(x,t)
			}{
				u_{j,0}(x,t)
			}
			=
			\mathcal P_j^{(T)}
			:=
			1+
			\frac{
				2|c_T|^2\lambda_0h_j(\chi_0)
			}{
				d_T
			},
			\qquad
			j=1,2.
		\end{equation}
		In particular, \(\mathcal P_j^{(T)}\neq1, \qquad j=1,2\). Thus the rational modulation does not decay to the plane-wave background along the unbounded characteristic line \(\Gamma_T\).
	\end{proposition}
	
	\begin{proof}
		From \eqref{eq:mR-triple},
		\[
		m_R^{(3)} = \lambda_0 \left( \mathcal Z_T^* - \mathcal Z_T \right) - \frac{ \langle{\bf y}_R|K|{\bf y}_R\rangle }{ 2\lambda_0 }.
		\]
		Substitution of \eqref{eq:ZT-asymptotic} and \eqref{eq:KR-line-asymptotic} gives
		\[
		\begin{aligned}
			m_R^{(3)}
			={}&
			-2{\rm i}\lambda_0
			\operatorname{Im}\zeta_{T,2}\,t^2
			-
			\frac{
				|c_T|^2\kappa_0
			}{
				2\lambda_0
			}t^2
			+
			O(|t|),
		\end{aligned}
		\]
		which proves \eqref{eq:mR-triple-line-asymptotic}.
		
		From \eqref{eq:yR-line-expansion}, the leading components of \({\bf y}_R\) are \(\psi_R = c_Tt\,\lambda_0E_0+O(1)\), \(\phi_R = c_Tt\,u_{1,0}h_1(\chi_0)E_0+O(1)\), and \(\varphi_R = c_Tt\,u_{2,0}h_2(\chi_0)E_0+O(1)\). Since \(|E_0|=1\), these expressions give \eqref{eq:triple-numerator-line}. Dividing \eqref{eq:triple-numerator-line} by \eqref{eq:mR-triple-line-asymptotic} proves \eqref{eq:triple-line-normalized-limit}.
		
		Finally, \(c_T\neq0, \qquad \lambda_0>0, \qquad h_j(\chi_0)\neq0, \qquad d_T\neq0\). Hence the correction term in \eqref{eq:triple-line-normalized-limit} is nonzero, and \(\mathcal P_j^{(T)}\neq1\).
	\end{proof}
	
	\begin{corollary}
		\label{cor:triple-plateau-component-constraint}
		Define
		\begin{equation}
			\label{eq:triple-common-plateau-parameter}
			C_T
			=
			\frac{2|c_T|^2\lambda_0}{d_T}.
		\end{equation}
		Then every admissible triple-root plateau with \(M_2\neq0\) satisfies
		\begin{equation}
			\label{eq:triple-plateau-component-constraint}
			(1+q_{\rm tri})
			\bigl(\mathcal P_1^{(T)}-1\bigr)
			=
			(1-q_{\rm tri})
			\bigl(\mathcal P_2^{(T)}-1\bigr)
			=
			C_T.
		\end{equation}
		Thus the two complex plateau factors are linked by one common parameter and cannot be prescribed independently.
	\end{corollary}
	
	\begin{proof}
		Since \(q_{\rm tri}=\chi_0/k\) and \(\chi_0\neq\pm k\), one has
		\[
		h_1(\chi_0)
		=
		\frac{k}{k+\chi_0}
		=
		\frac{1}{1+q_{\rm tri}},
		\qquad
		h_2(\chi_0)
		=
		\frac{k}{k-\chi_0}
		=
		\frac{1}{1-q_{\rm tri}}.
		\]
		Equation~\eqref{eq:triple-line-normalized-limit} can therefore be written as
		\[
		\mathcal P_1^{(T)}-1
		=
		\frac{C_T}{1+q_{\rm tri}},
		\qquad
		\mathcal P_2^{(T)}-1
		=
		\frac{C_T}{1-q_{\rm tri}}.
		\]
		Multiplication by the corresponding denominators proves \eqref{eq:triple-plateau-component-constraint}.
	\end{proof}
	
	Consequently, the physical amplitudes satisfy
	\begin{equation*}
		\lim_{\substack{|t|\to\infty\\(x,t)\in\Gamma_T}}
		\left|
		u_j^{(R,3;{\rm i})}(x,t)
		\right|
		=
		a_j
		\left|
		\mathcal P_j^{(T)}
		\right|,
		\qquad
		j=1,2.
	\end{equation*}
	The direction of \(\Gamma_T\) depends only on the triple-root background parameter \(k\), whereas the asymptotic plateau \(\mathcal P_j^{(T)}\) may depend on the admissible internal coefficients \(M_1\) and \(M_2\).
	
	\begin{remark}
		\label{rem:triple-line-scope}
		The result above establishes nonlocalization along the distinguished straight characteristic
		\[
		x+\frac{t}{3k^2}=0.
		\]
		It does not assert that this line exhausts all possible asymptotic trajectories. A complete classification of the behavior in other joint limits of \(x\) and \(t\) is not needed for the present construction and is not pursued here.
	\end{remark}
	\subsection{Comparison of the two nonlocal mechanisms}
	\label{subsec:nonlocal-comparison}
	
	The double- and triple-root limits are both nonlocalized because the Darboux numerator and denominator retain the same leading order along distinguished unbounded sets. Their geometric mechanisms are nevertheless different, as summarized in Table~\ref{tab:nonlocal-comparison}.
	
	\begin{table}[htbp]
		\centering
		\caption{Comparison of the nonlocalized asymptotic mechanisms.}
		\label{tab:nonlocal-comparison}
		\begin{tabularx}{\textwidth}{>{\raggedright\arraybackslash}p{0.22\textwidth} X X}
			\toprule
			Feature & Pure double-root family & Triple-root family with \(M_2\neq0\) \tabularnewline
			\midrule
			Distinguished unbounded set & Cubic kernel level-set graphs when \(\gamma_t(r)\neq0\) & Characteristic line \(x+t/(3k^2)=0\) \tabularnewline
			Leading balance & Numerator and denominator are both \(O(\xi^2)\) & Numerator and denominator are both \(O(t^2)\) \tabularnewline
			Normalized plateau & Reciprocal real factors \(\mathcal P_1^{(D)}\mathcal P_2^{(D)}=1\), independent of \(r\), the approach angle, and the level parameter & Complex factors satisfying \((1+q_{\rm tri})(\mathcal P_1^{(T)}-1)=(1-q_{\rm tri})(\mathcal P_2^{(T)}-1)\) \tabularnewline
			Scope of the result & Complete nonzero pure double-root null-circle family, subject to the graph condition & Only the stated characteristic line; no exhaustive classification of other directions is claimed \tabularnewline
			\bottomrule
		\end{tabularx}
	\end{table}
	
	Thus the double-root nonlocality is organized by a one-parameter family of cubic level sets, whereas the triple-root calculation identifies a distinguished linear characteristic. This difference reflects the distinct highest-order polynomial structures of the two confluent chains.
	
	\section{Spectral Transitions and Representative Solutions on the Slice
		\(A=5\), \(B=1\)}
	\label{sec:examples}
	\label{sec:example}
	
	This section illustrates how the spatial-root structure established
	above is reflected in the corresponding Darboux solutions. We fix the
	unequal-amplitude background
	\[
	A=5,\qquad B=1,
	\]
	for which the positive \(k\)-axis contains both a
	persistent-zero-branch intersection and a nearby triple-root
	transition. We first trace the spatial roots and their associated
	spectral values across these transitions, and then compare
	representative double- and triple-root solutions obtained from the
	admissible coefficient sectors. Finally, the nonlocalized curves
	derived in Section~\ref{sec:nonlocal-asymptotics} are superimposed on
	selected exact-solution profiles to illustrate their agreement with
	the preceding asymptotic analysis.
	
	\subsection{Spatial-root transitions along the \(A=5\) slice}
	\label{subsec:A5-spectral-transitions}
	
	For \(A=5\) and \(B=1\), the scaled double-root equation
	\eqref{eq:P-y} becomes
	\begin{equation}
		\label{eq:P-A5}
		P_k(y)
		=
		y^4-6ky^3+(6k-2)y^2+1-2k=0,
		\qquad
		\chi=ky.
	\end{equation}
	Two distinguished values divide the positive \(k\)-axis. The first is
	\begin{equation}
		\label{eq:k0-A5}
		k_0=\frac12,
	\end{equation}
	at which the quartic \eqref{eq:P-A5} develops a double zero at
	\(y=0\), equivalently \(\chi=0\). The second is the triple-root value
	\begin{equation}
		\label{eq:ktri-A5}
		k_{\rm tri}
		=
		\frac{2}{A+B}
		\left(
		q_{\rm tri}+\frac{1}{3q_{\rm tri}}
		\right),
		\qquad
		q_{\rm tri}
		=
		\tanh\left[
		\frac13\operatorname{arctanh}
		\left(
		\frac{A-B}{A+B}
		\right)
		\right].
	\end{equation}
	For the present amplitudes, the triple-root data are
	\begin{equation}
		\label{eq:A5-triple-data}
		\begin{aligned}
			q_{\rm tri}
			&\approx0.2619860695,
			&
			k_{\rm tri}
			&\approx0.5114394501,
			\\
			\chi_{\rm tri}
			=q_{\rm tri}k_{\rm tri}
			&\approx0.1339900113,
			&
			\mu_{\rm tri}
			=\frac{2}{3\chi_{\rm tri}}
			&\approx4.9754952637.
		\end{aligned}
	\end{equation}
	The two transition values are close,
	\[
	k_{\rm tri}-k_0\approx0.0114394501,
	\]
	so the interval in which all roots of \eqref{eq:P-A5} are real is
	narrow. The locations of the two transition curves and the three
	representative points used below are shown in
	Fig.~\ref{fig:double-chi-region}.
	
	\begin{figure}[htbp]
		\centering
		\includegraphics[width=0.62\textwidth]{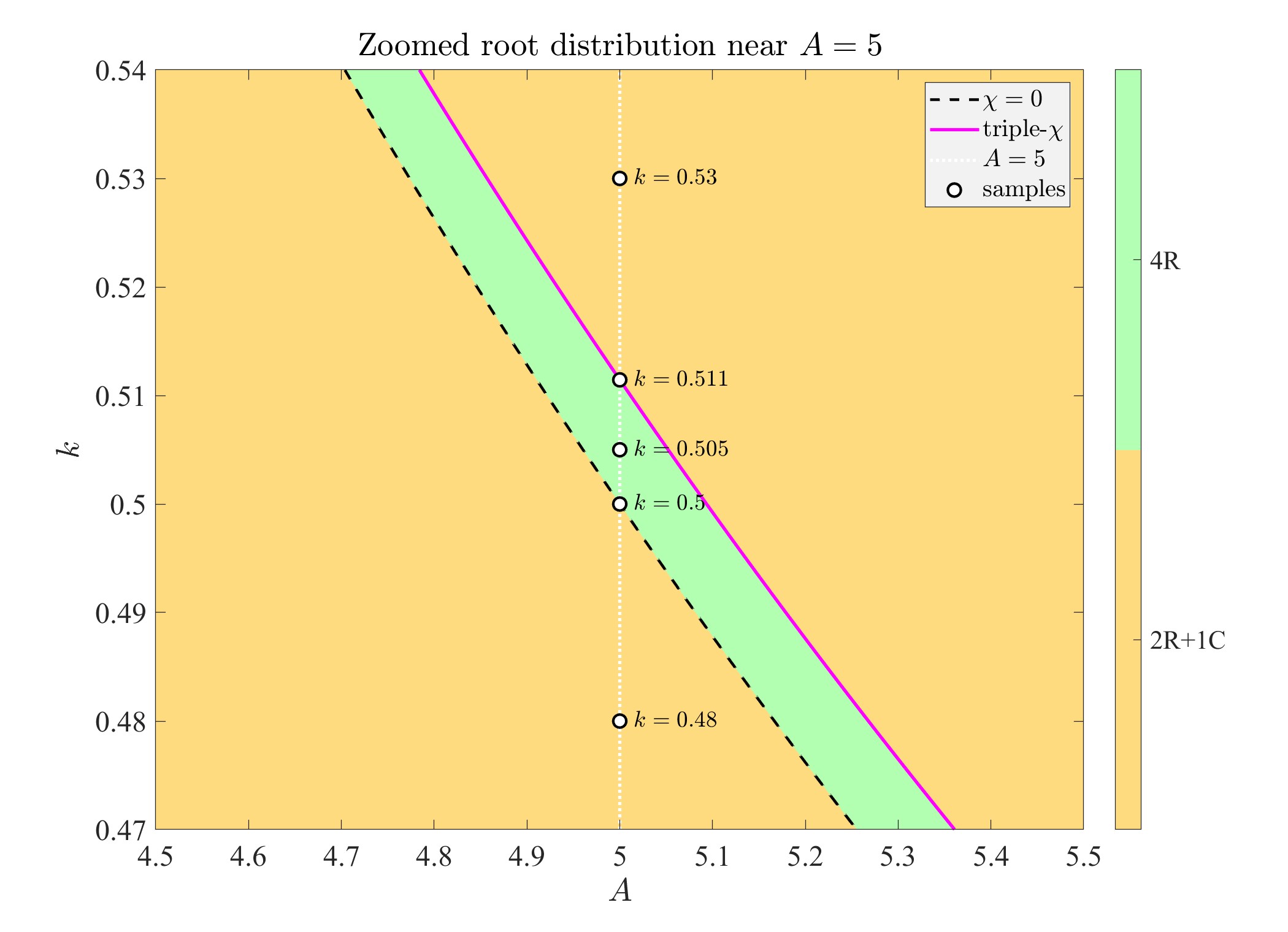}
		\caption{
			Root-distribution regions of the double-root equation in the
			\((A,k)\)-plane with \(B=1\). The dashed curve represents the
			persistent-zero-branch condition \(2-(A-B)k=0\), whereas the
			solid curve represents the triple-root boundary. The vertical
			line marks the slice \(A=5\). The three sample points on this
			slice correspond to \(k=0.4\), \(k=0.505\), and \(k=0.6\).
		}
		\label{fig:double-chi-region}
	\end{figure}
	
	Along the slice \(A=5\), the root configurations of
	\eqref{eq:P-A5} are summarized as follows:
	\begin{equation}
		\label{eq:A5-root-configurations}
		\renewcommand{\arraystretch}{1.35}
		\begin{array}{rcl}
			0<k<k_0
			&:&
			\text{two real roots and one nonreal conjugate pair},
			\\
			k=k_0
			&:&
			\text{a double zero at \(\chi=0\) and two effective real roots},
			\\
			k_0<k<k_{\rm tri}
			&:&
			\text{four distinct real roots},
			\\
			k=k_{\rm tri}
			&:&
			\text{a double zero of the branch-condition quartic,
				corresponding to a triple spatial root},
			\\
			k>k_{\rm tri}
			&:&
			\text{two real roots and one nonreal conjugate pair}.
		\end{array}
	\end{equation}
	At \(k=k_{\rm tri}\), the repeated zero of
	\eqref{eq:P-A5} occurs at \(\chi=\chi_{\rm tri}\), while the two
	remaining zeros of the quartic are real and simple. Thus
	\[
	k_0<k<k_{\rm tri}
	\]
	is a narrow four-real-root window, whereas the two exterior regions
	have mixed real--complex root configurations.
	
	The point \(k=k_0\) requires a separate interpretation. At
	\(k=1/2\), equation \eqref{eq:P-A5} factorizes as
	\begin{equation}
		\label{eq:P-A5-k0-factorization}
		P_{1/2}(y)
		=
		y^2(y^2-3y+1).
	\end{equation}
	The double zero \(y=0\), equivalently \(\chi=0\), does not represent
	an effective stationary point of the reduced rational map
	\(\chi\mapsto\Lambda(\chi)\). Instead, it records the intersection of
	the persistent zero branch with a movable spatial branch. Indeed,
	\((A-B)k_0=2\), and at the spectral value \(\mu=S=6\), the full
	fixed-\(\mu\) characteristic equation becomes
	\begin{equation}
		\label{eq:F-A5-k0-factorization}
		F(\chi;6)=2\chi^2(3\chi-1).
	\end{equation}
	The double zero of \eqref{eq:F-A5-k0-factorization} is therefore a
	branch intersection rather than a ramification point of the reduced
	map. The persistent zero branch remains a genuine eigenfunction branch
	and is represented by \eqref{eq:zero-branch-eigenfunction}.
	
	By contrast, the other two zeros of
	\eqref{eq:P-A5-k0-factorization},
	\begin{equation}
		\label{eq:A5-k0-effective-roots}
		\chi_\pm
		=
		\frac{3\pm\sqrt5}{4},
	\end{equation}
	are effective real branch points.
	
	For later comparison, Table~\ref{tab:A5-spectral-data} lists the
	spatial roots and their corresponding squared spectral values at the
	three representative values of \(k\). For a nonreal conjugate pair,
	the two signs are correlated as displayed.
	
	\begin{table}[htbp]
		\centering
		\caption{
			Spatial branch points and their squared spectral values for
			\(A=5\), \(B=1\), at the three representative values of \(k\).
		}
		\label{tab:A5-spectral-data}
		\renewcommand{\arraystretch}{1.25}
		\setlength{\tabcolsep}{7pt}
		\begin{tabular}{ccl}
			\toprule
			\(k\)
			&
			Spatial root \(\chi_0\)
			&
			Squared spectral value \(\mu_0=\Lambda(\chi_0)\)
			\\
			\midrule
			\multirow{3}{*}{\(0.4\)}
			&
			\(-0.07228038\pm0.14374134{\rm i}\)
			&
			\(4.77639953\mp4.23009618{\rm i}\)
			\\
			&
			\(0.22483038\)
			&
			\(7.26348867\)
			\\
			&
			\(0.87973038\)
			&
			\(1.18371227\)
			\\
			\midrule
			\multirow{4}{*}{\(0.505\)}
			&
			\(-0.04423862\)
			&
			\(6.85160830\)
			\\
			&
			\(0.06140505\)
			&
			\(5.27066031\)
			\\
			&
			\(0.17955147\)
			&
			\(5.12886259\)
			\\
			&
			\(1.33343210\)
			&
			\(0.74886880\)
			\\
			\midrule
			\multirow{3}{*}{\(0.6\)}
			&
			\(-0.16446338\)
			&
			\(10.10507121\)
			\\
			&
			\(0.23559662\pm0.17185743{\rm i}\)
			&
			\(3.68332672\pm0.73514830{\rm i}\)
			\\
			&
			\(1.85327014\)
			&
			\(0.52827535\)
			\\
			\bottomrule
		\end{tabular}
	\end{table}
	
	The table confirms that the same root count in the two exterior
	regions does not imply identical spectral data. In particular, the
	locations of both the real roots and the nonreal conjugate pair change
	substantially after passage through the four-real-root window and the
	triple-root transition.
	
	\subsection{Double-root solutions across the spectral transitions}
	\label{subsec:A5-double-solutions}
	
	We now compare representative solutions associated with the branch
	points listed in Table~\ref{tab:A5-spectral-data}. Every effective real
	root automatically satisfies \(\mu_0>0\) by
	Proposition~\ref{prop:real-double-positive} and therefore belongs to
	the real-spectrum limiting construction. Each nonreal root corresponds
	to a nonreal squared spectral value and gives a regular confluent
	one-fold Darboux solution.
	
	The figures below display the changes associated with the spatial-root
	branches. They do not exhaust the full dependence on the admissible
	confluent coefficients, which has already been characterized by
	\eqref{eq:M1-circle-double} and \eqref{eq:double-null-general}.
	
	\paragraph{The lower-\(k\) real--complex regime: \(k=0.4\).}
	
	At \(k=0.4\), the branch-condition quartic has two real roots and one
	nonreal conjugate pair. The two real roots generate real-spectrum
	limiting solutions satisfying the \(J\)-null condition, whereas the
	nonreal pair generates regular confluent Darboux solutions without a
	real-spectrum limiting procedure. The corresponding real- and
	nonreal-spectrum profiles are compared in
	Fig.~\ref{fig:tr1-tc1}.
	
	\begin{figure}[htbp]
		\centering
		\begin{subfigure}{0.49\textwidth}
			\centering
			\includegraphics[width=\textwidth]{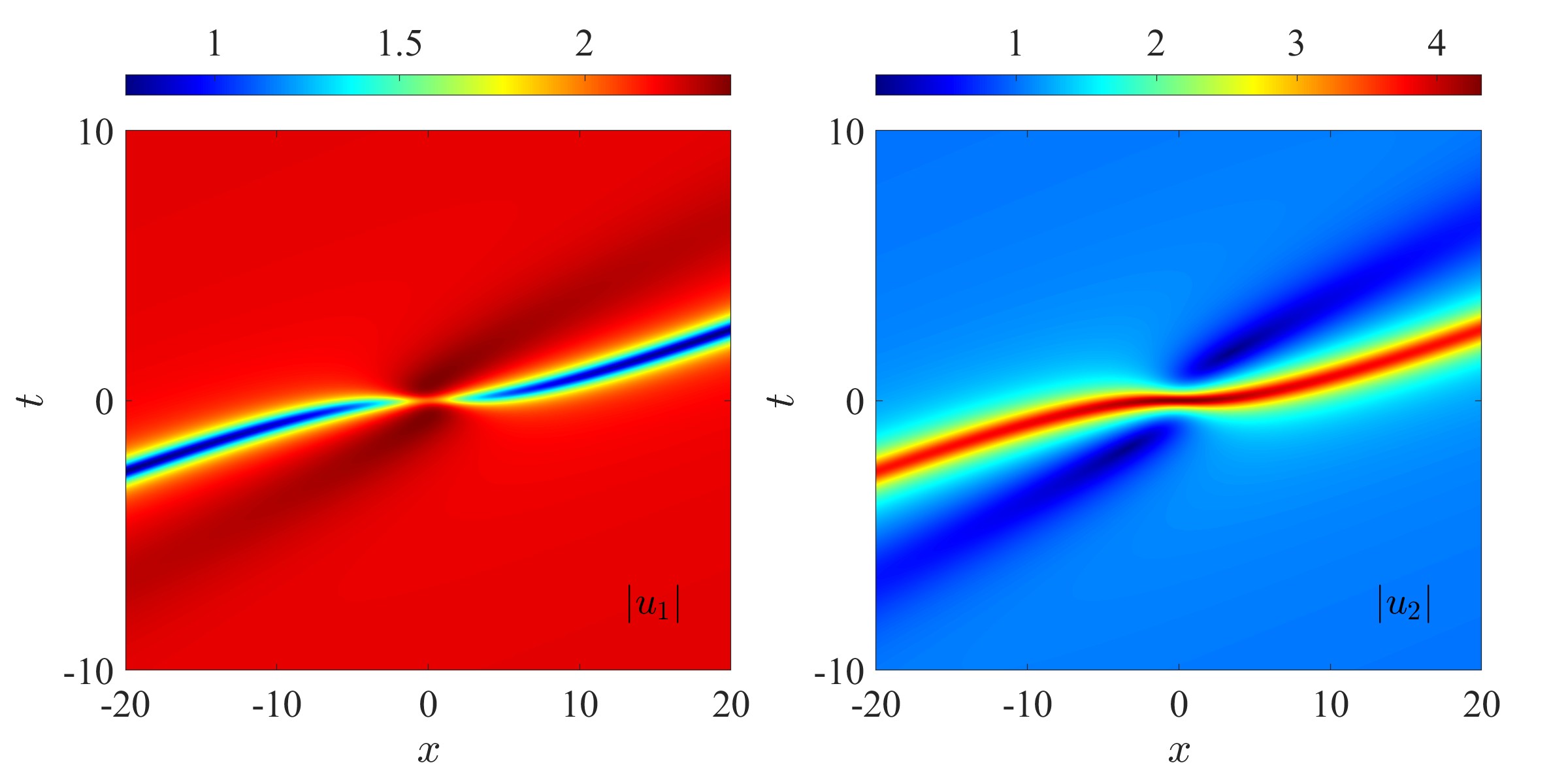}
			\caption{\(\chi_0\approx0.22483038\)}
			\label{fig:tr1-1}
		\end{subfigure}
		\hfill
		\begin{subfigure}{0.49\textwidth}
			\centering
			\includegraphics[width=\textwidth]{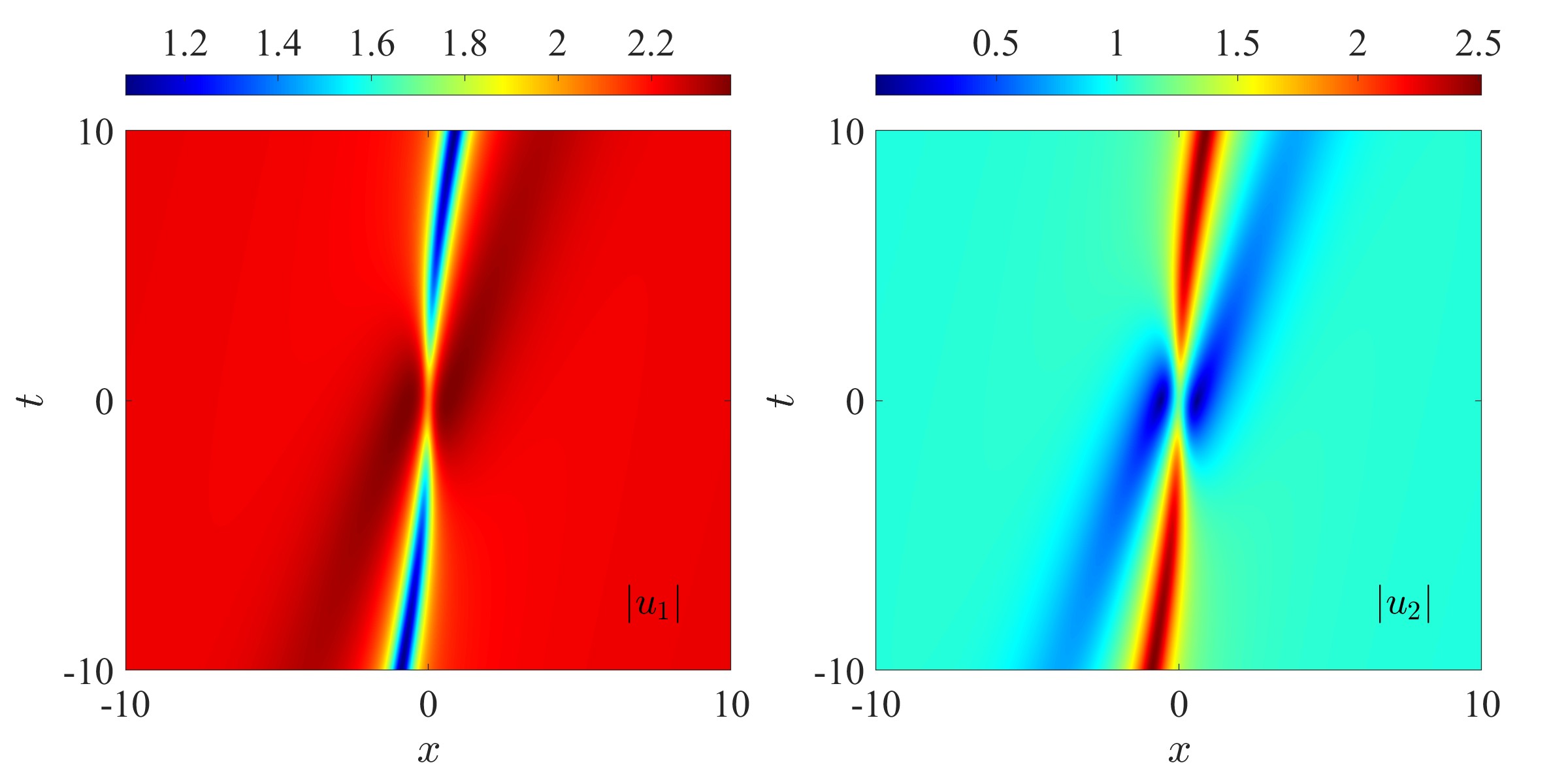}
			\caption{\(\chi_0\approx0.87973038\)}
			\label{fig:tr1-2}
		\end{subfigure}
		
		\vspace{0.3cm}
		
		\begin{subfigure}{0.49\textwidth}
			\centering
			\includegraphics[width=\textwidth]{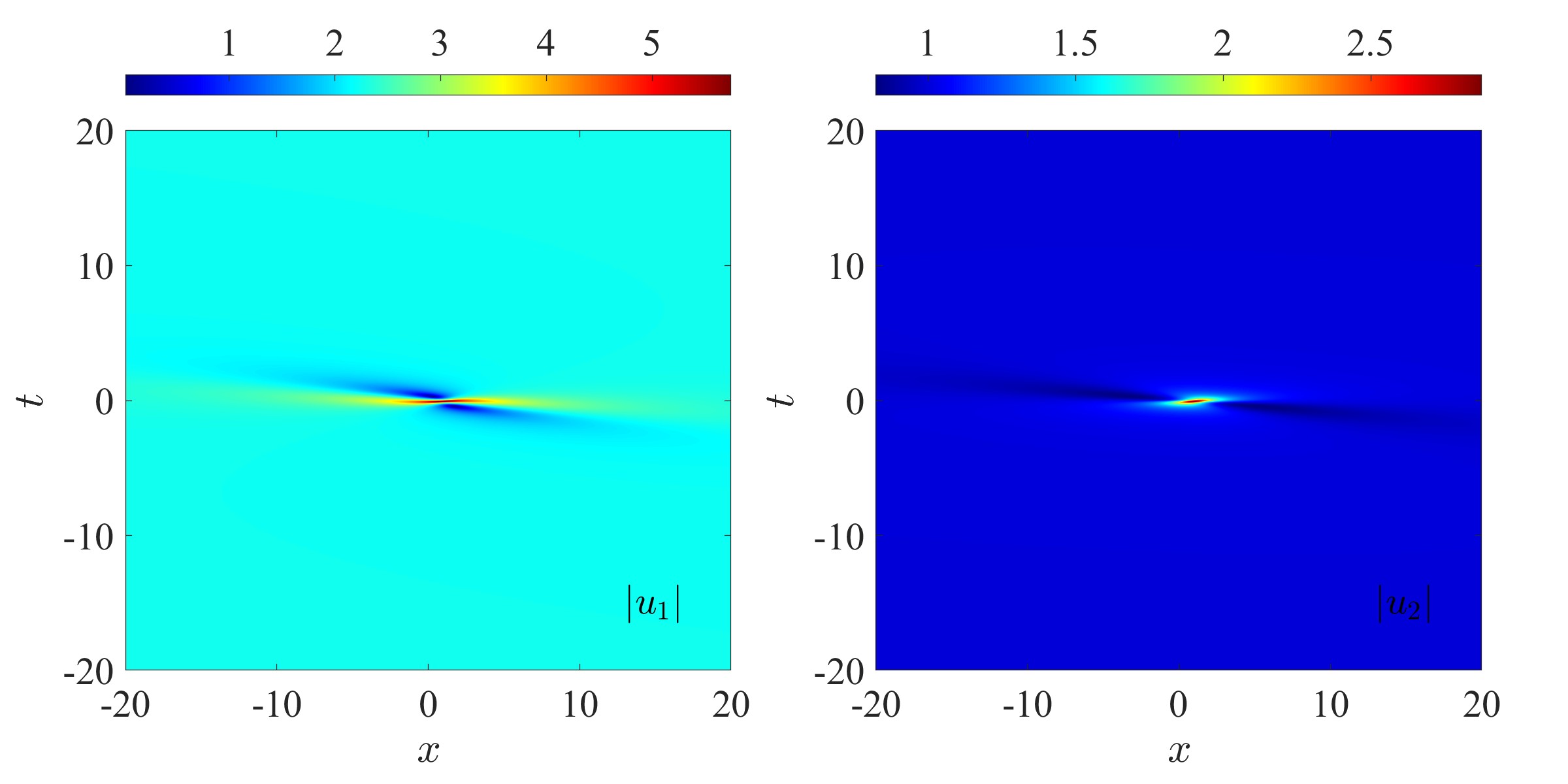}
			\caption{\(\chi_0\approx-0.07228038+0.14374134{\rm i}\)}
			\label{fig:tc1-1}
		\end{subfigure}
		\hfill
		\begin{subfigure}{0.49\textwidth}
			\centering
			\includegraphics[width=\textwidth]{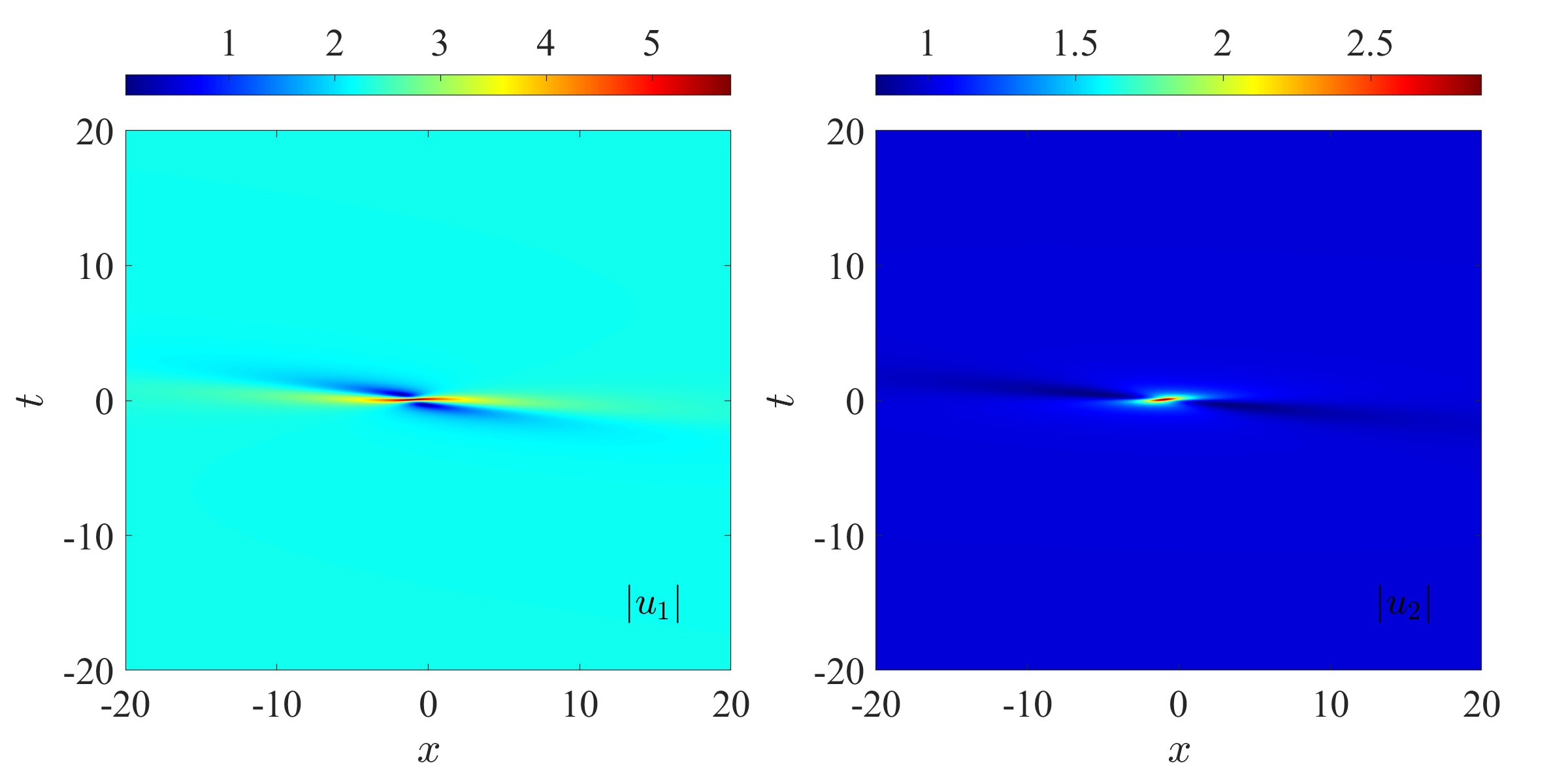}
			\caption{\(\chi_0\approx-0.07228038-0.14374134{\rm i}\)}
			\label{fig:tc1-2}
		\end{subfigure}
		
		\caption{
			Representative double-root solution profiles for
			\(A=5\), \(B=1\), and \(k=0.4\). Panels (a) and (b)
			correspond to the two real branch points and are obtained from
			the real-spectrum limiting construction. Panels (c) and (d)
			correspond to the nonreal conjugate pair and are obtained from
			the regular confluent Darboux construction. The associated
			squared spectral values are listed in
			Table~\ref{tab:A5-spectral-data}.
		}
		\label{fig:tr1-tc1}
	\end{figure}
	
	The real and nonreal sectors are therefore simultaneously present at
	this value of \(k\). Their distinction is determined by the spectral
	location of the branch point and not merely by the multiplicity of the
	spatial root.
	
	\paragraph{The four-real-root window: \(k=0.505\).}
	
	At \(k=0.505\), all four roots of \eqref{eq:P-A5} are real and
	effective. By Proposition~\ref{prop:real-double-positive}, their
	associated squared spectral values are automatically positive.
	Consequently, all four roots give admissible real-spectrum Darboux
	limits. Their distinct root locations and spectral values nevertheless
	produce different solution profiles, as displayed in
	Fig.~\ref{fig:tr2}.
	
	\begin{figure}[htbp]
		\centering
		\begin{subfigure}{0.49\textwidth}
			\centering
			\includegraphics[width=\textwidth]{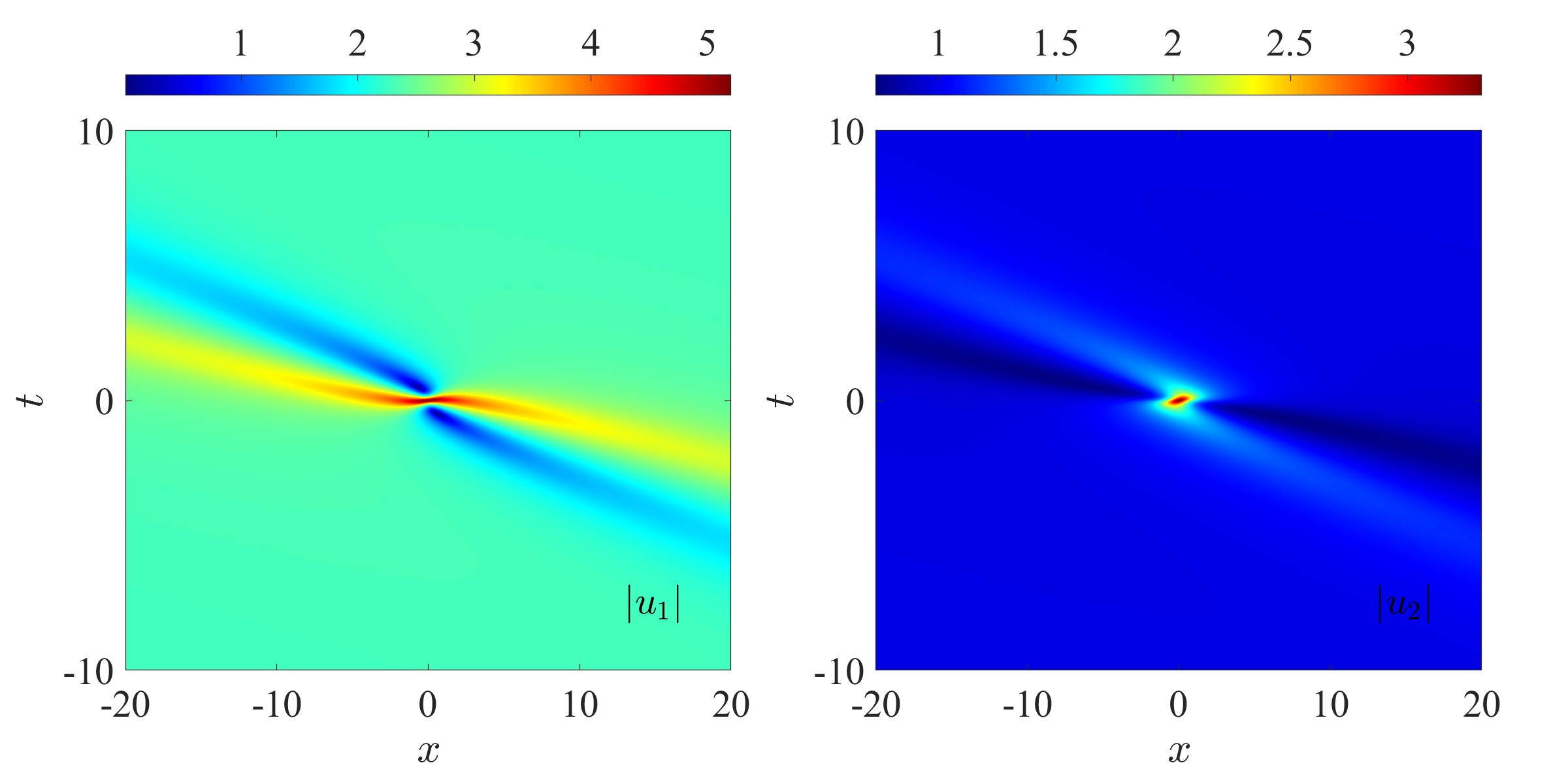}
			\caption{\(\chi_0\approx-0.04423862\)}
			\label{fig:tr2-1}
		\end{subfigure}
		\hfill
		\begin{subfigure}{0.49\textwidth}
			\centering
			\includegraphics[width=\textwidth]{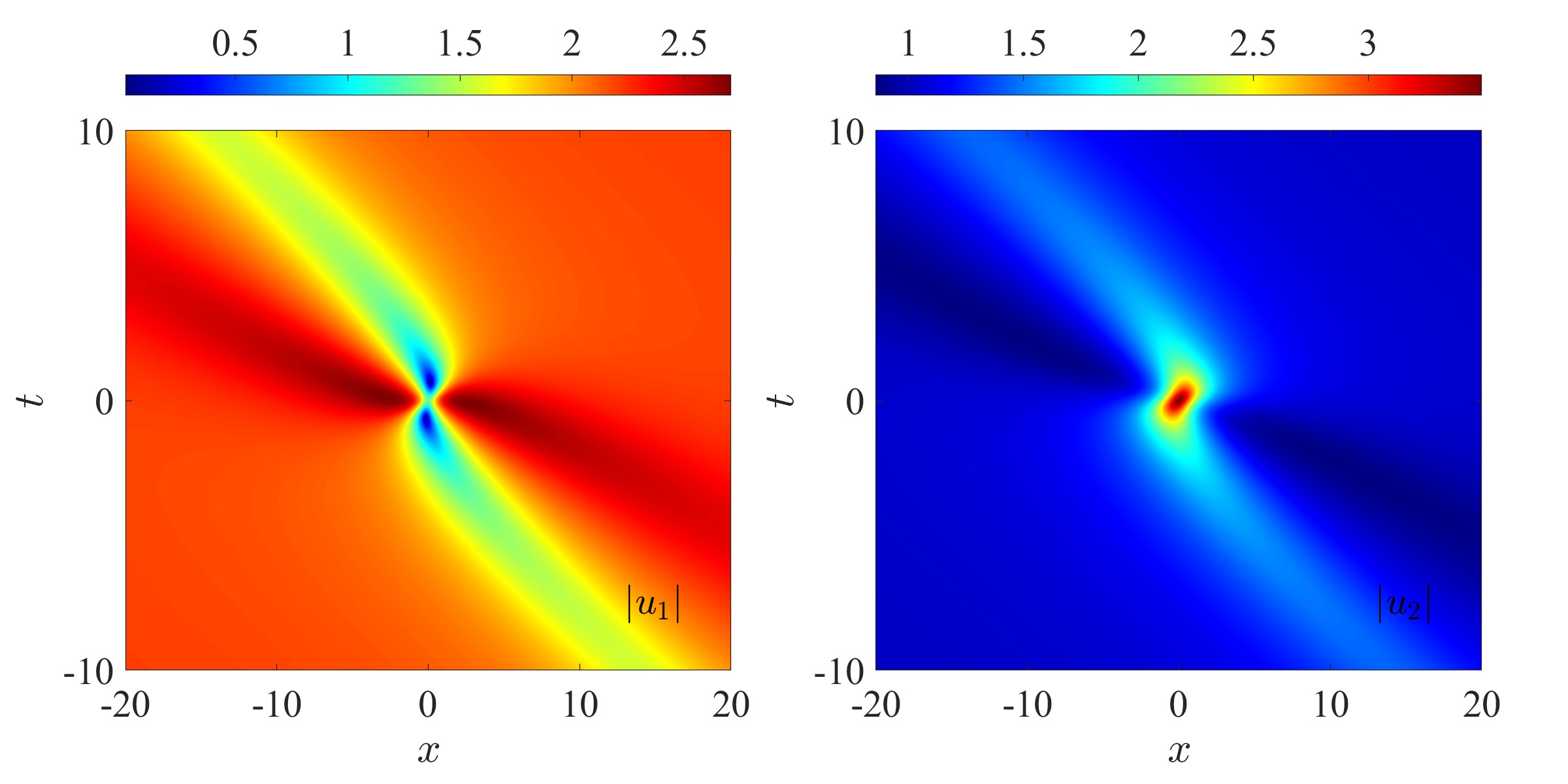}
			\caption{\(\chi_0\approx0.06140505\)}
			\label{fig:tr2-2}
		\end{subfigure}
		
		\vspace{0.3cm}
		
		\begin{subfigure}{0.49\textwidth}
			\centering
			\includegraphics[width=\textwidth]{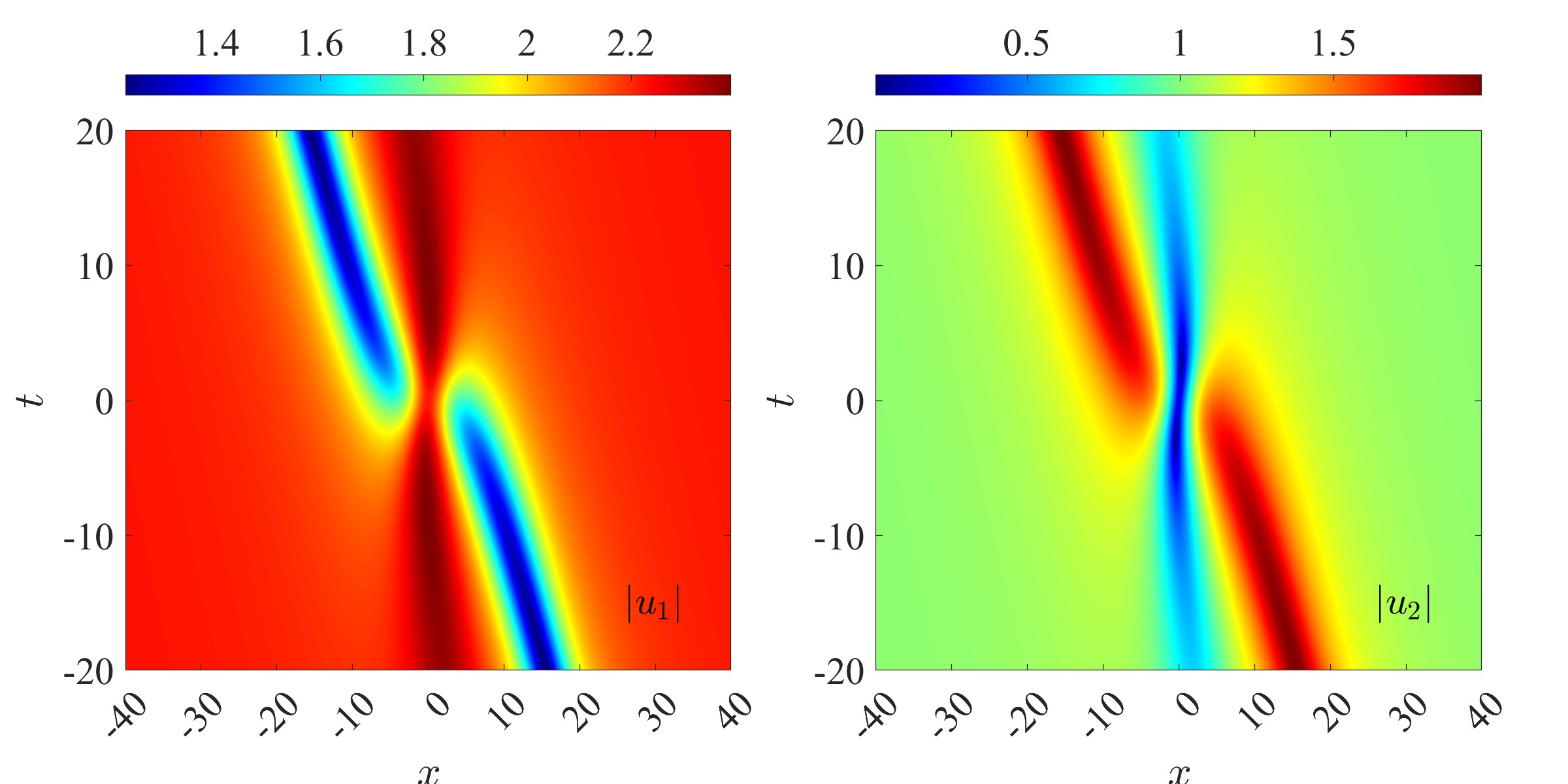}
			\caption{\(\chi_0\approx0.17955147\)}
			\label{fig:tr2-3}
		\end{subfigure}
		\hfill
		\begin{subfigure}{0.49\textwidth}
			\centering
			\includegraphics[width=\textwidth]{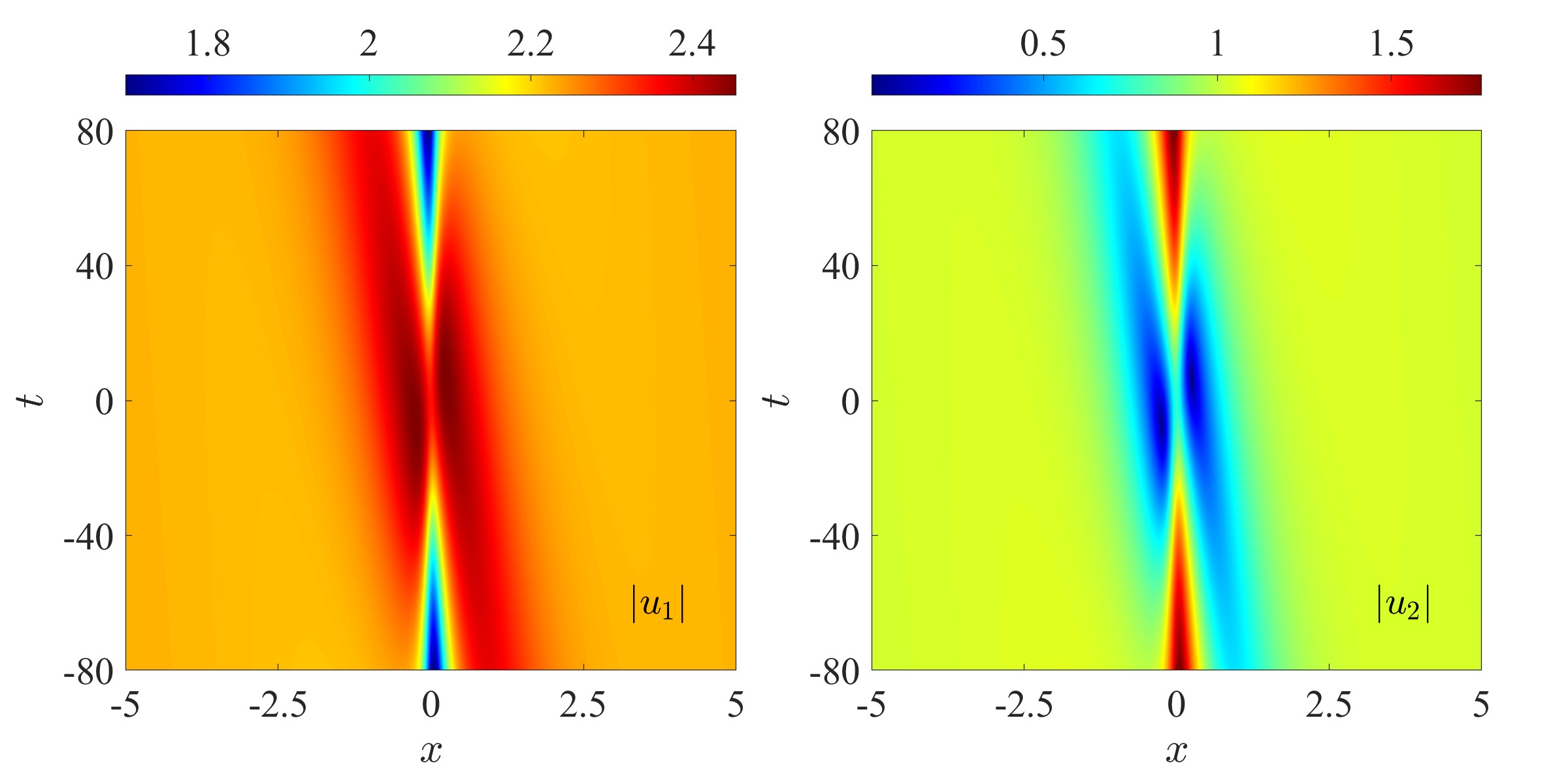}
			\caption{\(\chi_0\approx1.33343210\)}
			\label{fig:tr2-4}
		\end{subfigure}
		
		\caption{
			Representative real-spectrum double-root solutions for
			\(A=5\), \(B=1\), and \(k=0.505\), which lies inside the
			narrow four-real-root window. The panels are ordered by
			increasing spatial-root value. All four associated squared
			spectral values are positive and are listed in
			Table~\ref{tab:A5-spectral-data}.
		}
		\label{fig:tr2}
	\end{figure}
	
	This example shows that the number of admissible real-spectrum limits
	may change under a small variation of the carrier wavenumber. The
	transition from \(k=0.4\) to \(k=0.505\) converts the nonreal conjugate
	pair into two additional real branch points.
	
	\paragraph{The upper-\(k\) real--complex regime: \(k=0.6\).}
	
	At \(k=0.6\), the root configuration returns to two real roots and one
	nonreal conjugate pair. The coexistence of regular nonreal-spectrum
	constructions and real-spectrum limiting solutions is therefore
	restored after passage through the triple-root transition.
	Figure~\ref{fig:tr3-tc3} shows that the resulting profiles differ from
	those in the lower-\(k\) real--complex regime even though the root
	count is the same.
	
	\begin{figure}[htbp]
		\centering
		\begin{subfigure}{0.49\textwidth}
			\centering
			\includegraphics[width=\textwidth]{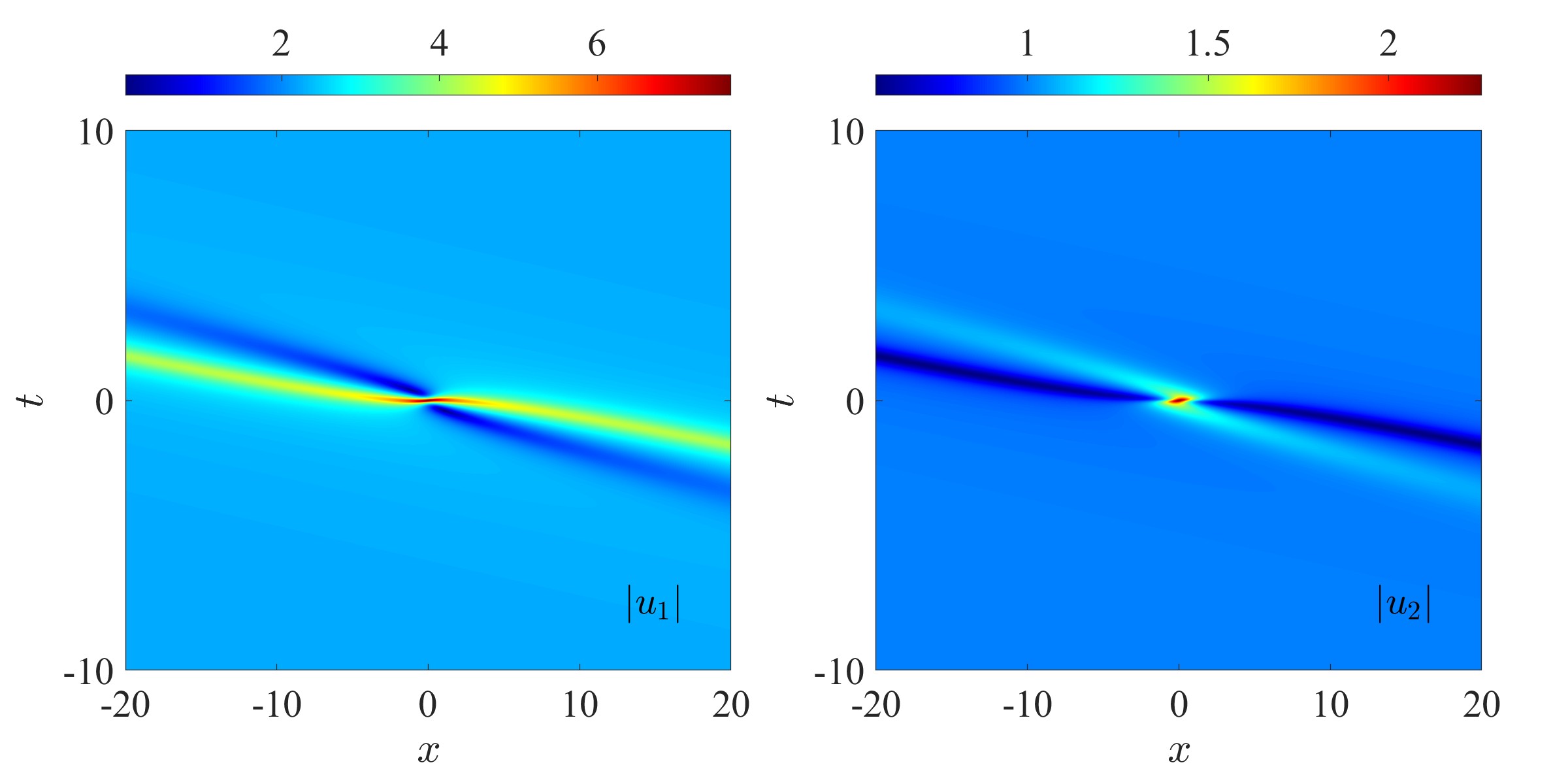}
			\caption{\(\chi_0\approx-0.16446338\)}
			\label{fig:tr3-1}
		\end{subfigure}
		\hfill
		\begin{subfigure}{0.49\textwidth}
			\centering
			\includegraphics[width=\textwidth]{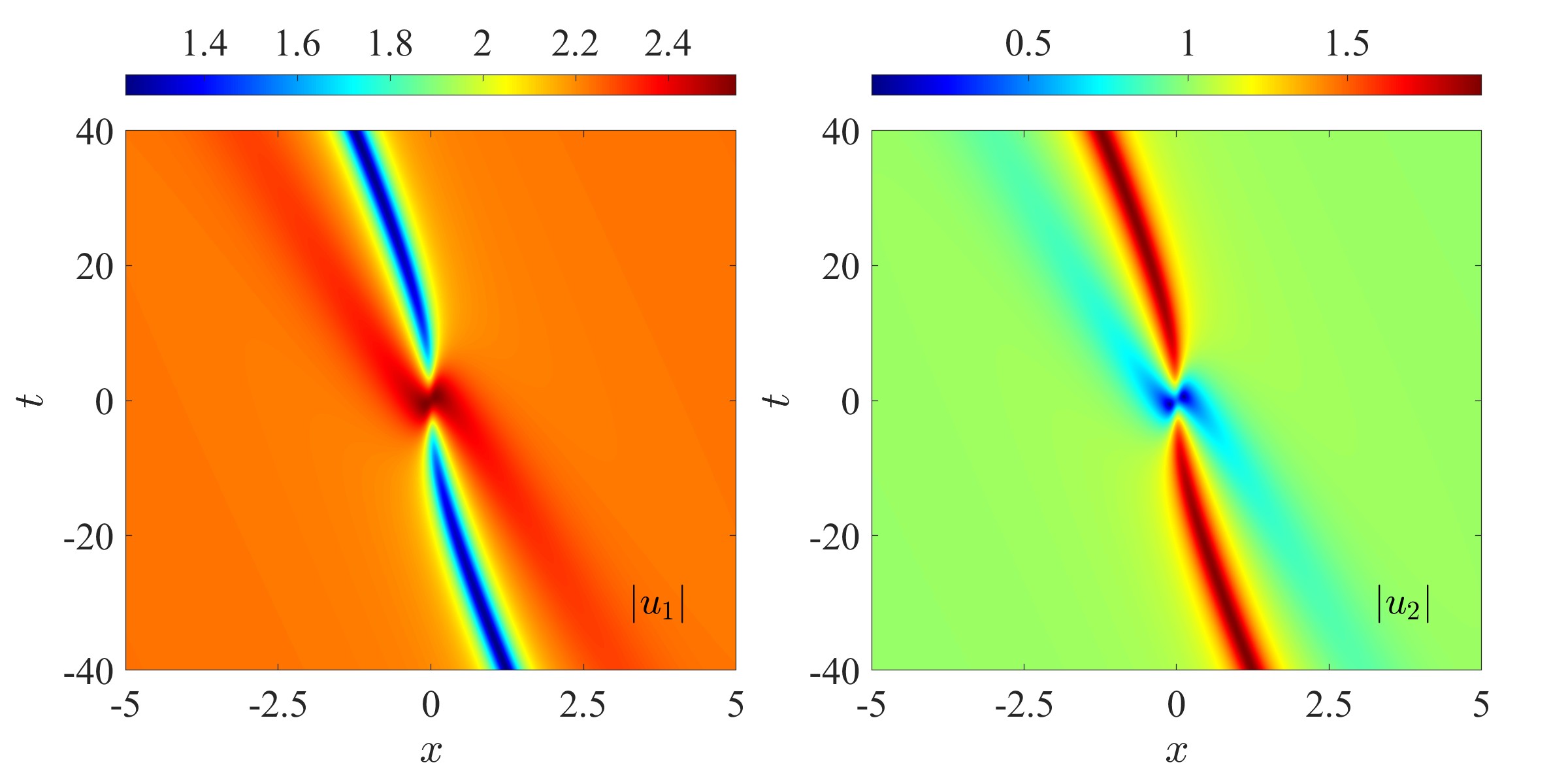}
			\caption{\(\chi_0\approx1.85327014\)}
			\label{fig:tr3-2}
		\end{subfigure}
		
		\vspace{0.3cm}
		
		\begin{subfigure}{0.49\textwidth}
			\centering
			\includegraphics[width=\textwidth]{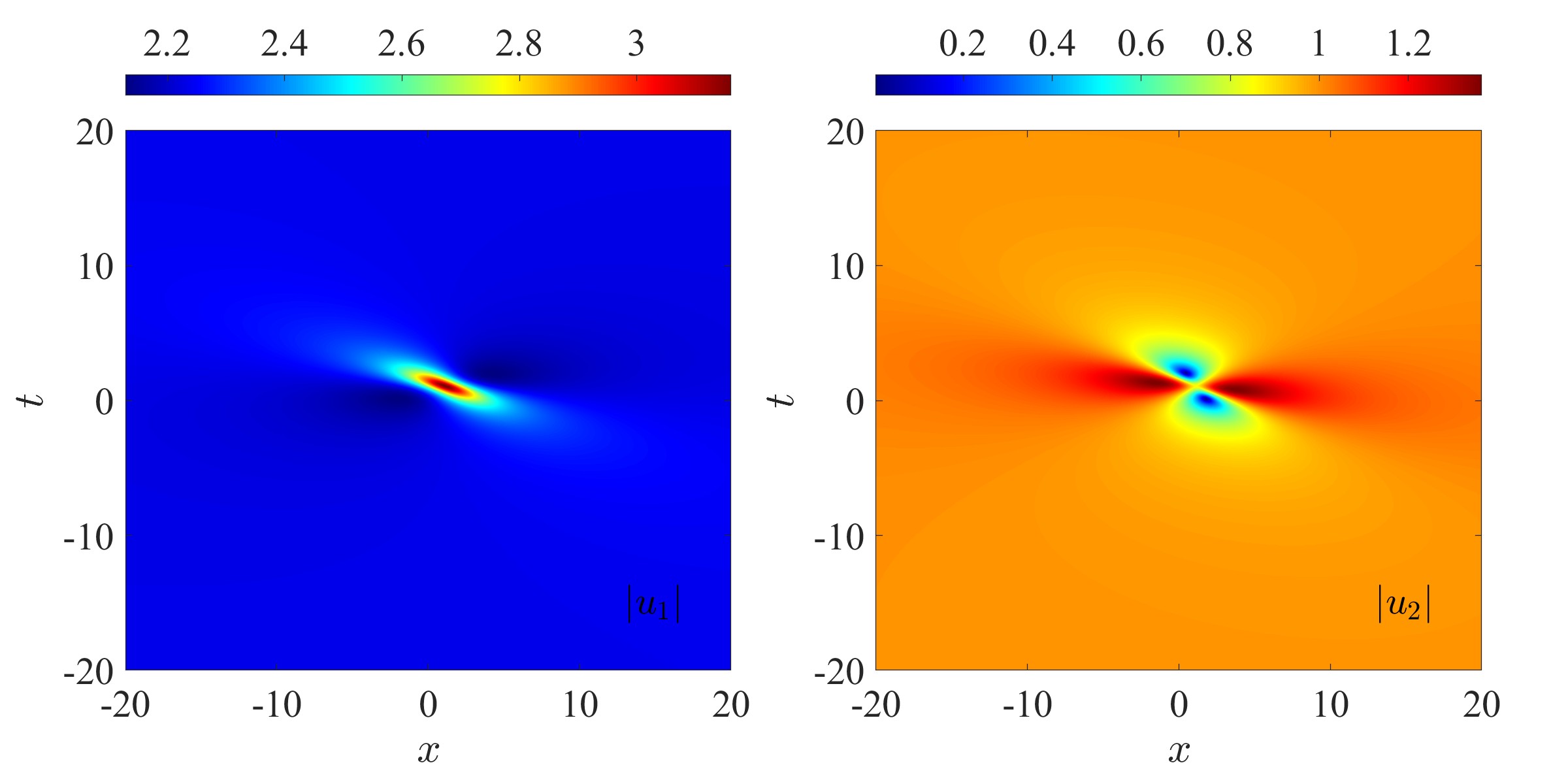}
			\caption{\(\chi_0\approx0.23559662+0.17185743{\rm i}\)}
			\label{fig:tc3-1}
		\end{subfigure}
		\hfill
		\begin{subfigure}{0.49\textwidth}
			\centering
			\includegraphics[width=\textwidth]{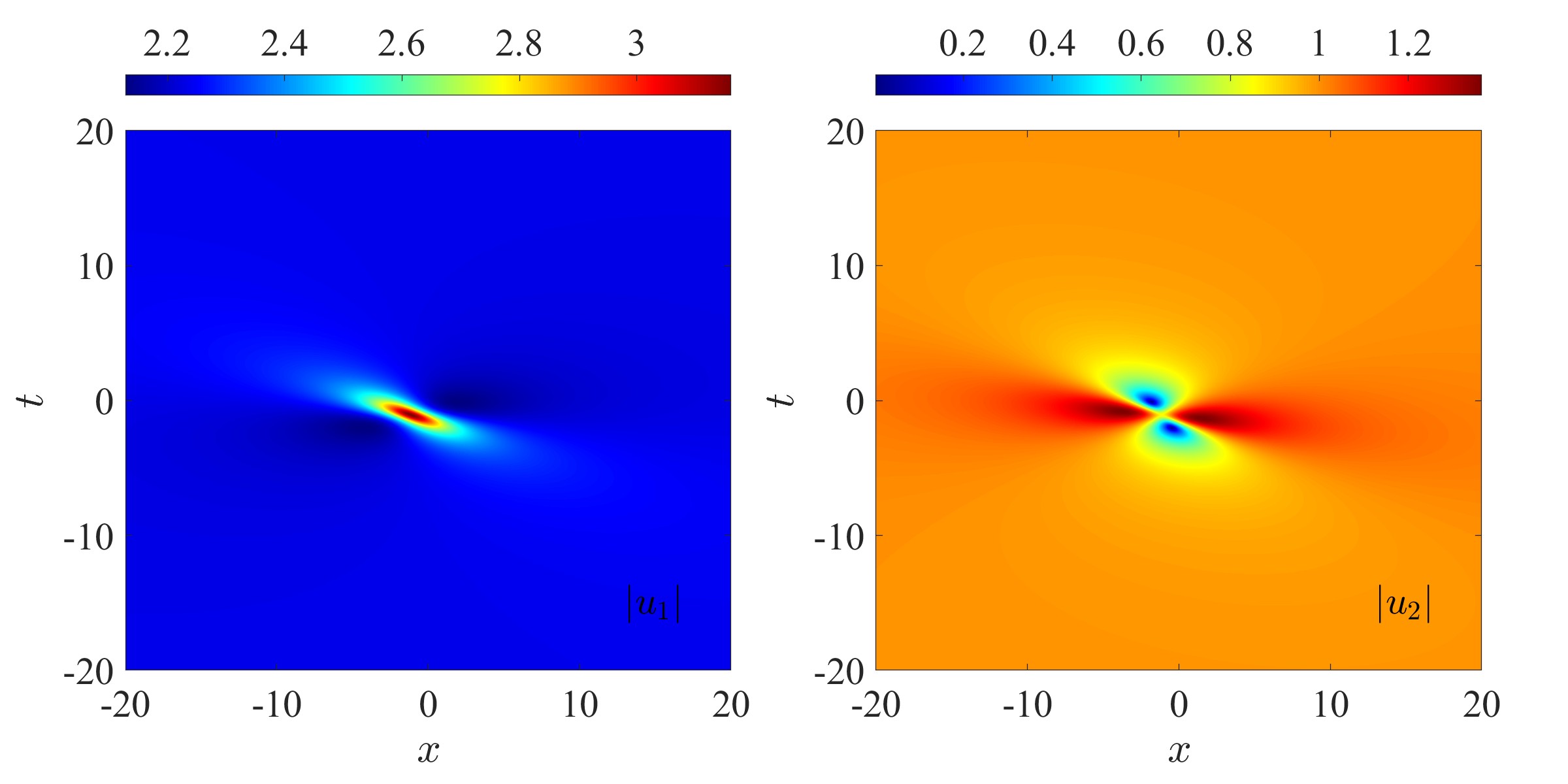}
			\caption{\(\chi_0\approx0.23559662-0.17185743{\rm i}\)}
			\label{fig:tc3-2}
		\end{subfigure}
		
		\caption{
			Representative double-root solution profiles for
			\(A=5\), \(B=1\), and \(k=0.6\). Panels (a) and (b)
			correspond to the two real branch points, whereas panels
			(c) and (d) correspond to the nonreal conjugate pair.
			Although the root count agrees with that at \(k=0.4\), the
			root locations and squared spectral values differ
			substantially after passage through the four-real-root window
			and the triple-root transition.
		}
		\label{fig:tr3-tc3}
	\end{figure}
	
	The three representative values therefore illustrate the sequence
	\[
	\text{real--complex configuration}
	\;\longrightarrow\;
	\text{four-real-root configuration}
	\;\longrightarrow\;
	\text{real--complex configuration}.
	\]
	The spatial-root configuration determines whether the regular
	one-fold Darboux formula or the real-spectrum limiting construction is
	required. Within each spectral sector, the admissible confluent
	coefficients further determine the detailed solution profile.
	
	\subsection{Representative triple-root solutions}
	\label{subsec:A5-triple-solutions}
	
	We next consider the triple spatial root at
	\[
	A=5,\qquad
	B=1,\qquad
	k=k_{\rm tri}.
	\]
	The corresponding root and squared spectral value are given in
	\eqref{eq:A5-triple-data}. After the normalization \(M_0=1\), the
	leading confluent vector has the form
	\begin{equation}
		\label{eq:A5-triple-leading-vector}
		{\bf y}_R
		=
		{\bf Y}_0+M_1{\bf Y}_1+\frac{M_2}{2}{\bf Y}_2.
	\end{equation}
	Apart from the canonical choice \(M_1=M_2=0\), every admissible
	noncanonical triple-root vector has \(M_2\neq0\). This follows because
	the sector \(M_2=0\), \(M_1\neq0\), does not satisfy the triple-root
	null constraint. We compare two representative coefficient sectors.
	
	\paragraph{Pure second-derivative sector.}
	
	Here
	\[
	M_1=0,\qquad M_2\neq0,
	\]
	with \(M_2\) lying on the admissible circle
	\eqref{eq:M2-circle-triple}. This choice isolates the contribution of
	the second-derivative mode \({\bf Y}_2\).
	
	\paragraph{Full triple-chain sector.}
	
	Here
	\[
	M_1M_2\neq0,
	\]
	and the coefficients satisfy the full null constraint
	\eqref{eq:triple-null-general}. All three confluent modes
	\({\bf Y}_0\), \({\bf Y}_1\), and \({\bf Y}_2\) then contribute to
	the leading vector. The two sectors are compared in
	Fig.~\ref{fig:triple-root-sectors}.
	
	\begin{figure}[htbp]
		\centering
		\begin{subfigure}{0.49\textwidth}
			\centering
			\includegraphics[width=\textwidth]{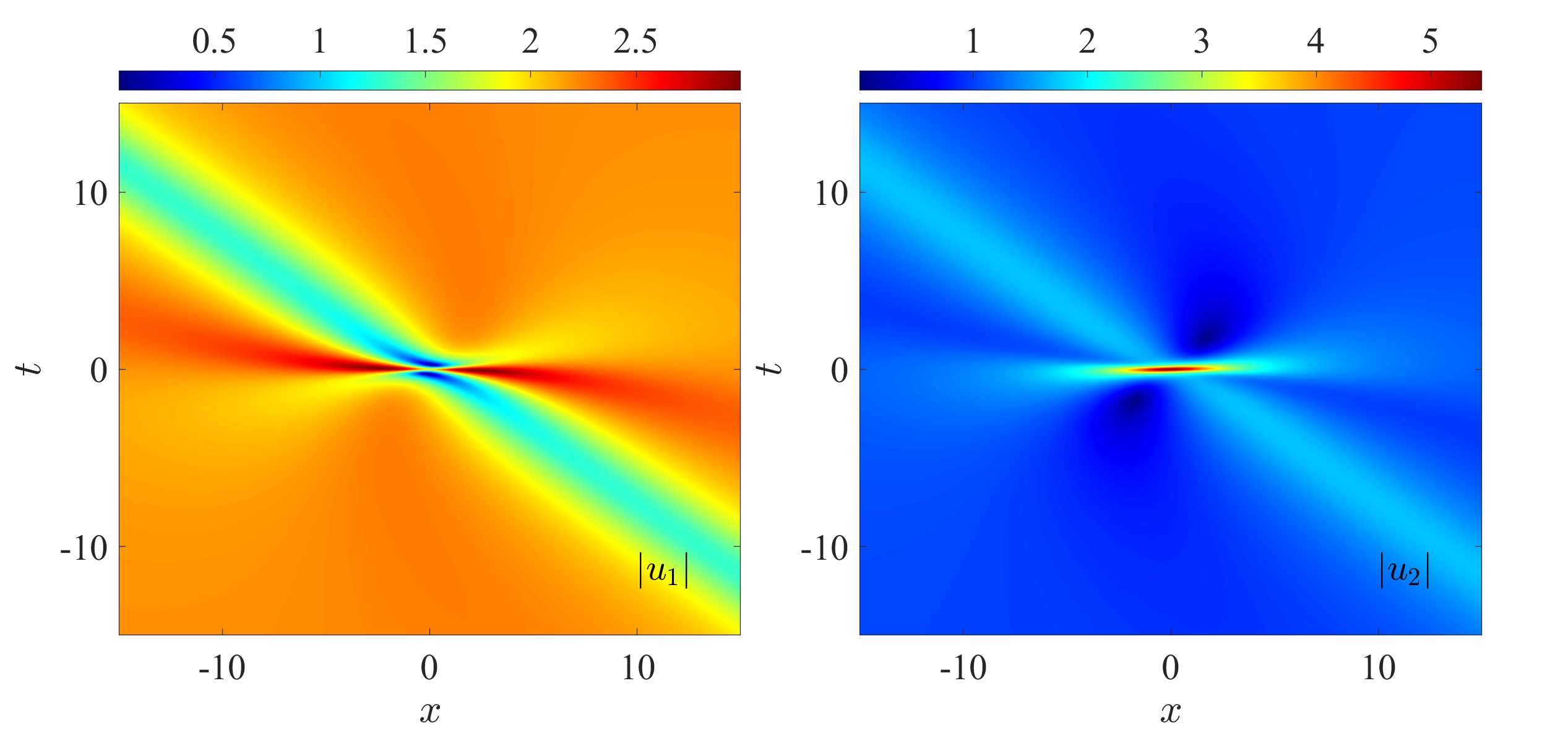}
			\caption{Pure second-derivative sector:
				\(M_1=0\), \(M_2\neq0\)}
			\label{fig:triple-pure-M2}
		\end{subfigure}
		\hfill
		\begin{subfigure}{0.49\textwidth}
			\centering
			\includegraphics[width=\textwidth]{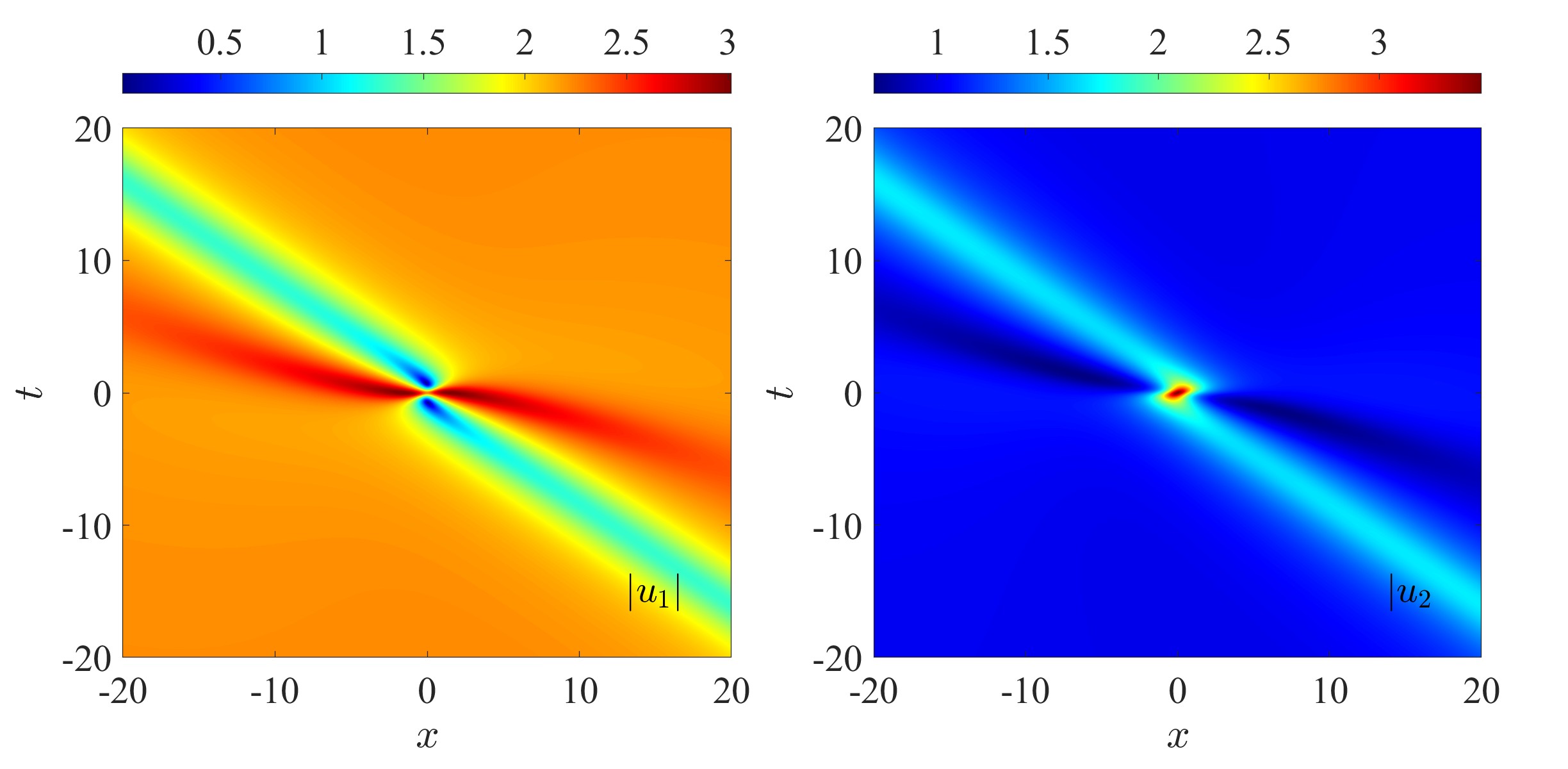}
			\caption{Full triple-chain sector:
				\(M_1M_2\neq0\)}
			\label{fig:triple-full-chain}
		\end{subfigure}
		
		\caption{
			Representative triple-root solutions for
			\(A=5\), \(B=1\), and \(k=k_{\rm tri}\).
			Panel (a) corresponds to the pure second-derivative sector
			\(M_1=0\), \(M_2\neq0\), whereas panel (b) corresponds to
			the full triple-chain sector \(M_1M_2\neq0\). The two
			examples show that the admissible coefficient sector has a
			substantial effect on the resulting space--time profile.
		}
		\label{fig:triple-root-sectors}
	\end{figure}
	
	The comparison in Fig.~\ref{fig:triple-root-sectors} shows that the
	first- and second-derivative contributions control different aspects
	of the triple-root solution. In the pure second-derivative sector,
	only \({\bf Y}_0\) and \({\bf Y}_2\) enter the leading vector, whereas
	the full triple-chain sector also contains the independent
	first-derivative contribution \(M_1{\bf Y}_1\). The latter may
	substantially modify the finite space--time profile without changing
	the underlying triple-root spectral data.
	
	\subsection{Comparison with the nonlocalized asymptotic predictions}
	\label{subsec:A5-asymptotic-comparison}
	\label{sec:nonlocal-examples}
	
	We conclude this section by comparing the asymptotic curves obtained in
	Section~\ref{sec:nonlocal-asymptotics} with representative exact
	solutions on the slice \(A=5\), \(B=1\). The relevant theoretical
	curves are superimposed on the modulus plots in
	Fig.~\ref{fig:A5-asymptotic-comparison}. The purpose of this comparison
	is to examine whether the persistent structures observed in the exact
	solutions follow the unbounded sets selected by the preceding
	asymptotic analysis.
	
	\begin{figure}[htbp]
		\centering
		\begin{subfigure}{0.49\textwidth}
			\centering
			\includegraphics[width=\textwidth]
			{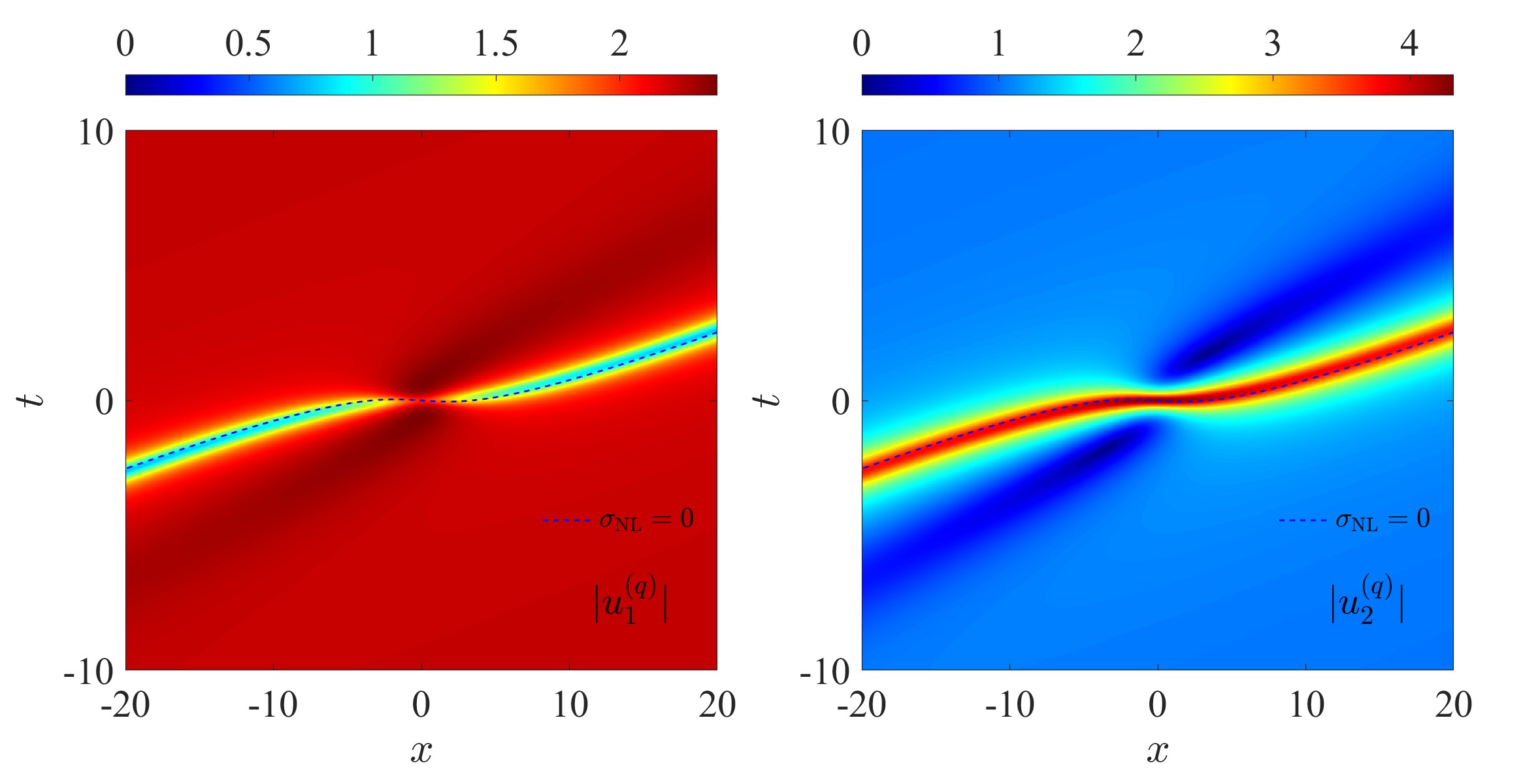}
			\caption{Pure double-root solution and a predicted kernel
				level-set graph.}
			\label{fig:A5-double-asymptotic-curve}
		\end{subfigure}
		\hfill
		\begin{subfigure}{0.49\textwidth}
			\centering
			\includegraphics[width=\textwidth]
			{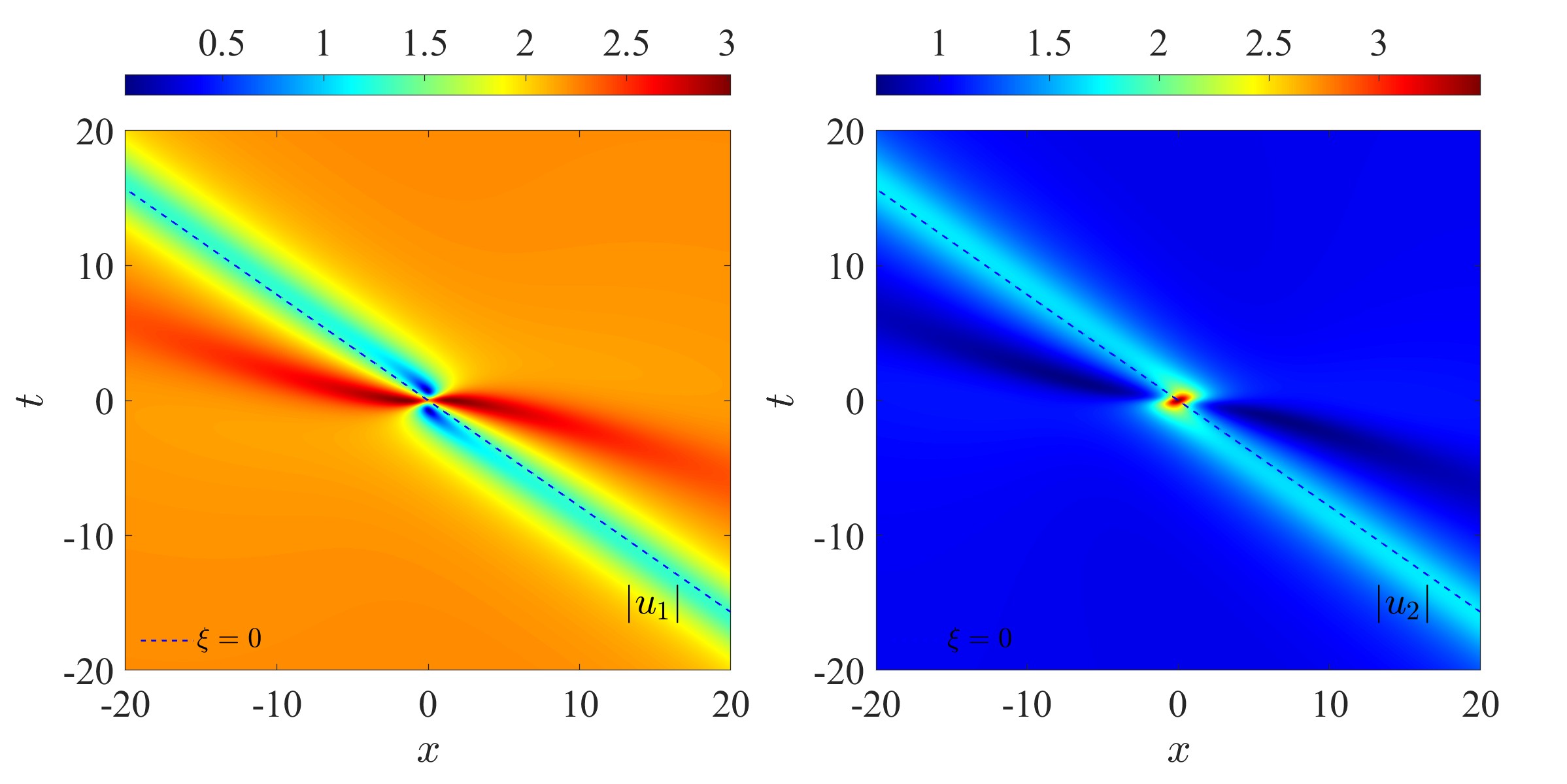}
			\caption{Triple-root solution and the predicted
				characteristic line.}
			\label{fig:A5-triple-asymptotic-line}
		\end{subfigure}
		
		\caption{
			Comparison between representative exact-solution profiles and
			the nonlocalized asymptotic curves derived in
			Section~\ref{sec:nonlocal-asymptotics}. In panel (a), the
			superimposed curve is a kernel level-set graph of the pure
			double-root limiting denominator determined by
			\eqref{eq:C-rho-implicit}. In panel (b), the superimposed
			straight line is the distinguished triple-root characteristic
			\(x+t/(3k_{\rm tri}^{2})=0\). The persistent structures in
			the exact solutions follow the respective theoretical curves
			over the displayed space--time ranges.
		}
		\label{fig:A5-asymptotic-comparison}
	\end{figure}
	
	For the pure double-root example in
	Fig.~\ref{fig:A5-asymptotic-comparison}(a), the visible persistent
	structure follows the kernel level-set graph defined by
	\eqref{eq:C-rho-implicit}. Its bending and displacement are reproduced
	by the lower-order terms of the selected level-set equation, whereas
	its large-distance behavior is governed by the nonzero cubic
	coefficient. The plotted curve therefore agrees with the unbounded
	component selected in the double-root asymptotic analysis.
	
	For the triple-root example in
	Fig.~\ref{fig:A5-asymptotic-comparison}(b), the persistent structure is
	aligned with
	\[
	x+\frac{t}{3k_{\rm tri}^{2}}=0,
	\]
	as predicted by
	Proposition~\ref{prop:triple-nonlocal-corridor}. Along this
	characteristic, the Darboux numerator and denominator have the same
	leading quadratic order, and the background-normalized field
	approaches the nonbackground complex limit in
	\eqref{eq:triple-line-normalized-limit}.
	
	The agreement displayed in
	Fig.~\ref{fig:A5-asymptotic-comparison} is consistent with the two
	distinct asymptotic mechanisms identified above. The pure double-root
	structure follows a nonlinear kernel level-set graph, whereas the
	triple-root structure follows a straight characteristic. These
	representative examples therefore support the asymptotic conclusions
	obtained in Section~\ref{sec:nonlocal-asymptotics}.
	
	\section{Conclusions and Discussions}
	\label{sec:conclusion}
	
	We have developed a systematic Darboux framework for degenerations of the spatial spectral branches of the two-component coupled Fokas--Lenells system on a counterpropagating plane-wave background. The degeneration studied here is distinct from the conventional coalescence of Darboux poles. It is defined by a double or triple root of the single fixed-endpoint equation \(F(\chi;\mu_0)=0\). Consequently, the relevant confluent structure is determined by the local structure of the spatial branches at \(\mu_0\), rather than by the multiplicity of a pole in the Darboux spectral parameter. For the singular real-spectrum construction, the endpoint is subsequently displaced as \(\mu=\mu_0+\eta\) to provide a regular approximating family; this displacement is not a collision of several distinct Darboux poles.
	
	An important preliminary point is the distinction between the polynomial spatial equation and its reduced rational parameterization. Under the critical condition \(2-k(A-B)=0\), the root \(\chi=0\) is a genuine spatial eigenvalue branch that persists for every value of \(\mu\). At \(\mu=A+B\), a movable branch intersects this persistent zero branch and produces a multiple zero of the fixed-\(\mu\) polynomial. This event is a branch intersection, rather than an effective ramification point of the reduced map \(\chi\mapsto\Lambda(\chi)\). Nevertheless, the zero branch must be retained whenever it appears as the remaining simple mode in a Darboux limit.
	
	For double branch points, we derived the governing quartic equation and proved that it has at least two real roots in the physical parameter regime \(A>0, \qquad B>0, \qquad k\neq0\). Its generic root distribution is therefore either four real roots or two real roots together with one nonreal complex-conjugate pair. In the equal-amplitude case, exactly two real roots and one nonreal conjugate pair occur. We further proved that an effective double branch point is real if and only if its associated squared spectral value is real.
	
	This classification separates two different Darboux mechanisms. A nonreal double branch point corresponds to a nonreal value of \(\mu\), so its square root is neither real nor purely imaginary. The ordinary one-fold Darboux formula is then regular, and arbitrary linear combinations of the double-root derivative chain and the remaining simple mode generate globally nonsingular confluent solutions. Because the simple mode generally carries an independent exponential phase, these solutions are not generically rational after normalization by the plane-wave background.
	
	A real double branch point, by contrast, has a positive spectral value \(\mu_0>0, \qquad \lambda_0=\sqrt{\mu_0}\in\mathbb R\), and requires a singular limiting construction. The central regularization mechanism is the cancellation of the vanishing factor \(\lambda^*-\lambda\) by a \(J\)-null leading vector. We established the conservation and inertia properties of the \(J\)-form and used them to determine the complete projective null cone in the mixed double--simple eigenspace. After normalization, the general leading vector is \({\bf y}_R = {\bf Y}_0 + M_1{\bf Y}_1 + c_s{\bf Y}_s\), where \(M_1\) and \(c_s\) satisfy an explicit quadratic constraint. The remaining simple root obeys
	\[
	d_s
	=
	\frac{\mu_0(\chi_0-\chi_s)^2}{\chi_s^2-k^2}<0,
	\qquad
	|\chi_s|<|k|,
	\]
	including the persistent-zero-branch case.
	
	We isolated the common real-spectrum mechanism in Theorem~\ref{thm:unified-real-spectrum-directional-limit}. Once an expansion \({\bf y}(\eta)={\bf y}_R+\eta{\bf y}_\mu+O_{C^2_{\rm loc}}(|\eta|^2)\) and the nullity of \({\bf y}_R\) are known, the theorem gives both the directional denominator and its strictly negative real part. The double- and triple-root sections therefore differ only in the construction of their leading and correction vectors, not in the final kernel argument.
	
	The full mixed vector was lifted from a nearby nonreal spectral point by combining the two splitting branches near the double root with the analytic continuation of the remaining simple branch. The symmetric root moments determine the \({\bf Y}_2\)- and \({\bf Y}_3\)-contributions, while the motion of the simple branch contributes the additional correction \({\bf Y}_{s,\mu}\). This yields a general directional real-\(\lambda\) Darboux limit for the complete mixed sector. The real part of the limiting denominator is strictly negative, which ensures global nonsingularity. The pure double-root circle, the elementary vector, the canonical nonzero real representative, and the saturated mixed boundary are obtained directly by specializing the general formula.
	
	For the complete nonzero pure double-root sector, we derived explicit background-normalized rational formulas. The admissible coefficient is
	\[
	r=M_1\in\mathbb C\setminus\{0\}, \qquad b_J|r|^2+2a_J\operatorname{Re}r=0.
	\]
	The real part of the corresponding limiting denominator is
	\[
	-\lambda_0 \left| 1+{\rm i}r\xi \right|^2.
	\]
	A nonzero admissible \(r\) cannot be purely imaginary, and therefore
	\[
	\left| 1+{\rm i}r\xi \right|^2>0
	\]
	for every real \(\xi\). This proves global regularity for every nonzero complex point of the pure-sector null circle, rather than only for its canonical real representative.
	
	For a general complex \(r\), the imaginary part of the denominator is a real polynomial of the form \(\gamma_3\xi^3 + \gamma_2\xi^2 + \gamma_1\xi + \gamma_t t + \gamma_0\). The identity \(a_Jp_1=\chi_0\) gives the simplified leading coefficients
	\[
	\gamma_3=-\frac{\chi_0}{3}|r|^2,
	\qquad
	\gamma_2=\chi_0\operatorname{Im}r.
	\]
	The coefficients \(\gamma_2\) and \(\gamma_0\) vanish for the canonical real point
	\[
	r_{\rm c}=-\frac{2a_J}{b_J},
	\]
	but are generally nonzero for a genuinely complex null coefficient. They change the center and lower-order asymmetry of the rational structure without changing its leading cubic behavior.
	
	For a fixed admissible \(r\), changing the approach angle is equivalent to a rigid space--time translation preserving \(\xi=x+\nu'(\chi_0)t\) whenever \(\gamma_t(r)\neq0\). If \(\gamma_t(r)=0\), such a translation law does not follow unless the remaining constant directional contribution also vanishes.
	
	We further showed that, provided \(\gamma_t(r_{\rm c})\neq0\), every nonzero complex point of the pure-sector null circle is related to the canonical real point by a rigid space--time translation and a constant phase factor. Consequently, the complex null-circle parameter does not generate a new intrinsic modulus profile. These conclusions concern only the pure double-root sector \(c_s=0\); the mixed double--simple sector contains an independent simple-branch phase and is governed by the more general limiting formula.
	
	Triple branch points were characterized by solving the simultaneous multiple-root conditions for the spatial equation. Their reality was derived directly from the coefficient identities: first \(\chi_0^2\in\mathbb R\), then the purely imaginary alternative was excluded using \(S=A+B>0\), and only afterwards was \(\mu_0=2/(3\chi_0)\in\mathbb R\) concluded. Thus no real-coefficient or conjugate-root argument was used circularly. In the physical regime, the spectral value is positive. The full projective leading vector is
	\[
	{\bf y}_R = {\bf Y}_0 + M_1{\bf Y}_1 + \frac{M_2}{2}{\bf Y}_2.
	\]
	We proved the required Gram relations and the strict negative definiteness of the lower derivative block, thereby obtaining the complete null constraint on \(M_1\) and \(M_2\).
	
	For the general full-chain sector \(M_1M_2\neq0\), the three splitting roots were combined by prescribing their first three weighted moments. The cubic recurrence of these moments produces the higher derivative contributions \({\bf Y}_3,\qquad {\bf Y}_4,\qquad {\bf Y}_5\) in the first correction vector. This construction gives the general directional real-spectrum triple-root Darboux limit. The pure second-derivative sector and the canonical sector \({\bf y}_R={\bf Y}_0\) follow as direct corollaries, whereas a nonzero first-derivative contribution without a second-derivative contribution is excluded by the null condition.
	
	The resulting real-spectrum solutions may fail to approach the plane-wave background uniformly in all joint space--time limits. For the complete nonzero pure double-root null-circle family, we introduced level sets of the imaginary part of the limiting denominator. When \(\gamma_t(r)\neq0\), each level set is a connected global graph over \(\xi\). Since \(\gamma_3(r)\neq0\) for every admissible \(r\neq0\), these graphs have a common leading cubic geometry.
	
	Along every such graph, \(|\xi|\longrightarrow\infty\), and the two background-normalized components satisfy
	\[
	\frac{ u_j^{(R,2;\alpha,r)} }{ u_{j,0} } \longrightarrow 1-2h_j(\chi_0), \qquad j=1,2.
	\]
	More explicitly,
	\[
	\mathcal P_1^{(D)}
	=
	\frac{\chi_0-k}{\chi_0+k},
	\qquad
	\mathcal P_2^{(D)}
	=
	\frac{\chi_0+k}{\chi_0-k},
	\qquad
	\mathcal P_1^{(D)}\mathcal P_2^{(D)}=1.
	\]
	Thus the normalized intensity plateaus also have product one. The limiting factors are independent of \(r\), the approach angle \(\alpha\), and the kernel level; these parameters affect only the location and lower-order structure of the curves.
	
	If \(\gamma_t(r)=0\), the level-set equation reduces to a cubic equation for finitely many constant values of \(\xi\), and no component satisfies \(|\xi|\longrightarrow\infty\). The nonbackground plateau theorem therefore applies precisely to the generic graph case \(\gamma_t(r)\neq0\), rather than constituting an unconditional classification of all far-field trajectories.
	
	For triple-root solutions with \(M_2\neq0\), we analyzed the distinguished characteristic line
	\[
	x+\frac{t}{3k^2}=0.
	\]
	Along this line, the leading vector grows linearly in \(t\), whereas both the Darboux numerator and denominator grow quadratically. Their ratio therefore approaches a finite nonzero limit, and both background-normalized components converge to explicit nonbackground complex plateaus. The two plateau factors satisfy the exact weighted constraint \eqref{eq:triple-plateau-component-constraint}, but the present argument does not assert that their moduli are different from one in every admissible coefficient sector. The direction of the line depends only on the triple-root background parameter \(k\), while the plateau values may depend on the admissible internal coefficients \(M_1\) and \(M_2\). The present analysis establishes nonlocalization on this line but does not claim a complete classification of other possible triple-root asymptotic directions.
	
	The framework developed here separates three ingredients that are often combined in multiple-root Darboux calculations: the geometry of the fixed-\(\mu\) spatial branches, the projective null geometry of the confluent eigenspace, and the directional regularization of the real-spectrum Darboux kernel. This separation provides a systematic way to determine which confluent data yield regular nonreal-spectrum solutions, which data produce finite nontrivial real-spectrum limits, and which coefficient choices are inadmissible.
	
	Possible extensions include higher spatial branch multiplicities, higher-fold Darboux transformations built from several branch points, and analogous real-spectrum limits for other multicomponent integrable systems. It would also be of interest to develop a complete asymptotic classification beyond the selected double-root level-set components and the distinguished triple-root characteristic line considered here, and to investigate the dynamical excitation and stability of the resulting nonlocalized structures.
	
	\section*{CRediT authorship contribution statement}
	
	\textbf{Muchen Dong}: Idea, methodology, analysis, simulations, calculations, visualization, code, writing. \textbf{Lei Wang}: Idea, conceptualization, methodology, analysis, writing, supervision, project administration. 
	\section*{Declaration of competing interest}
	The authors declare that they have no known competing financial interests or personal relationships that could have appeared to influence the work reported in this paper.
	\section*{Data availability statements}
	Data availability is not applicable to this article as no new data were created or analyzed in this study.
	\section*{Acknowledgment}
	We express our sincere thanks to all the members of our discussion group for their valuable comments. The authors acknowledge financial support from the National Natural Science Foundation of China (No.~12375002).
	
	\bibliography{sn-bibliography}


\begin{thebibliography}{9}
\ifx \bisbn   \undefined \def \bisbn  #1{ISBN #1}\fi
\ifx \binits  \undefined \def \binits#1{#1}\fi
\ifx \bauthor  \undefined \def \bauthor#1{#1}\fi
\ifx \batitle  \undefined \def \batitle#1{#1}\fi
\ifx \bjtitle  \undefined \def \bjtitle#1{#1}\fi
\ifx \bvolume  \undefined \def \bvolume#1{\textbf{#1}}\fi
\ifx \byear  \undefined \def \byear#1{#1}\fi
\ifx \bissue  \undefined \def \bissue#1{#1}\fi
\ifx \bfpage  \undefined \def \bfpage#1{#1}\fi
\ifx \blpage  \undefined \def \blpage #1{#1}\fi
\ifx \burl  \undefined \def \burl#1{\textsf{#1}}\fi
\ifx \doiurl  \undefined \def \doiurl#1{\url{https://doi.org/#1}}\fi
\ifx \betal  \undefined \def \betal{\textit{et al.}}\fi
\ifx \binstitute  \undefined \def \binstitute#1{#1}\fi
\ifx \binstitutionaled  \undefined \def \binstitutionaled#1{#1}\fi
\ifx \bctitle  \undefined \def \bctitle#1{#1}\fi
\ifx \beditor  \undefined \def \beditor#1{#1}\fi
\ifx \bpublisher  \undefined \def \bpublisher#1{#1}\fi
\ifx \bbtitle  \undefined \def \bbtitle#1{#1}\fi
\ifx \bedition  \undefined \def \bedition#1{#1}\fi
\ifx \bseriesno  \undefined \def \bseriesno#1{#1}\fi
\ifx \blocation  \undefined \def \blocation#1{#1}\fi
\ifx \bsertitle  \undefined \def \bsertitle#1{#1}\fi
\ifx \bsnm \undefined \def \bsnm#1{#1}\fi
\ifx \bsuffix \undefined \def \bsuffix#1{#1}\fi
\ifx \bparticle \undefined \def \bparticle#1{#1}\fi
\ifx \barticle \undefined \def \barticle#1{#1}\fi
\bibcommenthead
\ifx \bconfdate \undefined \def \bconfdate #1{#1}\fi
\ifx \botherref \undefined \def \botherref #1{#1}\fi
\ifx \url \undefined \def \url#1{\textsf{#1}}\fi
\ifx \bchapter \undefined \def \bchapter#1{#1}\fi
\ifx \bbook \undefined \def \bbook#1{#1}\fi
\ifx \bcomment \undefined \def \bcomment#1{#1}\fi
\ifx \oauthor \undefined \def \oauthor#1{#1}\fi
\ifx \citeauthoryear \undefined \def \citeauthoryear#1{#1}\fi
\ifx \endbibitem  \undefined \def \endbibitem {}\fi
\ifx \bconflocation  \undefined \def \bconflocation#1{#1}\fi
\ifx \arxivurl  \undefined \def \arxivurl#1{\textsf{#1}}\fi
\csname PreBibitemsHook\endcsname

\bibitem[\protect\citeauthoryear{Fokas}{1995}]{Fokas1995}
\begin{barticle}
\bauthor{\bsnm{Fokas}, \binits{A.S.}}:
\batitle{On a class of physically important integrable equations}.
\bjtitle{Physica D: Nonlinear Phenomena}
\bvolume{87},
\bfpage{145}--\blpage{150}
(\byear{1995})
\doiurl{10.1016/0167-2789(95)00133-O}
\end{barticle}
\endbibitem

\bibitem[\protect\citeauthoryear{Lenells}{2009}]{Lenells2009}
\begin{barticle}
\bauthor{\bsnm{Lenells}, \binits{J.}}:
\batitle{Exactly solvable model for nonlinear pulse propagation in optical
  fibers}.
\bjtitle{Studies in Applied Mathematics}
\bvolume{123}(\bissue{2}),
\bfpage{215}--\blpage{232}
(\byear{2009})
\end{barticle}
\endbibitem

\bibitem[\protect\citeauthoryear{Lenells and Fokas}{2009}]{LenellsFokas2009}
\begin{barticle}
\bauthor{\bsnm{Lenells}, \binits{J.}},
\bauthor{\bsnm{Fokas}, \binits{A.S.}}:
\batitle{On a novel integrable generalization of the nonlinear {Schrödinger}
  equation}.
\bjtitle{Nonlinearity}
\bvolume{22}(\bissue{1}),
\bfpage{11}--\blpage{27}
(\byear{2009})
\end{barticle}
\endbibitem

\bibitem[\protect\citeauthoryear{Chen et~al.}{2018}]{ChenEtAl2018}
\begin{barticle}
\bauthor{\bsnm{Chen}, \binits{S.}},
\bauthor{\bsnm{Ye}, \binits{Y.}},
\bauthor{\bsnm{Soto-Crespo}, \binits{J.M.}},
\bauthor{\bsnm{Grelu}, \binits{P.}},
\bauthor{\bsnm{Baronio}, \binits{F.}}:
\batitle{Peregrine solitons beyond the threefold limit and their two-soliton
  interactions}.
\bjtitle{Physical Review Letters}
\bvolume{121}(\bissue{10}),
\bfpage{104101}
(\byear{2018})
\doiurl{10.1103/PhysRevLett.121.104101}
\end{barticle}
\endbibitem

\bibitem[\protect\citeauthoryear{Ye et~al.}{2019}]{YeEtAl2019}
\begin{barticle}
\bauthor{\bsnm{Ye}, \binits{Y.}},
\bauthor{\bsnm{Zhou}, \binits{Y.}},
\bauthor{\bsnm{Chen}, \binits{S.}},
\bauthor{\bsnm{Baronio}, \binits{F.}},
\bauthor{\bsnm{Grelu}, \binits{P.}}:
\batitle{General rogue wave solutions of the coupled {Fokas--Lenells} equations
  and non-recursive {Darboux} transformation}.
\bjtitle{Proceedings of the Royal Society A: Mathematical, Physical and
  Engineering Sciences}
\bvolume{475}(\bissue{2223}),
\bfpage{20180806}
(\byear{2019})
\doiurl{10.1098/rspa.2018.0806}
\end{barticle}
\endbibitem

\bibitem[\protect\citeauthoryear{Ling and Su}{2024}]{LingSu2024}
\begin{barticle}
\bauthor{\bsnm{Ling}, \binits{L.}},
\bauthor{\bsnm{Su}, \binits{H.}}:
\batitle{Rogue waves and their patterns for the coupled {Fokas--Lenells}
  equations}.
\bjtitle{Physica D: Nonlinear Phenomena}
\bvolume{461},
\bfpage{134111}
(\byear{2024})
\doiurl{10.1016/j.physd.2024.134111}
\end{barticle}
\endbibitem

\bibitem[\protect\citeauthoryear{Ling et~al.}{2018}]{LingFengZhu2018}
\begin{barticle}
\bauthor{\bsnm{Ling}, \binits{L.}},
\bauthor{\bsnm{Feng}, \binits{B.-F.}},
\bauthor{\bsnm{Zhu}, \binits{Z.}}:
\batitle{General soliton solutions to a coupled {Fokas--Lenells} equation}.
\bjtitle{Nonlinear Analysis: Real World Applications}
\bvolume{40},
\bfpage{185}--\blpage{214}
(\byear{2018})
\doiurl{10.1016/j.nonrwa.2017.08.013}
\end{barticle}
\endbibitem

\bibitem[\protect\citeauthoryear{Zhang et~al.}{2017}]{ZhangEtAl2017}
\begin{barticle}
\bauthor{\bsnm{Zhang}, \binits{Y.}},
\bauthor{\bsnm{Yang}, \binits{J.W.}},
\bauthor{\bsnm{Chow}, \binits{K.W.}},
\bauthor{\bsnm{Wu}, \binits{C.F.}}:
\batitle{Solitons, breathers and rogue waves for the coupled {Fokas--Lenells}
  system via {Darboux} transformation}.
\bjtitle{Nonlinear Analysis: Real World Applications}
\bvolume{33},
\bfpage{237}--\blpage{252}
(\byear{2017})
\doiurl{10.1016/j.nonrwa.2016.06.006}
\end{barticle}
\endbibitem

\bibitem[\protect\citeauthoryear{Yang and Zhang}{2018}]{YangZhang2018}
\begin{barticle}
\bauthor{\bsnm{Yang}, \binits{J.}},
\bauthor{\bsnm{Zhang}, \binits{Y.}}:
\batitle{Higher-order rogue wave solutions of a general coupled nonlinear
  {Fokas--Lenells} system}.
\bjtitle{Nonlinear Dynamics}
\bvolume{93},
\bfpage{585}--\blpage{597}
(\byear{2018})
\doiurl{10.1007/s11071-018-4211-4}
\end{barticle}
\endbibitem

\end{thebibliography}
\end{document}